\documentclass[10pt]{amsart}
\usepackage{macros}
\usepackage{comment}
\usepackage[T1]{fontenc}         
\usepackage{amsfonts}
\DeclareMathAlphabet{\mathmybb}{U}{bbold}{m}{n}
\newcommand*{\BOmega}{\mathmybb{\Omega}}

\usepackage{parskip}
\usepackage[flushmargin]{footmisc}

\def\opp{\mathrm{op}}

\DeclareMathOperator{\Map}{{\rm Map}}

\DeclareMathOperator{\cMap}{\MM{\rm ap}}
\DeclareMathOperator{\Fun}{{\rm Fun}}
\DeclareMathOperator{\Mfld}{\mathsf{Mfld}}
\DeclareMathOperator{\dgMfld}{\mathsf{dgMfld}}
\DeclareMathOperator{\DerMfld}{\mathcal{D}\mathsf{er}\mathcal{M}\mathsf{fld}}
\DeclareMathOperator{\CartSp}{\mathsf{CartSp}}
\DeclareMathOperator{\Riem}{\mathsf{Riem}}
\DeclareMathOperator{\Lrtz}{\mathsf{Lrtz}}
\DeclareMathOperator{\Open}{\mathsf{Open}}
\DeclareMathOperator{\Ab}{\mathsf{Ab}}
\DeclareMathOperator{\Ch}{\mathsf{Ch}}
\DeclareMathOperator{\Chz}{\mathsf{Ch}_\ZZ}
\DeclareMathOperator{\cDz}{\mathcal{D}_\ZZ}
\DeclareMathOperator{\Shv}{\mathcal{S}\mathsf{hv}}
\DeclareMathOperator{\PShv}{\mathcal{PS}\mathsf{hv}}
\DeclareMathOperator{\Ring}{\mathsf{Ring}}
\DeclareMathOperator{\sRing}{\mathcal{R}\mathsf{ing}}
\DeclareMathOperator{\Set}{\mathsf{Set}}
\DeclareMathOperator{\sSet}{\mathsf{sSet}}
\DeclareMathOperator{\sAb}{\mathsf{sAb}}
\DeclareMathOperator{\Spaces}{\mathcal{S}\mathsf{paces}}
\DeclareMathOperator{\dStk}{\mathcal{D}\mathsf{er}\mathcal{S}\mathsf{tacks}}
\DeclareMathOperator{\Top}{\mathsf{Top}}

\DeclareMathOperator{\sBun}{\mathbb{B}un}
\DeclareMathOperator{\cBun}{Bun}

\DeclareMathOperator{\Spec}{Spec}

\DeclareMathOperator{\GL}{GL}

\DeclareMathOperator{\Hom}{Hom}

\DeclareMathOperator{\vol}{vol}

\DeclareMathOperator{\Del}{{\cD}el}
\DeclareMathOperator{\Cone}{Cone}
\DeclareMathOperator{\cl}{cl}

\DeclareMathOperator{\Mxw}{Mxw}
\DeclareMathOperator{\sMxw}{\mathbb{M}xw}

\def\ort{\mathrm{or}}
\def\pert{\mathrm{pert}}
\def\tors{\mathrm{tors}}
\def\free{\mathrm{free}}

\def\ul{\underline}

\title{Duality of generalized Maxwell theories as an equivalence in derived geometry}

\author{Chris Elliott \and Owen Gwilliam \and Ingmar Saberi \and Brian R. Williams}

\begin{document}

\begin{abstract}
We propose a non-perturbative description of the moduli spaces encoding $p$-form generalized Maxwell
theories in any dimension, 
using derived differential geometry.
Our approach synthesizes the Batalin--Vilkovisky formalism with differential cohomology.
Within this framework we formulate Dirac charge quantization and show how such charge-quantized moduli spaces exhibit
abelian duality between generalized Maxwell theories of different types.
We also describe the compactification of generalized Maxwell theories along closed Riemannian manifolds by computing
the pushforward of the underlying sheaves of cochain complexes that model differential cohomology.
\end{abstract}

\maketitle

\setcounter{tocdepth}{1}
\tableofcontents

\section{Introduction}

Abelian duality plays a central role in contemporary field theory,
especially as a source of inspiration for much more subtle dualities among supersymmetric field theories and string theories.
Roughly speaking, it provides a way to relate pairs of abelian gauge theories. Examples include:
\begin{itemize}
\item the electromagnetic duality of Maxwell theory on four dimensional spacetimes, where the electric and magnetic fields are swapped, as is the coupling constant;
\item $T$-duality of the compact boson on two dimensional spacetimes, where the radius of the target circle is inverted (so $R \leftrightarrow 1/R$); and
\item the duality between a compact boson and electromagnetism on three dimensional spacetimes (the radius and the coupling constant are related).
\end{itemize}
The last example is particularly striking because it relates two theories of very different-seeming kind: 
a $\sigma$-model and a gauge theory.
In these dualities, there should be a transformation between operators in the two theories,
letting one match expected values under the path integral.
In this sense, the duality is analogous to a Fourier transform.\footnote{There is a vast literature on abelian duality. 
As useful starting points, we suggest Witten's article \cite{WittenAD}, the contemporaneous paper by Verlinde \cite{Verlinde}, Witten's IAS lecture \cite[Lecture 8]{WittenIAS}, and, especially for T-duality, the Big Red Book~\cite{HoriMS}.
Citations by and to these will connect with the bulk of work on this topic, but we mention some other classic references, like~\cite{DesTei}, along the way.}

Given the simplicity and brevity of the physical arguments,
and given the importance of this phenomenon,
it is frustrating that there is no mathematical treatment that seems to fully capture or explain the physical ideas.
This paper does not remove this frustration! 
It does offer, however, a perspective that may clarify the underlying claim of duality and its mechanism,
at least for those fond of homological algebra and sheaves,
and we hope it may be step toward a more complete, satisfactory mathematical formulation.

As a gloss of our main results,
\begin{itemize}
\item we give a careful, nonperturbative description of the moduli spaces that describe abelian $p$-form gauge theories;\footnote{We mean here generalizations of pure $U(1)$ Yang-Mills theory, where the equation of motion is $\d \star F_A = 0$.
We will use the term {\em Maxwell theories} as they generalize electromagnetism (with a gauge potential).}
synthesizing the BV formalism with differential cohomology,
\item we formulate Dirac charge quantization as picking out a subspace of this moduli space;
\item we exhibit abelian duality as a very simple map between these charge-quantized moduli spaces;\footnote{Note that we assert an equivalence of {\em classical} theories. Such an equivalence automatically guarantees an equivalence between quantizations,
although matching (conventional) parametrizations of quantizations  may be subtle.} 
and
\item we show that abelian $p$-form gauge theory is equivalent to abelian $n-p-2$-form gauge theory {\em coupled to a topological gauge theory}.\footnote{This last statement offers a perspective on $p$-form puzzles raised by Moore~\cite{CMSA}.}
\end{itemize}
A complete, sharp formulation of these results appears in Section~\ref{sec: main result}.
Our main tool is, in essence, sheaves of cochain complexes,
and so the primary maneuvers boil down to familiar tricks with short and long exact sequences and cohomology computations for integral and de Rham cohomology.
Finally, as a payoff, we show why compactifying these theories leads to the appearance of finite homotopy topological field theories (i.e., Dijkgraaf-Witten-type theories).
That is, we give a simple topological explanation for how torsion in cohomology contributes to physical systems.

We use the remainder of the introduction to sketch our point of view by way of an extended example,  
before comparing to prior work and mentioning some possible lines of future development.

\subsection{Abelian duality as an equivalence of derived stacks: an overview in three dimensions}

The essential thrust of this paper is to formulate the classical field theories using derived stacks,
but the key ideas can be seen just with explicit cochain complexes.
As an invitation, and to orient the reader, we describe everything here for the 3d case,
where abelian duality should be an equivalence between 3d electromagnetism and the 3d compact boson.\footnote{This example of duality seems to go back to~\cite{Polyakov}.}
(We will clarify how to upgrade these cochain complexes to derived stacks in Section 2.)

We describe first the gauge theory and then the sigma model.

\subsubsection{3d electromagnetism and cochain complexes}

We work on $\RR^3$ with a Riemannian metric, for simplicity.
There is, up to isomorphism, just one $U(1)$-bundle on $\RR^3$,
and we view it as the trivial bundle.
Given this trivialization, the space of gauge connections is isomorphic to $\Omega^1(\RR^3)$ by sending the 1-form $A$ to the connection~$\d + iA$.
The group of bundle automorphisms $\Map(\RR^3, U(1))$ acts on the space of connections, 
and under the isomorphism just given, we express the action as
\[
g \cdot iA = iA + g^{-1} \d g.
\]
If we write the gauge transformation as $g = e^{i f}$ with $f$ a real-valued function on~$\RR^3$,
then
\[
iA \mapsto iA + i\, \d f.
\]
As may be familiar from the BRST formalism (or differential cohomology), 
all this information can be succinctly encoded in a cochain complex
\[
\begin{tikzcd}[row sep=tiny]
\underline{-1} & \underline{0}\\
\Map(\RR^3, U(1)) \arrow[r, "\d \log"] & \Omega^1(\RR^3)
\end{tikzcd}
\]
where we put the 1-forms in degree zero and the gauge transformations (or ghosts) in degree~$-1$.
Note that the image of the differential describes all 1-forms (or connections) that are gauge-equivalent to the trivial connection.
More generally, each coset for this image describes the orbit of a connection under the gauge transformations.
Finally, the kernel of the differential consists of the global symmetries of the bundle.\footnote{A mathematician might note that we are describing an orbifold built from abelian Lie groups as a cochain complex,
since groupoid objects in abelian groups are equivalent to two-term cochain complexes.}

One can impose the equations of motion in terms of a cochain complex as well,
following the Batalin-Vilkovisky (BV) formalism.
Here we use 
\[
\begin{tikzcd}[row sep=tiny]
\underline{-1} & \underline{0} & \underline{1} & \underline{2}\\
\Map(\RR^3, U(1)) \arrow[r, "-i \d \log"] & \Omega^1(\RR^3)  \ar[r, "\d \star \d"] & \Omega^2(\RR^3)  \ar[r, "\d"] & \Omega^3(\RR^3) 
\end{tikzcd}
\]
where the degree~1 component consists of the antifields and the degree~2 component contains the antifields for the ghosts.\footnote{The reader familiar with BV formalism might note that we are deviating slightly from standard examples by letting the ghosts be the group of gauge transformations rather than the Lie algebra. This deviation is quite useful, as we hope to show in this paper.}
Observe that the kernel of the map $\d \star \d$ encodes the solutions to Maxwell's equation 
\[
\d \star F_A = 0 = \d \star \d A
\]
and hence the degree zero cohomology consists of solutions modulo gauge equivalence.
The positive cohomology groups vanish, as we show later in the text.

As a final observation, note that 
\[
\Map(\RR^3, U(1)) \cong C^\infty(\RR^3)/\ZZ
\]
because every gauge transformation $g$ can be expressed as $e^{i f}$ for some smooth function $f$ 
and the choice of $f$ is ambiguous up to a shift $f + 2\pi n$ determined by an integer~$n$.
Hence the complex of BV fields has the same cohomology as the complex
\begin{equation}
\label{eqn: 3d em}
\begin{tikzcd}
\ZZ \ar[r] & C^\infty(\RR^3) \arrow[r, "\d"] & \Omega^1(\RR^3)  \ar[r, "\d \star \d"] & \Omega^2(\RR^3)  \ar[r, "\d"] & \Omega^3(\RR^3) 
\end{tikzcd}
\end{equation}
where $\Omega^1(\RR^3)$ sits in degree zero so the whole complex sits between $-2$ and~$2$.
In fact, there is a natural cochain map from this complex  to the earlier complex of BV fields and it induces an isomorphism on cohomology.
(Such a map is called a {\em quasi-isomorphism}.)

Later in this paper we promote this complex to a sheaf valued in cochain complexes,
in the style of Deligne complexes.

\subsubsection{3d compact boson}

The compact boson has fields given by maps $\sigma$ into $U(1)$ with equation of motion~$\d \star \d \sigma = 0$.
In the BV formalism it is encoded by the cochain complex
\[
\begin{tikzcd}[row sep=tiny]
\underline{0} & \underline{1}\\
\Map(\RR^3, U(1)) \ar[r, "\d \star \d"] & \Omega^3(\RR^3) 
\end{tikzcd}
\]
There are no ghosts and no antifields for ghosts.
The complex
\begin{equation}
\label{eqn: 3d boson}
\begin{tikzcd}[row sep=tiny]
\ZZ \ar[r] & C^\infty(\RR^3) \ar[r, "\d \star \d"] & \Omega^3(\RR^3) 
\end{tikzcd}
\end{equation}
is quasi-isomorphic,
since any smooth map $\sigma: \RR^3 \to U(1)$ can be expressed as $e^{i f}$ up to an integer's worth of ambiguity.

\subsubsection{Duality as an isomorphism of cochain complexes}

It should be apparent that, as BV theories, the 3d compact boson and 3d electromagnetism are {\em not} equivalent.
The cochain complexes~\eqref{eqn: 3d boson} and~\eqref{eqn: 3d em} have different cohomology and hence are {\em not} quasi-isomorphic (much less isomorphic).
Duality cannot mean that these classical theories are the same, in this sense of equivalence of BV theories.

There are, however, closely related classical theories that \emph{are} equivalent.
As a first step, we replace our descriptions above by quasi-isomorphic cochain complexes.
For example, in the case of 3d electromagnetism,
consider the cochain complex
\begin{equation}
\label{eqn: 3d em alt}
\begin{tikzcd}[row sep = tiny]
\underline{-2} & \underline{-1} & \underline{0} & \underline{1} \\
\ZZ \ar[r, "1"] & \Omega^0 \ar[r, "\d"] & \Omega^1 \ar[dr, "\star \d"]& \\
&&& \Omega^{1} \\
& \RR \ar[r] & \Omega^0 \ar[ur, "-\kappa \d"']&
\end{tikzcd}.
\end{equation}
Much of this paper is devoted to explaining how this kind of complex is related to the complexes we wrote earlier,
so we postpone justifying this choice for now.
Instead, we will quickly analyze the complex to see what information it encodes.

We remark that we have introduced a constant $\kappa$ that should be seen as $1/g^2$, 
the typical coupling constant for a Yang-Mills theory.
This will make it easier to recognize the duality.

Observe that in degree zero, we are asking for a 1-form $A$ and a function $f$ such that 
\[
\star F_A = \star \d A = \kappa\, \d f.
\]
In this situation, because $\d^2 f = 0$, we find $\d \star F_A = 0$, which is Maxwell's equation.
Going the other way, note that if a 1-form satisfies Maxwell's equation $\d \star F_A = 0$, 
then $\star F_A$ must be exact by Poincar\'e's lemma,
so that there exists a function~$f$ (a {\em potential} for $\star F_A$) that is unique up to a constant.
At the level of cohomology, we identify such pairs $(A,f)$ if they differ by a gauge transformation and a shift by a real number, respectively,
so the zeroth cohomology is isomorphic to 1-forms satisfying Maxwell's equation.
In fact, as we show later, the complex~\eqref{eqn: 3d em alt} is quasi-isomorphic to the complex~\eqref{eqn: 3d em} of BV fields.

To obtain the theory pertinent to duality, 
we modify this complex in a simple way:
we replace the copy of $\RR$ by $\ZZ$.
Explicitly, we work with the complex
\begin{equation}
\label{eqn: 3d cq em}
\begin{tikzcd}[row sep = tiny]
\ZZ \ar[r, "1"] & \Omega^0 \ar[r, "\d"] & \Omega^1 \ar[dr, "\star \d"]& \\
&&& \Omega^{1} \\
& \ZZ \ar[r, "e"] & \Omega^0 \ar[ur, "-\kappa \d"']&
\end{tikzcd}.
\end{equation}
Note that this complex~\eqref{eqn: 3d cq em} is a subcomplex of the complex~\eqref{eqn: 3d em alt}:
it imposes a constraint by making the ``ghosts'' for $f$ into a copy of $\ZZ$ inside~$\RR$.
As we explain later, this maneuver is equivalent to imposing Dirac charge quantization;
we prefer to call it {\em charge discretization}.
(Briefly, we replace the real-valued function $f$ by a map $\sigma$ into the circle~$\RR/e\ZZ$.)\footnote{%
\label{footnote of sophisticate}A sophisticate might recognize that the complex~\eqref{eqn: 3d cq em} is the fiber product of a map from the complex~\eqref{eqn: 3d em alt} to a copy of $\RR/\ZZ[1]$, which models the stack $BU(1)$. In other words, the usual gauge theory has a natural line bundle, and the charge-discretized theory can be viewed as a trivialization of that line bundle.}

The compact boson admits a parallel modification to the complex
\begin{equation}
\label{eqn: 3d cq boson}
\begin{tikzcd}[row sep = tiny]
&\ZZ \ar[r, "1"] & \Omega^0  \ar[dr, "\star \d"]& \\
&&& \Omega^{2} \\
\ZZ \ar[r, "1/e"] & \Omega^0 \ar[r, "\d"] & \Omega^1 \ar[ur, "-(1/\kappa) \d"']&
\end{tikzcd}.
\end{equation}
Note that a zero cocycle is a pair $(f,A)$ such that $\star \d f = (1/\kappa) \d A$,
and that implies $\d \star \d f = 0$ as well as that $A$ is determined by $f$, up to an exact 1-form.
(That is, $A$ is a potential for $\star \d f$.)
This complex is also a subcomplex of an alternative version of the compact boson, parallel to~\eqref{eqn: 3d em alt}.

These two complexes \eqref{eqn: 3d cq boson} and~\eqref{eqn: 3d cq em} are very similar at a glance, 
if the top and bottom rows are swapped so that lengths match.
In fact, they are isomorphic:
the isomorphism consists of
\begin{itemize}
\item for the rightmost term, apply the Hodge star and scale by $1/\kappa$,
\item map the top row of~\eqref{eqn: 3d cq em} to the bottom row of~\eqref{eqn: 3d cq boson} but scale $\Omega^0$ by $1/e$, and
\item map the bottom row of~\eqref{eqn: 3d cq em} to the top row of~\eqref{eqn: 3d cq boson} but scale $\Omega^0$ and $\Omega^1$ by~$1/e$.
\end{itemize}
In other words, these complexes describe isomorphic theories, related by Hodge duality and interchange of which field is primary vs. potential.

More briefly, we find that these generalized Maxwell theories {\em with the constraint of charge discretization} are isomorphic as cochain complexes. We interpret these cochain complexes from the point of view of the Batalin-Vilkovisky , but it is important to note that this requires a generalization of the BV formalism: due to the charge discretization constraint, the moduli stack presented by such a complex only admits a shifted \emph{presymplectic} structure. (In physics terms, the corresponding theory is not Lagrangian.)
Having taken account of this, abelian duality becomes an isomorphism of classical field theories when formulated using the appropriate framework.\footnote{In fact, building on footnote~\ref{footnote of sophisticate}, one can view the charge-discretized theory as a correspondence between the generalized Maxwell theories.
This viewpoint fits nicely with the understanding of duality as adjoining extra fields and integrating out some fields,
which we discuss in Section~\ref{sec: correspondence}.}

\subsubsection{An overview of our main results}

The basic goal of this paper is to upgrade this observation in the following ways:
\begin{itemize}
\item We allow an arbitrary $n$-dimensional Riemannian (or Lorentzian) manifold, not just $\RR^3$, 
and we observe that any generalized Maxwell theory restricts along isometric embeddings (of codimension zero).
\item We allow higher form gauge theories: the connection can be a $p$-form for any $p \leq n-2$.
\item We show there is an isomorphism between the charge-discretized $p$-form theory and the charge-discretized $n-p-2$-form theory.\footnote{The curvatures are $p+1$- and $n-(p+1)$-forms and they should be Hodge dual.}
\item We show this isomorphism intertwines with restriction along isometric embeddings.
\end{itemize}
The key point here is to identify the natural functoriality, in spacetimes, of these theories,
and the mathematical tool is sheaf theory.

A second goal of this paper is to explain how these cochain complexes are avatars of a richer structure:
they can be enhanced to derived stacks.\footnote{Perturbative BV theories already possess a systematic description in the language of formal derived stacks.}
The isomorphism we provide is then an equivalence of sheaves, on spacetimes, valued in derived stacks.\footnote{In physics language, we have a global (nonperturbative) equivalence of classical theories.}

Finally, as a concrete payoff for the formulation we develop,
we compute compactifications carefully,
leading to interesting refinements of claims in the physics literature.

\subsection{On prior work and the novelty here}

There is already a tremendous amount of prior work on how sophisticated mathematical tools, like differential cohomology, illuminate gauge theory.
As we have slowly come to grips with higher abelian gauge theory,
we turned repeatedly to the papers of Freed, Moore, and Segal \cite{FMSAP, FMSCMP}, 
Brylinski's pioneering work \cite{BryBook},
and the {\em n}Lab, where we absorbed many insights 
--- presumably due mostly to Fiorenza, Sati, and Schreiber --- about how differential cohomology relates to higher and derived geometry, 
as well as~\cite{FreHopRR,BBSS}.
The recent TASI lectures of Moore and Saxena \cite{MooSax} are a masterful (and exceedingly thorough) treatment from a more physical perspective.
We highly recommend these sources as places to start exploring the literature.
The first part of this paper --- and many comments sprinkled throughout --- contain our take of those ideas and insights.

A key goal for us is to show how easy it is to synthesize these insights with the Batalin--Vilkovisky formalism for field theory,
which is often applied primarily to perturbative issues.
More precisely, we show how the framework of derived stacks enhances the illumination of gauge theory begun by introducing stacks.\footnote{For instance, Hopkins and Singer mention how it is fruitful to understand differential cohomology via groupoids,
and this stacky view is pushed further and exploited by Sati, Schreiber, and their collaborators.
We expect the recent work~\cite{BMNS} to allow a treatment of nonabelian gauge theories in a parallel fashion,
although the technical aspects are substantially higher.
For an overview of how BV theory should interface with derived geometry, see~\cite{AlfYou}.}

The second part of the paper develops an alternative description of abelian gauge theory first mentioned,
so far as we know in~\cite{SabWil}, 
and shows how Dirac charge quantization in this description is very clean.\footnote{This formulation is very close to comments in~\cite{FreHopRR}, albeit with the equations of motion imposed cohomologically.}
Duality becomes an obvious isomorphism of complexes.
At the end we undertake some straightforward sheaf computations that recover, in a clean way, physicists' claims about compactifications of higher gauge theories.

For us this paper is a preliminary step towards studying and quantizing theories like the self-dual 2-form gauge theory on 6-manifolds and related theories, 
such as the holomorphic twist of the 6d $\cN = (1,0)$ superconformal theory known as the abelian tensor multiplet,
in a style modeled on the work of Beilinson and Drinfeld for chiral CFT.
We wish to calibrate our approach on examples where the physics and mathematics are better understood and better documented.\footnote{Our approach is informed by~\cite{Elliott}, but takes a substantially different route.}

\subsection{Acknowledgements}

Over the years, we have benefited from conversations with many people on this topic,
including David Ayala, Ilka Brunner, Kevin Costello, Daniel Grady, Andres Klene-Sanchez, Isaac Pliskin, Surya Raghavendran, Alex Schenkel, Mayuko Yamashita, and Philsang Yoo.
During our work on this paper, Severin Bunk, Vivek Saxena, and Grisha Taroyan offered incredibly helpful feedback on drafts,
in addition to useful conversations on the math and physics around this topic.

The National Science Foundation supported O.G. through DMS Grant 2042052. 
In addition, his sabbatical in fall 2024 at the Max Planck Institute of Mathematics offered a wonderful place to explore and write about these ideas, and he acknowledges its support with gratitude.

This work is funded by the Deutsche Forschungsgemeinschaft (DFG, German Research Foundation) under Projektnummer 517493862 (Homologische Algebra der Supersymmetrie: Lokalit\"at, Unitarit\"at, Dualit\"at). I.A.S. also  thanks the Max Planck Institute of Mathematics for hospitality during a visit to O.G. during his sabbatical there, as well as the Free State of Bavaria for support while this work was being performed.

\section{Our mathematical setting for higher form gauge theory}
\label{sec: notandconv}

We use a large amount of mathematical machinery in this paper, and we introduce it here.
Readers keen to get to the meat of the paper can pass through the first subsection,
which reviews how abelian gauge theories in the BRST and BV formalisms are encoded by sheaves of cochain complexes.
The second subsection then puts those ideas into a contemporary framework for derived and higher geometry.
Finally, in the third subsection, we discuss our conventions and notations for sheaves, categories, and so on.

\subsection{Our working context, for physicists or the computationally minded, and a puzzle}

We will describe our classical field theories using sheaves of cochain complexes of abelian groups,
which may sound abstract but in practice will mean familiar manipulations with differential forms and integral cohomology.
In other words, if one is familiar with basic differential geometry and algebraic topology,
all our constructions below are also familiar.

For instance, abelian Chern-Simons theory on an oriented 3-manifold $M$ is described by the sheaf\footnote{There are subtleties here, and we give orienting remarks in Section~\ref{sec: notandconv}.} of cochain complexes
\begin{equation}
\label{eqn: CS example}
\cA \colon \quad \ZZ \to \Omega^0 \xto{\d} \Omega^1 \xto{\d} \Omega^2 \xto{\d} \Omega^3
\end{equation}
where $\Omega^1$ sits in degree~0 since the main type of field is a connection 1-form.\footnote{This sheaf $\cA$ is a version of smooth Deligne cochains and appears in work on differential cohomology.}
The term $\ZZ$ in degree~$-2$ denotes the locally constant sheaf.
It is included to encode the gauge {\em group}: 
note that for $U$ an open ball, 
\[
\Omega^0(U)/\ZZ \cong \Map_{C^\infty}(U,S^1)
\]
since $S^1 \cong \RR/\ZZ$,
so the degree~$-1$ cohomology of $\cA(U)$ is precisely the gauge group.
On an arbitrary open set $V \subset M$,
the derived global sections of $\ZZ$ on~$V$ is quasi-isomorphic to the singular cochains~$C^\bullet(V, \ZZ)$,
so the derived global sections encode subtle topological information about~$V$.

Working perturbatively, the tangent complex to the moduli stack of abelian flat connections (at any flat
connection) is
the (shifted) de Rham complex
\[
\Omega^0 \xto{\d} \Omega^1 \xto{\d} \Omega^2 \xto{\d} \Omega^3,
\]
concentrated in degrees $-1$ to $2$,
where again $\Omega^1$ sits in degree zero.
This cochain complex is familiar in the typical BV/BRST approach to perturbative Chern-Simons theory.
Note that by dropping the $\ZZ$,
we have now only the gauge {\em Lie algebra} $\Omega^0$,
so we see how the appearance of the copy of $\ZZ$ encodes (at least part of) the nonperturbative
information of abelian Chern--Simons theory.
Let's see that information more explicitly.

There are two explicit ways to compute the cohomology of this sheaf $\cA$ of fields on a 3-manifold~$M$:
\begin{enumerate}
\item pick a good cover of $M$ and work with the \v{C}ech complex (here a double complex); or
\item use the long exact sequence in cohomology
\[
\cdots \to H^{k+1}(M,\RR) \to H^k(M,\cA) \to H^{k+2}(M,\ZZ) \to H^{k+2}(M,\RR) \to \cdots
\]
arising from the short exact sequence
\[
0 \to \Omega^\bullet[1] \to \cA \to \ZZ[2] \to 0
\]
of sheaves.
\end{enumerate}
Consider, for example, what we see about $H^0(M,\cA)$ from approach~(2): 
it sits in an exact sequence
\[
0 \to H^1(M,\RR)/H^1_{\rm free}(M,\ZZ) \to H^0(M,\cA) \to H^{2}_{\rm tors}(M,\ZZ) \to 0
\]
where the subscripts indicate the free or torsion parts of the integral cohomology.

As a concrete example, let $M = \Sigma_g \times \RR$ where $\Sigma_g$ is a closed oriented 2-manifold.
Then we have shown that
\begin{equation}
\label{eqn: ab cs naive}
H^0(M,\cA) = H^0(\Sigma_g,\cA) \cong \RR^{2g}/\ZZ^{2g} 
\end{equation}
is the torus describing the flat $U(1)$-bundles on~$\Sigma_g$.

For any oriented 3-manifold $M$, 
this cohomology group $H^0(M,\cA)$ is supposed to encode the on-shell fields, i.e., solutions to the equations of motion,
which here should be a principal $U(1)$-bundle $L$ with a flat connection~$\nabla$.
Such a solution has an underlying principal bundle $L$ whose first Chern class $c_1(L) \in H^2(M,\ZZ)$ must be torsion, 
since its curvature $F_\nabla$ is flat and hence its real cohomology class is $[F_\nabla] = 0$ in $H^2(M,\RR)$.
On the other hand, we can modify $\nabla$ to another flat connection 
by adding a closed 1-form $A$ to get a new flat connection $\nabla + A$ on~$L$.
If we identify $U(1)$ with $\RR/\ZZ$,\footnote{It's more common to use $\RR/2\pi \ZZ$ but that means we have to adjust the sheaf $\ZZ \to \Omega^0$ by scaling by $2 \pi$. We will instead use the isomorphism $\theta \mapsto \exp(2\pi i \theta)$ from $\RR/\ZZ$ to $U(1)$ throughout this paper.} 
then $A$ and $A'$ determine the same new flat connection if their difference $A-A'$ is a 1-form with integral periods
(i.e., they produce the same holonomy, or parallel transport, on~$L$).

In general, we will describe our field theories in a similar style:
we give a sheaf of cochain complexes and 
describe how the cohomology encodes the usual physical content.

At this point, something might seem puzzling about this approach:
a kind of mismatch between the physical content and the mathematical context.
From a field theory perspective,
one expects the solutions to the equations of motion to form a space,
something like a manifold.
For instance, we want to view the torus of~\eqref{eqn: ab cs naive} as a manifold.\footnote{The group $H^{-1}(M,\cA)$ encodes the symmetry data of an orbifold structure on the ``space of solutions.''}
But mathematically, we only specified $\cA$ as a sheaf of cochain complexes,
so the cohomology group~\eqref{eqn: ab cs naive} is merely an abelian group without any topology or smooth structure.

How can we adjust our perspective on $\cA$ so as to incorporate that richer information, which seems essential to the physics?
Just working with $\cA$ as above is insufficient.\footnote{One option is to view these cochain complexes as valued in topological abelian groups (or something similar, like bornological abelian groups). 
We do not use that approach here,
although it can be fruitful. 
See \cite{BBSS}, and companion papers, for such a treatment, along with an approach to quantization.}

Our approach here is to view $\cA$ as presenting a sheaf valued not in cochain complexes, 
but in higher and derived geometry.
We describe how to do that in the next section (and more extensively in Section~\ref{sec: dag stuff}).
For many purposes, though, the reader can blackbox that setting,
simply working with the sheaf of cochain complexes 
until more subtle questions make it necessary to look inside the box.

\subsection{Highbrow commentary on this context}

For those who want more details about a precise mathematical context,
we make orienting remarks here.
In Section~\ref{sec: dag stuff} we offer a more detailed discussion.

Inspired by the BRST and BV formalisms, 
we view a classical field theory on a manifold $M$ as a sheaf $\SS$ on $M$ valued in derived stacks,
typically equipped with some kind of shifted symplectic or presymplectic or Poisson structure.
In particular, we have a functor
\[
\SS: \Open(M)^\opp \to \dStk.
\]
Given an open set $U$, 
the derived stack $\SS(U)$ should be viewed as the space of solutions to the equations of motion over the region~$U$: 
in other words, the ``on-shell fields'' over~$U$.

We emphasize that this perspective embodies, in a very precise sense, the physicists' idea that a theory is just specified by local data on $\RR^n$ (or an open ball therein) equipped with some geometric data, like a metric or spin structure.
The sheaf theory takes care of extending this description to a more general class of spacetime manifolds.

In this paper we will focus on a special type of theories:
those where $\SS$ takes values in {\em abelian group} derived stacks.
That means that two solutions can be ``multiplied'' or ``added.''
As examples, consider:
\begin{itemize}
\item The free massless scalar field. Since the equation of motion is linear, the space of solutions is a vector space, and hence an abelian group.
\item The nonlinear $\sigma$-model into $S^1$, equipped with a flat metric (since $S^1$ itself is an abelian group).
\item The abelian gauge theory, where a solution is a pair $(L \to M, \nabla_L)$ of a $U(1)$-bundle $L \to M$ and a connection $\nabla_L$ satisfying $\d \star F_\nabla = 0$,
which we call {\em Maxwell theory}.
The group operation here is tensor product of line bundles and addition of connections.
\end{itemize}
On the surface these examples might look rather different,
but they all share the feature that $\SS$ can be seen as a kind of $\sigma$-model:
$\SS$ amounts to mapping into an abelian group object in derived stacks,
so the abelian group structure is inherited from the target space.

Such theories often can be presented in a succinct way:
$\SS$ is presented by a presheaf $S$ on $M$ valued in cochain complexes of abelian groups,
so that concrete and familiar constructions suffice, 
rather than more involved methods from derived geometry.
For instance, we describe abelian $p$-form gauge theories using Deligne-type complexes below.

We now outline the justification for working with such sheaves of cochain complexes.

The essential idea is that a derived stack $\XX$ is a certain kind of functor
\[
\XX: \mathsf{T}^\opp \to \Spaces
\] 
from a category $\mathsf{T}$ of test objects (e.g., derived affine schemes or  dg manifolds) to $\Spaces$,
so that we understand $\dStk$ in terms of $\Fun(\mathsf{T}^\opp, \Spaces)$.
Because we are studying abelian theories, we want a functor that passes through $\Ab(\Spaces)$,
the abelian group objects in $\Spaces$.
Now $\Spaces$ is the $\infty$-categorical localization of $\sSet$ at weak homotopy equivalences (following Kan),
and abelian group objects $\Ab(\sSet)$ in $\sSet$ are the same as simplicial abelian groups $\sAb$.
Hence, we might {\em present} $\XX$ by giving a functor
\[
X: \mathsf{T}^\opp \to \sAb
\]
and then identifying $\XX$ as the composite functor
\[
\mathsf{T}^\opp \to \sAb \to \Ab(\Spaces)
\]
with the last step given by the localization functor $\sAb \to \Ab(\Spaces)$.
But there is an even more familiar way of producing such derived stacks.

Recall that any cochain complex $A^\bullet$ of abelian groups determines a simplicial abelian group in two steps:
first, take the smart truncation to a complex in nonpositive degrees
\[
\tau_{\leq 0} A^\bullet = (\cdots A^{-2} \to A^{-1} \to Z(A^0) \to 0 \cdots )
\]
and then apply the Dold--Kan correspondence.\footnote{We work throughout this paper with {\em co}\/chain complexes, rather than homologically with chain complexes.
This choice flips the grading in many standards. 
Hence, for us, the Dold--Kan correspondence is between simplicial abelian groups and nonpositively graded cochain complexes.}
Let's denote the output by $[A^\bullet]$ so there is a composite functor
\[
[-]: \Ch(\Ab) \to \sAb
\]
turning cochain complexes into simplicial abelian groups.

Given a sheaf,  
\[
A: \mathsf{T}^\opp \to \Ch(\Ab),
\]
postcomposing with this functor thus determines a derived stack 
\[
\AA: \mathsf{T}^\opp \xto{A} \Ch(\Ab) \xto{[-]} \sAb \to \Ab(\Spaces).
\]
Due to this functoriality, it is typically straightforward to track how manipulations with the sheaf $A$ corresponds to manipulations with the stack~$\AA$.

Putting this all together, we view a field theory as giving a functor
\[
\SS: \Open(M)^\opp \to \Fun(\mathsf{T}^\opp,\Spaces)
\]
or equivalently a functor
\[
\SS: \Open(M)^\opp \times \mathsf{T}^\opp \to \Spaces.
\]
In our setting, we assume the solutions always form a derived abelian group stack,
so the target is $\Ab(\dStk)$.
Thus we can {\em present} $\SS$ by a functor
\[
S: \Open(M)^\opp \times \mathsf{T}^\opp \to \Ch(\Ab),
\]
also known as a presheaf of cochain complexes (of abelian groups).

We discuss two examples briefly.
First, we describe differential forms on $M$ (of a fixed degree),
as these play a key role throughout the paper.
Variations on this construction can be used to build most of the examples we use.
Second, we briefly discuss the example of abelian Chern-Simons theory from the preceding section.

\begin{eg}
Given a smooth manifold~$M$, 
consider the functor
\[
\Omega^k_M: \Open(M)^\opp \times \mathsf{T}^\opp \to \Spaces
\]
where
\[
\Omega^k_M(U, \cM) = 
\{ \text{sections of the bundle } \pi^*( \Lambda^k T^*U) \to U \times \cM \}
\]
where $\cM$ denotes an object in $\mathsf{T}$ and
where $\pi: U \times \cM \to U$ is the projection map.
This functor encodes $k$-forms on $M$ in families over the parameter space~$\cM$.
This functor presents a functor $\BOmega^k_M$ on $\Open(M)^\opp$ valued in derived stacks.~\hfill$\Diamond$
\end{eg}

\begin{eg}
Abelian Chern-Simons theory is about the moduli of flat $U(1)$-bundles on an oriented 3-manifold~$M$.
Consider the functor $\cA$ of~\eqref{eqn: CS example} we introduced earlier (known as the smooth Deligne complex).
For each open set $U \subset M$ and each test space $T \in \mathsf{T}$,
the cochain complex $\cA(U, T)$ determines a simplicial set, and hence an $\infty$-groupoid,
which encodes $T$-families of flat $U(1)$-bundles on~$U$.~\hfill$\Diamond$
\end{eg}

The factor $\Open(M)^\opp$ encodes how the derived stack varies over the spacetime~$M$,
while the factor $\mathsf{T}^\opp$ encodes the structure of the derived stack $\SS(U)$ on a fixed open set~$U$.
For instance, for fixed $U$, taking a point $\ast$ in $\mathsf{T}$ just returns the raw collection $\SS(U,\ast)$ of solutions on~$U$ as an $\infty$-groupoid,
which has no ``geometry'' and has just a homotopy type.
Taking the real line $\RR \in \mathsf{T}$ lets one identify 1-dimensional families of solutions on~$U$.
If the thickened point $\Spec(\RR[\epsilon]/(\epsilon^2))$ lives in $\mathsf{T}$, 
it returns the tangent directions at each solution.

One can generalize from a fixed manifold $M$ by replacing the category $\Open(M)$ by a site of spacetime manifolds.
For instance, one can work over $\Riem_n$, the site of Riemannian $n$-manifolds, 
which is relevant for gauge theory in Euclidean signature.
In the Chern-Simons example, the  natural site is oriented 3-manifolds and orientation-preserving embeddings.

{\bf Note bene:} Below we will typically {\em focus on the spacetime dependence} and {\em suppress the dependence on the test objects} $\mathsf{T}$ in our notation and discussion.
We provide details about our key examples as moduli spaces, via their functors of points, in Section~\ref{sec: dag stuff}.

\begin{rmk}
The nature of the test site $\mathsf{T}$ affects what kind of information can be recovered.
For instance, including ``thick points'' (whose ring of functions are artinian algebras) in $\mathsf{T}$ lets one access ``perturbative'' information; 
in mathematical terms, it means one can do deformation theory.
For our abelian theories, it lets one recovers the perturbative BRST complexes that see the Lie algebra of gauge transformations.
Similarly,
including ``derived'' directions in $\mathsf{T}$,
such derived schemes whose ring of functions are dg commutative algebras in nonpositive degrees,
lets one recover information in the positive degree components of a cochain complex~$A^\bullet$ that presents a derived abelian group stack~$\AA$.
This feature is useful, for example, in witnessing the antifields.~\hfill$\Diamond$
\end{rmk}

\subsection{Notations and conventions}

\subsubsection{On categories}

We work with several closely related categories and $\infty$-categories (or, better $(1,1)$-categories and $(\infty,1)$-categories).
Note that a category (in the classic sense) is a special kind of $\infty$-category,
so we will typically use the term ``category'' to mean $\infty$-category, 
unless we feel it clarifying to distinguish.

We describe the most important here:
\begin{itemize}
\item $\Ab$ denotes the 1-category of abelian groups and group maps.
\item $\Chz$ (or, equivalently, $\Ch(\Ab)$ as in the previous discussion) denotes the 1-category of cochain complexes of abelian groups and cochain maps (i.e., levelwise group maps that intertwine the differentials).
\item $\cDz = \cD(\Chz)$ denotes the stable $\infty$-category obtained by localizing $\Chz$ (as an $\infty$-category) at quasi-isomorphisms. 
Its homotopy category\footnote{This category is $\Chz$, as a 1-category, localized at quasi-isomorphisms.} is traditionally called the {\em derived category} but we freely use this term to refer to the $\infty$-category.
\item $\sSet$ denotes the 1-category of simplicial sets and simplicial maps.
\item $\Spaces$ denotes the $\infty$-category obtained by localizing $\sSet$ (as an $\infty$-category) at weak equivalences.
\item $\Top$ denotes the 1-category of topological spaces and continuous maps. 
Its $\infty$-categorical localization at weak homotopy equivalences is also~$\Spaces$.
\end{itemize}
Thus there are functors (of $\infty$-categories)
\[
\Ch \to \cD(\Chz) = \cDz
\]
and
\[
\Top \to \Spaces
\]
and
\[
\sSet \to \Spaces
\]
that we will invoke repeatedly because
if we give a functor into $\Chz$ or $\Top$ or $\sSet$,
we immediately obtain a functor into their localizations.
In particular, if we give a presheaf (i.e., functor) into one of these categories using classical technology from homological algebra or algebraic topology,
then we get a higher categorical version for free.

For a category like $\sSet$, we write $\Ab(\sSet)$ for the category of abelian group objects in $\sSet$ and maps of such.\footnote{Hence $\Ab$ is shorthand for $\Ab(\Set)$, the abelian group objects in the category of sets.}
This category is isomorphic to $\sAb$, the category of simplicial abelian groups.
We also write $\Ab(\Spaces)$ for the $\infty$-category of abelian group objects in $\Spaces$,
which is, in some sense, a more complicated construction since the data of being an abelian group object involves homotopy coherences.

\begin{rmk}
\label{rmk: truncation preserves lims}
Note that $\Ab(\Spaces)$ is equivalent to the $\infty$-category $\cD^{\leq 0}_\ZZ$, 
which is that localization at quasi-isomorphisms of $\Ch^{\leq 0}_\ZZ$,
the subcategory of cochain complexes concentrated in nonpositive degrees
(i.e., the ``connective'' cochain complexes).
The truncation functor $\cD_\ZZ \to \cD^{\leq 0}_\ZZ$ is a right adjoint and hence preserves limits.~\hfill$\Diamond$
\end{rmk}

Similarly, for any ring $R$, there is an abelian category of left $R$-modules,
and $\Ch_R$ denotes the category of cochain complexes of $R$-modules.
We use $\cD_R = \cD(\Ch_R)$ to denote the stable $\infty$-category obtained by localizing at quasi-isomorphisms.

\subsubsection{On sheaves}

Presheaves and sheaves come in many flavors, depending on the sites (their domains) and on their values (their codomains).
In practice we will often describe a presheaf on a site $\cC$ as, say, a functor of 1-categories
\[
F: \cC^\opp \to \Chz
\]  
and then view it as presenting a functor of $\infty$-categories
\[
\cF: \cC^\opp \to \cDz
\]
via the composite
\[
\cC^\opp \xto{F} \Chz \to \cDz.
\]
When we want to distinguish these functors,
we will refer to $F$ as a {\it $\Chz$-valued presheaf} and its companion $\cF$ as a {\it $\cDz$-valued presheaf}.

The situation with sheaves is much subtler,
because sheaves satisfy local-to-global conditions determined by the Grothendieck topology on the site.
These conditions involve limits in the target category, 
so a sheaf valued in $\Chz$ does not automatically turn into a sheaf valued in~$\cDz$.

On the other hand, a presheaf always determines a sheaf by sheafification
(the adjoint functor to the forgetful functor from sheaves to presheaves).
Hence we will often present a sheaf $\cF$ valued in $\cDz$ by giving a presheaf $F$ valued in $\Chz$, then viewing it as a presheaf valued in $\cDz$, and finally working with its sheafification.

To be precise, we use $\Shv(S, \cDz)$ to denote the $\infty$-category of $\cDz$-valued sheaves on a site~$S$.
By this, we mean the $\infty$-categorical localization of functors (i.e., presheaves)
\[
\PShv(S, \cDz) = \Fun(S^\opp, \cDz) \to \Shv(S, \cDz)
\]
at stalkwise quasi-isomorphisms.
All the sites we use have enough points and hence this notion is well-behaved.
(In other words, we require descent with respect to hypercovers.)

\begin{rmk}
\label{rmk: truncation on D}
Because truncation $\cD_\ZZ \to \cD^{\leq 0}_\ZZ \simeq \Ab(\Spaces)$ preserves limits,
we see that a sheaf valued in $\cDz$ truncates to a sheaf valued in abelian $\infty$-groupoids.~\hfill$\Diamond$
\end{rmk}

When we speak of ``global sections'' $\cF(M)$ of a sheaf $\cF$ on $M$ valued in~$\cDz$,
we mean its value in that $\infty$-category.
In practice, the cohomology of $\cF(M)$ is the hypercohomology of any sheaf $F$ valued in $\Chz$ on~$M$ that presents $\cF$.
Note that if $F$ is simply a sheaf valued in abelian groups~$\Ab$,
then the global sections of $\cF$ is the derived functor of global sections of $F$,
also known as the (traditional) sheaf cohomology of~$F$.
Hence we may use the phrase ``the cohomology of $\cF$'' to mean these (derived) global sections of a sheaf $\cF$ valued in~$\cDz$.

Note that we can thus present such a sheaf~$\cF$ by describing a sheaf $F$ valued in $\Chz$ {\em locally},
i.e., giving its values on coordinate balls.
We will also freely use, for example, short exact sequences that are visible locally.

\subsubsection{On higher bundles}

We have chosen to talk about higher bundles rather than higher gerbes.
The $U(1)$ $p$-bundles on a manifold $M$ are classified, up to isomorphism, by the connected components $\pi_0 \Map(M, B^pU(1))$ of the mapping space in the classifying space $B^p U(1)$. 
Note that a 1-bundle is precisely a principal $U(1)$-bundle and a 2-bundle is precisely a $U(1)$-gerbe.

We will typically denote a derived stack with blackboard bold. 
For instance --- and anticipating moduli spaces that appear late --- we write $\sBun_1(M)$ for the moduli of principal $U(1)$-bundles on the $n$-manifold~$M$ and $\sBun_p(M)$ for the moduli of principal $U(1)$ $p$-bundles on the $n$-manifold~$M$,
where a $p$-bundle means a principal bundle for the ``higher'' group $\BB^p U(1)$,
the group object in derived smooth geometry whose underlying homotopy type is $B^p U(1)$.
Similarly, $\sMxw_{p,n}(M,g)$
denotes the moduli of solutions to the $p$-form Maxwell gauge theory on~$M$ equipped with a metric~$g$
(the curvature $F_A$ of the connection must satisfy $\d \star F_A = 0$).

It is important to bear in mind 
that many different cochain complexes present the same derived stack.

We will typically denote a sheaf (of cochain complexes) with Roman letters, so $\Mxw_{p,n}$ (as defined in Definition~\ref{dfn: p form Maxwell theory} below) denotes a particular sheaf whose derived global sections $\Mxw_{p,n}(M)$ presents the derived stack $\sMxw_{p,n}(M)$.
But $\widetilde{\Mxw}_{p,n}$ (Definition~\ref{dfn: Maxwell-tilde theory}) is another such sheaf that also presents $\sMxw_{p,n}(M)$.

\subsubsection{On sites of spacetimes}

Any manifold $M$ determines a category $\Open(M)$ given by the poset of opens.
The usual topology of $M$ determines a Grothendieck topology on $\Open(M)$ and hence makes it into a site.

We will have use for other sites, that encompass large classes of spacetimes that share a common feature. 
Here are four that appear naturally in studying Maxwell theories, 
but there are many variants to incorporate other topological and geometric structures (e.g., spin structures or auxiliary bundles).

Let $\Riem_n$ denote the site whose objects are $n$-dimensional manifolds equipped with a Riemannian metric,
whose morphisms are isometric embeddings,
and where a cover $\{f_i: U_i\to M\}_{i \in I}$ of $M$ means that the union of the images is~$M$.

Let $\Riem^\ort_n$ denote the site where, in addition, the manifold is oriented and the morphisms are orientation-preserving.
Equivalently, each object $(M,g)$ has a distinguished, nowhere-vanishing top form $\vol_g$ that encodes volume.
This data determines a Hodge star operator $\star$ on the de Rham forms on~$M$.\footnote{One can define a Hodge star on unorientable manifolds by twisting the de Rham complex with the orientation bundle, 
and the reader is welcomed to generalize discussions below to incorporate that possibility.}

Let $\Lrtz_n$ denote the site whose objects are $n$-dimensional manifolds equipped with a Lorentzian metric (i.e., a pseudometric of signature $(-,+,\ldots,+)$),
whose morphisms are isometric embeddings,
and where a cover $\{f_i: U_i\to M\}_{i \in I}$ of $M$ means that the union of the images is~$M$.

Let $\Lrtz^\ort_n$ denote the site where, in addition, the manifold is oriented and the morphisms are orientation-preserving.
Equivalently, each object $(M,g)$ has a distinguished, nowhere-vanishing top form $\vol_g$ that encodes volume.
This data determines a Hodge star operator $\star$ on the de Rham forms on~$M$.

\section{Abelian $p$-form gauge theory}

In this section we review differential cohomology in the guise of Deligne complexes and then cast higher gauge theories in this framework.

\subsection{Smooth Deligne cohomology}
\label{sec: deligne stuff}

There is a sheaf of cochain complexes that presents abelian $p$-form gauge theory (the Maxwell, or Yang--Mills, version).
We construct it in stages.

First, as a cousin to the de Rham complex,
we introduce the smooth Deligne complex.
It resolves the sheaf of locally constant maps into the circle group~$\RR/\ZZ$,
just as the de Rham complex resolves the sheaf of locally constant real-valued functions (or maps into~$\RR$).

\begin{dfn}
\label{Deligne cplx}
On a smooth $n$-dimensional manifold $M$, the smooth {\em Deligne complex} $\Del_M$
is the sheaf valued in $\Chz$ given by the mapping cone of the cochain map
\[
\ZZ_M[1] \to \Omega^\bullet_M
\] 
sending an integer to a constant 0-form on $M$.
(Here $\Omega^\bullet_M$ denotes the de Rham complex as a sheaf and $\ZZ_M$ is the constant sheaf with stalk $\ZZ$.)
More explicitly $\Del_M$ is the complex
\[
\begin{tikzcd}
\ul{-1} & \ul{0} & \cdots & \ul{n} \\
\ZZ_M \ar[r] & \Omega^0_M \ar[r,"\d"] & \cdots \ar[r,"\d"] & \Omega^n_M .
\end{tikzcd}
\]
\end{dfn}

There are natural truncations of $\Del_M$ with rich meaning in themselves.

\begin{dfn}
Let $\Del^{\leq p}_M$ denote the truncated complex of sheaves
\[
\begin{tikzcd}
\ul{-1} & \ul{0} & \cdots & \ul{p} \\
\ZZ_M \ar[r] & \Omega^0_M \ar[r,"\d"] & \cdots \ar[r,"\d"] & \Omega^p_M .
\end{tikzcd}
\]
\end{dfn}

For example, 
\[
\Del^{\leq 0}_M = (\ZZ[1] \to C^\infty_M)
\]
resolves the sheaf $\Map(-,S^1)$ of smooth maps into the circle 
\[
S^1 = \RR/\ZZ \cong U(1)
\]
because, locally on $M$, any map into $S^1$ lifts to a map into $\RR$ but there are an integer's worth of lifts.

The cohomology group\footnote{Here we mean, as mentioned in Section~\ref{sec: notandconv}, the derived global sections, or hypercohomology.} $H^0(M, \Del^{\leq p}_M[p])$ is also known as the differential cohomology group $\check{H}^{p+1}(M)$,
and hence has received substantial attention from physicists and mathematicians,
but the other cohomology groups contain important information.
We view the whole complex $\Del^{\leq p}_M[p]$ as modeling the higher stack of $p$-bundles with connection $\sBun_{p}^\nabla(M)$
(or, if one prefers, the higher stack of $U(1)$ $p-1$-gerbes with connection).
As such, it models the BRST field content of $p$-form gauge theory, as we will see below.
Locally, this shifted complex is
\[
\ZZ_M[p+1] \to \Omega^0_M[p] \xto{\d} \cdots \xto{\d} \Omega^p_M
\]
with $\Omega^p_M$ in degree~0,
so that, locally on $M$, a degree 0 element $A$ is a $p$-form that encodes a ``connection $p$-form''
while a degree~$-1$ element encodes a gauge transformation, and so on.
(See \cite{BryBook, Gaj, FMSAP, FMSCMP, Bunk} for exposition.)
See Section~\ref{sec: coh computations} for computations of the cohomology on manifolds with nontrivial topology,
which exhibit more subtle features.

\begin{dfn}
\label{dfn: bun nabla}
Let $\cBun^\nabla_{p}$ denote the sheaf valued in $\cDz$ on the site of smooth manifolds whose local value is
\[
\ZZ \to \Omega^0 \xto{\d} \cdots \xto{\d} \Omega^p
\]
concentrated in degrees $-p-1$ up to $0$, 
so $\Omega^p$ sits in degree zero.
\end{dfn}

In other words, restricted to a fixed manifold~$M$, this sheaf $\cBun^\nabla_{p}\big|_M$ is presented by the truncated (and shifted) Deligne complex $\Del^{\leq p}_M[p]$ just described.

The derived stack presented by $\cBun^\nabla_{p}(M)$ is~$\sBun^\nabla_{p}(M)$. 
Note that there are other cochain models presenting $\sBun^\nabla_{p}(M)$,
such as the differential characters of Cheeger and Simons.

There is another important type of truncation.

\begin{dfn}
\label{dfn: closed forms}
For $p < n$ let $\Omega^{p+1}_{\cl}$ denote the sheaf valued in $\Chz$ on the site of smooth manifolds
\[
\Omega^{p+1} \xto{\d} \cdots \xto{\d} \Omega^n \cdots 
\]
sitting in degrees $0$ to $n-p-1$.
We call it the sheaf of {\em closed $p+1$-forms}.
\end{dfn}

On a manifold $M$ of dimension~$n$,
the restriction $\Omega^{p+1}_{\cl}\big|_M$ is concentrated in degrees $0$ to~$n-p-1$.

Note that $\Omega^{p+1}_{\cl}$ resolves the sheafification of the presheaf of closed $p$+$1$-forms.
Thus we view it as presenting the appropriate derived stack of {\em closed $p+1$-forms}.

On any manifold $M$ there is a short exact sequence 
\[
0 \to \Omega^{p+1}_{\cl}[-(p+1)] \to \Del_M \to \Del^{\leq p}_M \to 0
\]
of sheaves valued in~$\Chz$.
It says that the cohomology with locally constant coefficients in $U(1)$ amounts to studying closed $p$-forms and choices of trivializations by $p$-forms, up to homotopy.

If we shift this sequence
\[
0 \to \Omega^{p+1}_{\cl}[-1] \to \Del_M[p] \to \cBun^\nabla_{p}\big|_M \to 0,
\]
it presents an interesting map 
\[
\sBun^\flat_{p}(M) \to \sBun^\nabla_{p}(M)
\]
of moduli spaces:
\begin{itemize}
\item the middle term in the sequence presents the moduli space $\sBun^\flat_{p}(M)$ of $U(1)$ $p$-bundles with {\em flat} connection on~$M$;
\item the right map is the forgetful map to the moduli space of $U(1)$ $p$-bundles with connection; and
\item the fiber of that map is the data of a trivialization of the curvature of the connection.
\end{itemize}
The connecting map in the long exact sequence on cohomology is given by the curvature of a connection (not necessarily flat) on a $p$-bundle.
That is, it presents a map
\[
\sBun^\nabla_{p}(M) \xto{\text{curv}} \CC {\rm lsdFrm}^{p+1}(M)
\]
of derived stacks.
(We use, just this once, the notation $\CC {\rm lsdFrm}^{p+1}(M)$ for the stack of closed $p+1$-forms on~$M$.)

\subsection{The connection with physics}

We have nearly all the ingredients to define higher form gauge theory,
except the Hodge star $\star$.
We now require the manifold $M$ to be oriented and equipped with a pseudo-metric,
i.e., either an oriented Riemannian manifold or an oriented Lorentzian manifold.

As we will see below, the equation of motion for $p$-form Maxwell gauge theory\footnote{For an early discussion see~\cite{HenTei} and see~\cite{MooSax} for a recent exposition.} is
\[
\d \star F_A = 0
\]
where $A$ is a connection $p$-form and $F_A = \d A$ is its curvature.
Note that $\d \star F_A$ is an $(n-p)$-form for any $p$-form~$A$,
so we can view $\d \star \d$ as a cochain map from $\cBun^\nabla_{p}$ to $\Omega^{n-p}_{\cl}$.
The derived zero locus then models the derived space of solutions of the equation of motion.
It is conventional to include a coupling constant $\kappa \in \RR^\times = \RR -\{0\}$.

\begin{dfn}
\label{dfn: p form Maxwell theory}
For each nonzero coupling constant $\kappa \in \RR^\times$,
let $\Mxw_{p,n}\langle\kappa\rangle$ denote the sheaf on $\Riem^\ort_n$ (or $\Lrtz^\ort_n$) valued in $\cDz$ given locally by 
\[
 \ZZ \to \Omega^0 \xto{\d} \cdots \xto{\d} \Omega^p \xto{\kappa \d \star \d} \Omega^{n-p} \xto{\d} \cdots
\xto{\d} \Omega^n 
\]
where $\Omega^p$ sits in degree zero.\footnote{We use $\langle \kappa \rangle$ to record the coupling constant rather than $(\kappa)$ or $[\kappa]$ to avoid clashes of notation elsewhere.}
\end{dfn}

In other words it is the mapping cone of the map~$\kappa \d \star \d$:  
\[
\cBun^\nabla_{p} \xto{\kappa \d \star \d} \Omega^{n-p}_{\cl}
\]
and hence imposes the equation of motion at the level of cohomology.\footnote{Note that we use {\em closed} $n-p$-forms as the codomain of the map, rather than simply $\Omega^{n-p}$ without the closure condition. 
At one level, this condition takes into account the Bianchi identity, which is convenient.
As a deeper reason, the equation of motion arises from the differential $\delta S$ of the action functional,
and this 1-form $\delta S$ should be a section of the cotangent bundle of the space of fields, which here is $\cBun_p^\nabla$.
In this case, the tangent space to a field in $\cBun_p^\nabla$ consists of $\Omega^{\leq p}$,
and its dual, under Poincar\'e duality on an oriented $n$-manifold, is $\Omega^{\geq n-p}$.
These are precisely the {\em anti}\/fields of the BV formalism.
(This is a rough sketch of the reasoning behind the BV formalism; 
we suppress here further motivations and subtleties.)}

Note that for $M = \RR^n$, there is cohomology only in degree zero and in degree~$-p$. 
The degree zero cohomology consists of $p$-forms $A$ such that $\d \star F_A = 0$, modulo (infinitesimal) gauge equivalence.
The degree~$-p$ cohomology is~$\RR/\ZZ \cong U(1)$, 
which encodes the (higher) global symmetries of a $p$-bundle.
The degree one cohomology also vanishes locally.\footnote{\label{h1 vanishes }The argument here is a little longer, 
so we relegate it to this footnote. 
Let $\alpha \in \Omega^{n-p}(U)$ be annihilated by $\d$ and hence closed.
Since we are on $\RR^n$, it is also exact: $\alpha = \d \beta$ for some $\beta \in \Omega^{n-p-1}(\RR^n)$.  
By the theory of pseudodifferential operators,
the Hodge Laplacian $\Delta$ on $n-p-1$-forms is surjective
 (ellipticity implies local solvability). 
Thus, for some $\gamma$, 
\[
\beta = \Delta \gamma = \d \d^*\gamma + \d^*\d\gamma
\]  
and therefore $\alpha = \d\star\d(\pm \star \d\gamma)$, so that $\alpha \in \mathrm{im}(\d \star \d)$
as desired.}

This sheaf describes solutions to an abelian $p$-form gauge theory
in the BV formalism, 
as we now explain.

\begin{rmk}
\label{rmk: indexing for duality}
We choose to index the theory $\Mxw_{p,n}$ by the degree of the differential form where the {\em connection} (or gauge potential) lives,
but a natural alternative is to label by the degree of the differential form where the {\em curvature} (or field strength) lives.
This second choice has the virtue that it makes duality easier to remember:
duality sends a field strength $F$ of degree $k$ to a field strength $\star F$ of degree~$n-k$.
For us, duality will relate $\Mxw_{p,n}$ to $\Mxw_{n-p-2,n}$,
a more awkward indexing.~\hfill$\Diamond$
\end{rmk}

Recall that the conventional approach to $p$-form gauge theory~\cite{HenTei} starts with a $p$-form $A$ as the field and the action functional
\begin{equation}
\label{action}
S(A) = \frac{1}{g^2}\int_M \d A \wedge \star \d A
\end{equation}
where $g > 0$ denotes a coupling constant as conventionally parameterized in the physics literature. (Above, we used $\kappa = 1/g^2$ for the coupling constant, and will persist in this throughout for simplicity.) 
The equation of motion is 
\[
\kappa \d \star \d A = 0,
\]
so that the coupling constant does not affect the space of solutions.
In parallel with ordinary, 1-form gauge theory,
one introduces gauge symmetries by $p-1$-forms,
gauge symmetries for those symmetries, and so on.

In the perturbative BRST formalism designed to handle such gauge symmetries, 
people work with the complex 
\[
\Omega^0_M \xto{\d} \cdots \xto{\d} \Omega^{p-1}_M \xto{\d} \Omega^p_M
\]
where $\Omega^p$ is in degree zero.
This complex includes the fields, ghosts, ghosts-for-ghosts, etc. in the perturbative theory.
Note that the $p$-forms sit in degree~0 as the fields of obvious physical interest,
the ghost fields in degree~$-1$ are the $p-1$-forms,
the ghosts-for-ghosts in degree~$-2$ are $p-2$-forms, 
and so on.
It describes the tangent complex $\TT_A \sBun_p^\nabla$ to the moduli stack of $p$-forms modulo (abelian) gauge equivalence at any point~$A$ in the moduli space~$\sBun_p^\nabla$.

In the BV formalism, one introduces conjugate or dual fields, called antifields and antighosts, 
so that the complex for the BV classical theory is
\begin{equation}
\label{cplx: pert mxw}
\Omega^0_M \xto{\d} \cdots \xto{\d} \Omega^p_M \xto{\kappa \d \star \d} \Omega^{n-p}_M \xto{\d} \cdots \xto{\d} \Omega^n_M 
\end{equation}
where $\kappa = 1/g^2$ is the coupling constant and where $\Omega^p$ sits in degree zero.
The action functional for this BV theory is
\begin{equation}
\label{eqn: action for mxw}
S_{BV}(\cA) = \int_M \cA \wedge Q(\cA) = \sum_{k=-p}^{p+1} \int_M A_k \wedge Q(A_{-k})
\end{equation}
where an element of the complex has the form 
\[
\cA = (A_{-p}, \ldots, A_0, \ldots, A_{p+1})
\]
with $A_k$ in degree $k$ and with $Q$ the differential.
Note that for a field concentrated in degree 0 of the form
\[
\cA = (0,\ldots, 0 , A_0 = A,0, \ldots, 0)
\]
we recover the action~\eqref{action}.

We denote this complex~\eqref{cplx: pert mxw} by 
\[
\Mxw_{p,n}^\pert\langle\kappa\rangle.
\]
Note that this complex is similar to the complex from Definition~\ref{dfn: p form Maxwell theory},
except that it lacks the constant sheaf $\ZZ_M$ in degree~$-p-1$.
In more geometric terms,
the usual BRST/BV description $\Mxw_{p,n}^\pert\langle\kappa\rangle$ only sees the Lie {\em algebras} of (higher) gauge symmetries,
while $\Mxw_{p,n}\langle\kappa\rangle$ sees the {\em group}. 
The usual BRST/BV description cannot tell the difference between the Lie groups $\RR$ and~$U(1)$.
Remarkably, if we add the constant sheaf, following Deligne,
we recover all the global information about $U(1)$ $p$-bundles (or higher gerbes).

This construction manifestly makes sense on any smooth oriented Riemannian $n$-manifold and hence lives over the site~$\Riem^\ort_n$.
We view it as 
defining classical abelian $p$-form gauge theory on~$M$, as the following definition captures.

\begin{dfn}
For any Riemannian manifold $M \in \Riem^\ort_n$,\footnote{One can also work with Lorentzian manifolds using the same kind of complexes, with the Hodge star operator determined by the Lorentzian metric.}  the {\em moduli of solutions on $M$ to $p$-form Maxwell gauge theory}~$\sMxw_{p,n}\langle\kappa\rangle(M)$ with coupling constant~$\kappa \in \RR^\times$ is the derived abelian group stack presented by
the derived global sections $\Mxw_{p,n}\langle\kappa\rangle(M)$ via the Dold--Kan correspondence.
\end{dfn}

We now give several explicit examples, showing that this cochain complex locally recovers the usual equations of motion and gauge symmetries.

\begin{eg}[Electromagnetism in four dimensions]
For $M = \RR^4$, the cochain complex is
\[
\ZZ \to \Omega^0(\RR^4) \xto{\d} \Omega^1(\RR^4) \xto{\kappa \d \star \d} \Omega^{3}(\RR^4) \xto{\d} \Omega^4(\RR^4)
\]
and so the zeroth cohomology is
\[
\{ A \in \Omega^1(\RR^4) \, | \, \d \star \d A = 0\}/\{\d f \, |\, f \in C^\infty(\RR^4)\},
\]
namely solutions to Maxwell's equations with $F = \d A$ modulo gauge transformations.
There is cohomology in degree $-1$, namely $\RR/\ZZ \cong U(1)$, 
capturing the global symmetries by constant gauge transformations.
(Recall Footnote~\ref{h1 vanishes } for the vanishing of degree one cohomology.)~\hfill$\Diamond$
\end{eg}

\begin{eg}[Compact boson in two dimensions]
For $M = \RR^2$, the cochain complex is
\[
\ZZ \to \Omega^0(\RR^2) \xto{\kappa \d \star \d} \Omega^2(\RR^2)
\]
and so the zeroth cohomology is
\[
\{ \sigma: \RR^2 \to S^1 \, | \, \d \star \d \sigma = \star \Delta \sigma = 0\},
\]
the harmonic maps into $S^1$.
There is trivial cohomology in degree one because the Laplacian is surjective.
Note that this complex is isomorphic to
\[
\ZZ \xto{2\pi r} \Omega^0(\RR^2) \xto{\d \star \d} \Omega^2(\RR^2)
\]
if $\kappa = 2 \pi r$.
This new complex encodes maps into $\RR/2 \pi r\ZZ$, which is a flat circle with radius~$r$.
~\hfill$\Diamond$
\end{eg}

\begin{eg}[Electromagnetism in three dimensions]
For $M = \RR^3$, the cochain complex is
\[
\ZZ \to \Omega^0(\RR^3) \xto{\d} \Omega^1(\RR^3) \xto{\kappa \d \star \d} \Omega^{2}(\RR^3) \xto{\d} \Omega^3(\RR^3)
\]
and so the zeroth cohomology is
\[
\{ A \in \Omega^1(\RR^3) \, | \, \d \star \d A = 0\}/\{\d f \, |\, f \in C^\infty(\RR^3)\},
\]
namely solutions to Maxwell's equations with $F = \d A$ modulo gauge transformations.
There is cohomology in degree $-1$, namely $\RR/\ZZ \cong U(1)$, 
capturing the global symmetries by constant gauge transformations.~\hfill$\Diamond$
\end{eg}

\begin{eg}[Compact boson in three dimensions]
For $M = \RR^3$, the cochain complex is
\[
\ZZ \to \Omega^0(\RR^3) \xto{\kappa \d \star \d} \Omega^{3}(\RR^3)
\]
and so the zeroth cohomology is
\[
\{ \sigma: \RR^3 \to S^1 \, | \, \d \star \d \sigma = \Delta \sigma = 0\},
\]
the harmonic maps into~$S^1$. 
There is trivial cohomology in degree one because the Laplacian is surjective.~\hfill$\Diamond$
\end{eg}

\subsection{A shifted symplectic structure}

At any point $\cA$ in this stack, corresponding to a $p$-bundle with connection solving the gauge equation,
the tangent complex $\TT_\cA \Mxw_{p,n}\langle\kappa\rangle$ is given by the sheaf
\[
\Omega^0_M \xto{\d} \cdots \xto{\d} \Omega^p_M \xto{\d \star \d} \Omega^{n-p}_M \xto{\d} \cdots \xto{\d} \Omega^n_M
\]
with $\Omega^p_M$ in degree~0. 
Note that the major change here is that we only see infinitesimal gauge transformations,
encoded by the de Rham forms from degrees 0 up to~$p-1$;
the discrete sheaf $\ZZ$ is not there.
This tangent complex is precisely the usual BV formulation of abelian $p$-form gauge theory.

\begin{rmk}
This computation is quick using the functor of points.
View the theory as a presented by a functor on a site of spacetimes and a site of test spaces given by derived manifolds.
Now, fixing the spacetime, run over the derived manifolds associated to the shifted dual numbers $\RR[t]/(t^2)$, where the degree $|t|$ runs over the integers.~\hfill$\Diamond$
\end{rmk}

This observation about the tangent complex, and its recovery of the usual BV formulation,
shows that there is a (local) (-1)-shifted symplectic form at any point~$\cA$.
It is manifestly preserved by the translation action of the abelian group stack on itself.
Together with the data of the sheaf of cochain complexes, this local shifted symplectic structure completely determines the action functional underlying these free theories.

\subsection{On coupling constants}

We put a coupling constant $\kappa$ into the differential from degree zero to degree one,
to match the physics,
but there is another natural place to scale the differentials:
we could multiply the first (nontrivial) differential by a real number~$r$ to get the inclusion  
\[
\ZZ \xto{r} \Omega^0 \to \cdots.
\]
In the case of the compact boson, where $p=0$, 
this location seems particularly natural because $\RR/2 \pi R\ZZ$ is a circle of radius $R$ (here $r = 2 \pi R$).

Note that there is an isomorphism of sheaves
\[
\begin{tikzcd}
\ZZ \arrow[r, "1"] \arrow[d, "1"] & \Omega^0 \arrow[r, "\d"] \arrow[d, "r"] & \cdots \arrow[r, "\d"] & \Omega^p \arrow[r, "\kappa \d \star \d"] \arrow[d, "r"] & \Omega^{n-p} \arrow[r, "\d"] \arrow[d, "r\frac{\widetilde{\kappa}}{\kappa}"] & \cdots \arrow[r, "\d"] & \Omega^n \arrow[d, "r\frac{\widetilde{\kappa}}{\kappa}"]\\
\ZZ \arrow[r, "r"] & \Omega^0 \arrow[r, "\d"] & \cdots \arrow[r, "\d"] & \Omega^p \arrow[r, "\widetilde{\kappa} \d \star \d"] & \Omega^{n-p} \arrow[r, "\d"] & \cdots \arrow[r, "\d"] & \Omega^n
\end{tikzcd}
\]
for any nonzero values $r$, $\kappa$, and $\widetilde{\kappa}$. 
In other words, any ways of scaling the fields (i.e., ``coupling constants'') yield isomorphic moduli problems in a very simple way.

\begin{rmk}
Here is another way of describing the situation.
We have described a family of theories parametrized by $(\RR^\times)^2$,
where a point $(r,\widetilde{\kappa})$ denotes the bottom complex we just saw,
but we also gave an isomorphism from the theory at $(1,\kappa)$ to $(r,\widetilde{\kappa})$.
In fact, we can give an isomorphism between any pair of points.
Hence we have a groupoid whose objects form the space $(\RR^\times)^2$ and where every object has exactly one isomorphism to any other point.
This groupoid contains a skeletal subcategory parametrized the point~$(1,1)$, so there is really one theory up to isomorphism. Physicists will not find this surprising: after all, the theory is free, and thus cannot have any measurable coupling constants.~\hfill$\Diamond$
\end{rmk}

\subsection{On magnetic and electric charges}
\label{sec: charges take one}

We now make some remarks on physical interpretation.
The usual object of interest is a form $F \in \Omega^{p+1}$, 
which plays the role of electromagnetic field strength.
The field $A$ provides a potential so that we have $F = \d A$.
The most familiar case is with $p = 1$ and $n = 4$, 
so that $F$ is a 2-form $E \wedge \d t + \star B$, 
where $E$ and $B$ are the electric and magnetic fields, respectively, viewed as differential forms.
Below we go through our running examples, 
but first we make some general, orienting remarks.\footnote{A classic discussion, aimed at physicists, of how cohomology (notably, \v{C}ech and de Rham) interface with gauge theories and sigma models is~\cite{Alv}.}

\subsubsection{Structural features}

It is important to note that in our framework,
the situation is richer:
the sheaf encodes the data of a $p$-bundle with connection
satisfying the Maxwell equation $\d \star \d A = 0$.
In particular, the degree zero cohomology encodes the isomorphism classes of such (higher) bundles with connection.
To say this cleanly, recall that there is a quotient map of sheaves (valued in~$\cDz$)
\[
\Mxw_{p,n}\langle\kappa\rangle \to \cBun^\nabla_{p}
\]
where
\[
\cBun^\nabla_{p} = (\ZZ \to \Omega^0 \xto{\d} \cdots \xto{\d} \Omega^p)
\]
with $\Omega^p$ in degree zero;
this complex is just a shift of the Deligne complex~$\Del^{\leq p}$.
This map is simply ``project onto the complex in degrees zero and below.''

A $p$-bundle, even without connection, is an interesting object topologically: 
every isomorphism class of such a higher bundle on a manifold $M$ corresponds to a cohomology class $H^{p+1}(M,\ZZ)$.
For $p=1$, we are thinking about principal $U(1)$-bundles, 
and so this cohomology class is the first Chern class of the bundle.
For a higher bundle, its cohomology class can be seen as its identifying Chern class.

\subsubsection{Compact boson in two dimensions}

The compact boson has a field 
\[
\sigma: M \to \RR/\ZZ \cong S^1
\] 
where $M$ is a surface.
If we let $\theta$ denote the coordinate on the covering space $\RR$ over $\RR/\ZZ$,
then the 1-form $\d \theta$ is well-defined on $\RR/\ZZ$ and generates its cohomology.
Any field $\sigma$ pulls back this 1-form to a 1-form $\sigma^* \d\theta$ on $M$.
For any loop $\gamma: S^1 \to M$,
the {\em winding number around $\gamma$} of $\sigma$ is the integral
\[
w_\gamma(\sigma) = \int_\gamma \sigma^*\d \theta .
\]
It counts the degree of the composite map
\[
S^1 \xto{\gamma} M \xto{\sigma} S^1
\]
and is thus an integer.

 Since the winding number is constant on (free) homotopy classes of loops in $M$, it is an example of a {\em topological charge}.
In other words, we can view $\sigma^*\d \theta$ as a {\em topological current} whose charge on a codimension one submanifold $\gamma$ is~$w_\gamma(\sigma)$.
This charge is conserved due to topology,
not to the dynamics.

\subsubsection{Electromagnetism in four dimensions}

In the $p=1$ situation, we interpret this Chern class $c_1(P) \in H^2(M,\ZZ)$ as a topological current that captures {\em magnetic charge}.
The motivating example behind this interpretation is Dirac's magnetic monopole,
where we work with a spacetime $M = \RR^4 - L$ where an infinite line (e.g., coordinate axis) is removed, and so one must specify a $U(1)$-bundle on~$M$.
Any choice $P \to M$ that does {\em not} come from the trivial $U(1)$-bundle on $\RR^4$ deserves to be called a magnetic monopole:
a solution to Maxwell's equation is a connection $A$ on $P$ whose curvature $F_A$ represents $c_1(P) \in H^2(M,\RR)$, which is nonzero by hypothesis.
For any embedded 2-sphere $\Sigma$ linking $L$, 
we thus find a magnetic charge
\[
\int_{\Sigma} F_A \neq 0.
\]
As $M$ deformation-retracts onto $S^2$, the possible choices of $U(1)$-bundle form a copy of
\[
\ZZ \cong H^2(S^2,\ZZ) \cong H^2(M,\ZZ),
\]
and the magnetic charge of the monopole on $\Sigma$ is a linear functional on this cohomology group.

In our framework, the map from the constant sheaf $\ZZ$ to $\Omega^0$ 
gives a map from integral cohomology to real cohomology.
In particular, we get a map 
\[
H^2(M,\ZZ) \to H^2(M,\RR)
\]
that identifies an integral lattice 
\[
H^2_{\rm free}(M,\ZZ) \subset H^2(M,\RR).
\]
This lattice can be seen as the magnetic charges assigned to 2-cycles in~$M$.

Alternatively, think about Maxwell's equations {\em with} currents,
i.e., not in a vacuum and allowing both electric and magnetic currents.
The {\em magnetic} current density $j_{\rm mag}$ appears in the equation 
\[
j_{\rm mag} = \d F_A
\]
and so the magnetic charge within a compact oriented $3$-dimensional submanifold-with-boundary $N \subset M$ is
\[
\int_N j_{\rm mag} = \int_{\partial N} F_A.
\]
We note that the $2$-forms $F_A$ that appear in our sheaf of fields have integral periods, 
as they arise from bundles with connection.

So far we have only discussed magnetic charge and its relationship with the quotient complex~$\cBun^\nabla_{1}$.
Recall that the {\em electric} current density $j_{\rm el}$ appears in the Maxwell equation 
\[
j_{\rm el} = \d \star F_A
\]
and so the electric charge within a compact oriented $3$-dimensional submanifold-with-boundary $N \subset M$ is
\[
\int_N j_{\rm el} = \int_{\partial N} \star F_A,
\]
in parallel to the magnetic charge above, albeit without the integrality constraint.
The associated charges are electric and can take any real value,
i.e., they are not discrete --- or ``quantized'' --- in contrast to magnetic charge.

\subsection{The conceptual meaning of two short exact sequences}
\label{sec: conceptual SES}

A useful feature of describing these theories via sheaves of cochain complexes is access to tricks with sheaves, 
such as short exact sequences.
These not only give long exact sequences in cohomology (see the next section);
they also often foreground conceptual features of the problem.

Our first example involves breaking $\Mxw_{p,n}\langle\kappa\rangle$ into the BRST field content and the antifields:
consider the short exact sequence
\[
0 \to \Omega^{n-p}_{\cl}[-1] \to \Mxw_{p,n}\langle\kappa\rangle \to \cBun_p^\nabla \to 0.
\]
The map out of $\Mxw_{p,n}\langle\kappa\rangle$ sends a solution to the equation of motion to its underlying $p$-bundle with connection and thus corresponds to the forgetful map
\[
\sMxw_{p,n}\langle\kappa\rangle \to \sBun_p^\nabla
\]
at the level of derived stacks.
(We have just seen how that encodes ``magnetic charges.'')
The map of closed $n-p$-forms into $\Mxw_{p,n}\langle\kappa\rangle$ describes the fiber.
The connecting morphism in the long exact sequence of cohomology is, in essence, the differential operator $\d \star \d$ in the equation of motion.

In the derived category of sheaves, this short exact sequence corresponds to a distinguished triangle. 
If we ``rotate'' the triangle, we get a map $\cBun_p^\nabla \to \Omega^{n-p}_{\cl}$ by applying $\d \star \d$,
and it corresponds to a map of derived stacks
\[
\sBun_p^\nabla \to \Omega^{n-p}_{\cl}
\]
from the moduli of $p$-bundles with connection  maps to closed $n-p$-forms.
The moduli of solutions is the fiber of this map over zero.

Our second example involves breaking $\Mxw_{p,n}\langle\kappa\rangle$ into the constant sheaf~$\ZZ$ and the perturbative BV complex:
consider the short exact sequence
\[
0 \to \Mxw_{p,n}^\pert\langle\kappa\rangle \to \Mxw_{p,n}\langle\kappa\rangle \to \ZZ[-p-1] \to 0.
\]
The map out of $\Mxw_{p,n}\langle\kappa\rangle$ forgets a solution down to its underlying $p$-bundle.\footnote{Note that the constant sheaf $\ZZ[-p-1]$ encodes the mapping space into $B^p U(1)$ and hence describes the moduli of $U(1)$ $p$-bundles.
In other words, the cochain complex $\ZZ[-p-1]$, obtained by evaluating the sheaf on a point, presents $B^pU(1)$ as a derived stack, and the derived global sections $\ZZ[-p-1](M)$ over a manifold $M$ presents the stack $\Map(M,B^pU(1))$.}
In this sense it forgets the field theory and simply remembers topology.
In more poetic language, one might say it is a topological field theory underlying this gauge theory.
In fact, any functions on the moduli of $p$-bundles produce ``topological observables'' or ``operators''.

The subcomplex $\Mxw_{p,n}^\pert\langle\kappa\rangle$ of $\Mxw_{p,n}\langle\kappa\rangle$ 
encodes the differential geometric information and the equations of motion, and as such
describes the ``dynamics.''

We now unwind the associated long exact sequences.

\subsection{The moduli of solutions on manifolds with nontrivial topology}
\label{sec: coh computations}

We undertake some computations of cohomology with the Deligne complex and the Maxwell complex,
so as to exhibit the kind of global information encoded by them.
In particular, we want to exhibit the emergence of ``topological charges,''
such as winding number for the compact boson or the first Chern class for 1-form gauge theory, as a kind of magnetic charge.\footnote{These are called topological charges because they are conserved quantities, but conserved because they are topological invariants, not due to Noether's theorem (i.e., a symmetry).}

\subsubsection{Compact boson in two dimensions}

Our principal goal is to see how the winding number emerges.

Let $M$ be a connected surface.

We start with the BRST field content $\cBun_0$, the truncated Deligne complex 
\[
\ZZ_M \to \Omega^0_M
\]
concentrated in degrees~$-1$ to $0$,
which is quasi-isomorphic to $\Map_{C^\infty}(-,S^1)$.
That is, 
\[
H^0(\cBun_0(M)) = \Map(M,S^1) 
\]
encodes the smooth maps into $S^1$ (or ``$0$-bundles'').

Consider the short exact sequence of sheaves
\[
0 \to \Omega^0_M \hookrightarrow \cBun_0 \to \ZZ_M[1] \to 0
\]
and note the following identity in sheaf cohomology:
\[
H^k(M,\ZZ_M[1]) = H^{k+1}(M,\ZZ).
\]
This short exact sequence leads to a long exact sequence in cohomology
\[
\begin{tikzcd}
0 \arrow[r] 
& 0 \arrow[r] \arrow[d, phantom, ""{coordinate, name=Z1}] 
& H^{0}(M,\ZZ) \arrow[dll,rounded corners,to path={ -- ([xshift=2ex]\tikztostart.east)|-(Z1) [near end]\tikztonodes-| ([xshift=-2ex]\tikztotarget.west)-- (\tikztotarget)}] & \\
C^\infty(M) \arrow[r] 
& H^0(\cBun_{0}) \arrow[r] \arrow[d, phantom, ""{coordinate, name=Z2}]
& H^{1}(M,\ZZ) \arrow[dll,rounded corners,to path={ -- ([xshift=2ex]\tikztostart.east)|-(Z2) [near end]\tikztonodes-| ([xshift=-2ex]\tikztotarget.west)-- (\tikztotarget)}] &\\
0 \arrow[r] 
& H^{1}(\cBun_{0}) \arrow[r] 
& H^{2}(M,\ZZ) \arrow[r] & \cdots
\end{tikzcd}.
\]
so that 
\[
H^k(\cBun_{0}) = H^{k+1}(M,\ZZ)
\]
for $k \geq 1$. 
In degree zero, we are studying the smooth maps into $S^1$.
The map
\[
H^0(\cBun_{0}) \to H^{1}(M,\ZZ)
\]
encodes the winding number of such a map (or degree zero Deligne cocycle).
Given any two maps $f$ and $g$ with the same winding number,
let \[
c_p = f(p) - g(p) \in \RR/\ZZ \cong S^1
\]
denote their difference at a fixed point $p$ in $M$.
Then their difference everywhere on $M$ satisfies
\[
g(x) - f(x) = c(p) + \phi(x)
\]
for some smooth function $\phi: M \to \RR$ with $\phi(p) = 0$.
Thus the kernel of the winding map is precisely $C^\infty(M)/\ZZ$.
In short,
\[
H^0(\cBun_{0}) \cong H^{1}(M,\ZZ) \times C^\infty(M)/\ZZ,
\]
recording the winding number and a functional representative (up to integral ambiguity).
This information is independent of the dynamics (or equations of motion) by definition.

Now let us examine the full complex encoding the moduli of solutions:
\[
\ZZ_M \to \Omega^0_M \xto{\kappa \d \star \d} \Omega^2_M
\]
concentrated in degrees~$-1$ to~$1$.
Observe there is a short exact sequence
\[
0 \to \Omega^2[-1] \to \Mxw_{0,2}\langle\kappa\rangle \to \cBun_0 \to 0
\]
and hence a long exact sequence in cohomology
\[
\begin{tikzcd}
0 \arrow[r] 
& H^0(\Mxw_{0,2}\langle\kappa\rangle(M)) \arrow[r] \arrow[d, phantom, ""{coordinate, name=Z1}]
& H^0(\cBun_{0}) \arrow[dll,rounded corners,to path={ -- ([xshift=2ex]\tikztostart.east)|-(Z1) [near end]\tikztonodes-| ([xshift=-2ex]\tikztotarget.west)-- (\tikztotarget)}] & \\
\Omega^2(M) \arrow[r] 
& H^1(\Mxw_{0,2}\langle\kappa\rangle(M)) \arrow[r] \arrow[d, phantom, ""{coordinate, name=Z2}]
& H^{1}(\cBun_{0}) \arrow[dll,rounded corners,to path={ -- ([xshift=2ex]\tikztostart.east)|-(Z2) [near end]\tikztonodes-| ([xshift=-2ex]\tikztotarget.west)-- (\tikztotarget)}] &\\
0 \arrow[r] 
& H^2(\Mxw_{0,2}\langle\kappa\rangle(M)) \arrow[r] 
& H^{2}(\cBun_{0}) \arrow[r] & \cdots
\end{tikzcd}.
\]
Here the only interesting connecting map is the equation of motion $\kappa \d \star \d$ acting on $\Map(M,S^1)$ so
\[
H^0(\Mxw_{0,2}\langle\kappa\rangle(M)) = H^{1}(M,\ZZ) \times \ker(\kappa \d \star \d),
\]
as we expect.\footnote{When $M$ is closed, oriented, and connected, the second, dynamical factor is $\RR$ as it is the de Rham cohomology group $H^0(M,\RR)$, by Hodge theory.}
For $k > 1$, we have
\[
H^k(\Mxw_{0,2}\langle\kappa\rangle(M)) = H^k(\cBun_{0}) = H^{k+1}(M,\ZZ).
\]
The only remaining case is 
\[
H^1(\Mxw_{0,2}\langle\kappa\rangle(M)) \cong H^2(M,\ZZ) \times \Omega^2(M)/{\rm im}(\kappa \d \star \d)  
\]
where the second factor is dynamical (it is the cokernel of the linear operator encoding the equation of motion) 
and where the first factor is topological.\footnote{When $M$ is closed, oriented, and connected, the topological factor is just $\ZZ$. When it is not closed, it is zero. Similarly, when $M$ is closed, oriented, and connected, the topological factor is just $\RR$ as it is the de Rham cohomology group $H^2(M,\RR)$, by Hodge theory. In this setting, it is the linear dual (``antifields'') to the zeroth cohomology (the naive fields). In general, there is the Poincar\'e pairing between $\ker(\kappa \d \star \d)$ and $\Omega^2(M)/{\rm im}(\kappa \d \star \d)$.}

\subsubsection{Electromagnetism in four dimensions}

For simplicity, take $p=1$ and $n=4$, so that we focus on the usual setting of electromagnetism.
Let $(N,h)$ be a closed, oriented, connected Riemannian 3-manifold,
and consider the 4-manifold 
\[
M = N \times \RR
\] 
equipped with the product metric $g = h \pm \d t^2$, 
where the sign depends on the reader's preference.

The BRST field content of electromagnetism is the sheaf cohomology of $\cBun_1^\nabla$,
the truncated Deligne complex
\[
\ZZ_M \to \Omega^0_M \xto{\d} \Omega^1_M
\]
concentrated in degrees~$-2$ to~$0$. 
Using the short exact sequence of sheaves
\[
0 \to \Omega^{\leq 1}_M[1] \hookrightarrow \cBun_1^\nabla \to \ZZ_M[2] \to 0
\]
we obtain a long exact sequence in cohomology
\[
\begin{tikzcd}
0 \arrow[r] 
& 0 \arrow[r] \arrow[d, phantom, ""{coordinate, name=Z1}]
& H^{0}(N,\ZZ) \arrow[dll,rounded corners,to path={ -- ([xshift=2ex]\tikztostart.east)|-(Z1) [near end]\tikztonodes-| ([xshift=-2ex]\tikztotarget.west)-- (\tikztotarget)}]\\
H^0(N,\RR) \arrow[r] 
& H^{-1}(M,\cBun_1^\nabla) \arrow[r] \arrow[d, phantom, ""{coordinate, name=Z2}]
& H^{1}(N,\ZZ) \arrow[dll,rounded corners,to path={ -- ([xshift=2ex]\tikztostart.east)|-(Z2) [near end]\tikztonodes-| ([xshift=-2ex]\tikztotarget.west)-- (\tikztotarget)}]\\
\Omega^1(M)/\d \Omega^0(M) \arrow[r] 
& H^{0}(M,\cBun_1^\nabla) \arrow[r] \arrow[d, phantom, ""{coordinate, name=Z3}]
& H^{2}(N,\ZZ) \arrow[dll,rounded corners,to path={ -- ([xshift=2ex]\tikztostart.east)|-(Z3) [near end]\tikztonodes-| ([xshift=-2ex]\tikztotarget.west)-- (\tikztotarget)}]\\
0 \arrow[r] 
& H^{1}(M,\cBun_1^\nabla) \arrow[r] \arrow[d, phantom, ""{coordinate, name=Z4}]
& H^{3}(N,\ZZ) \arrow[dll,rounded corners,to path={ -- ([xshift=2ex]\tikztostart.east)|-(Z4) [near end]\tikztonodes-| ([xshift=-2ex]\tikztotarget.west)-- (\tikztotarget)}]\\
0 \arrow[r] 
& H^{2}(M,\cBun_1^\nabla) \arrow[r] 
& H^{4}(N,\ZZ) \cdots
\end{tikzcd}.
\]
That means 
\[
H^{k}(M,\cBun_1^\nabla) \cong H^{k+1}(N,\ZZ)
\]
for $k \geq 1$ since the complex $\Omega^{\leq 1}_M[1]$ has cohomology only in two degrees.
As the {\em free} part of integral cohomology maps into de Rham cohomology, we find
\[
H^{-1}(M,\cBun_1^\nabla) \cong \RR/\ZZ \oplus {\rm Tors}(H^{1}(N,\ZZ)) = \RR/\ZZ,
\]
since $H^1(N,\ZZ)$ has no torsion (by the universal coefficien theorem).
Finally
\[
H^{0}(M,\cBun_1^\nabla) \cong H^{2}(N,\ZZ) \oplus (\Omega^1(M)/\d \Omega^0(M))/H^1(N,\ZZ).
\]
That last summand (the iterated quotient) admits a simpler description:
first observe that
\[
\Omega^1(M)/\d \Omega^0(M) = H^1(N,\RR) \oplus \d \Omega^1(M)
\]
and so
\[
(\Omega^1(M)/\d \Omega^0(M))/H^1(N,\ZZ) \cong H^1(N,\RR)/H^1(N,\ZZ) \oplus \d \Omega^1(M).
\]
In sum,
\[
H^{0}(M,\cBun_1^\nabla) \cong H^{2}(N,\ZZ) \oplus H^1(N,\RR/\ZZ) \oplus \d \Omega^1(M).
\]
We note that $H^2(N,\ZZ)$ encodes the first Chern class, or {\em magnetic charge} of the $U(1)$-bundle,
while the torus $H^1(N,\RR/\ZZ)$ parametrizes the flat $U(1)$-connections by their monodromy.
The large factor $\d \Omega^1(M)$ parametrizes the 2-forms that could appear as possible field strengths (before imposing the equations of motion).

Now let us examine the full complex encoding the moduli of solutions:
\[
\ZZ_M \to \Omega^0_M \xto{\d} \Omega^1_M \xto{\kappa \d \star \d} \Omega^{3}_M \xto{\d} \Omega^4_M.
\]
Observe that there is a short exact sequence
\[
0 \to \Omega^3_{\cl}[-1] \to \Mxw_{2,4}\langle\kappa\rangle \to \cBun_1^\nabla \to 0
\]
and hence a long exact sequence
\[
\begin{tikzcd}
0 \arrow[r] 
& H^{-1}(\Mxw_{1,4}\langle\kappa\rangle(M)) \arrow[r]  \arrow[d, phantom, ""{coordinate, name=Z1}]
& H^{-1}(M,\cBun_1^\nabla) \arrow[dll,rounded corners,to path={ -- ([xshift=2ex]\tikztostart.east)|-(Z1) [near end]\tikztonodes-| ([xshift=-2ex]\tikztotarget.west)-- (\tikztotarget)}]\\
0 \arrow[r] 
& H^{0}(\Mxw_{1,4}\langle\kappa\rangle(M)) \arrow[r] \arrow[d, phantom, ""{coordinate, name=Z2}]
& H^{0}(M,\cBun_1^\nabla) \arrow[dll,rounded corners,to path={ -- ([xshift=2ex]\tikztostart.east)|-(Z2) [near end]\tikztonodes-| ([xshift=-2ex]\tikztotarget.west)-- (\tikztotarget)}]\\
\ker\left( \d|_{\Omega^3(M)}\right) \arrow[r] 
& H^{1}(\Mxw_{1,4}\langle\kappa\rangle(M)) \arrow[r] \arrow[d, phantom, ""{coordinate, name=Z3}]
& H^{1}(M,\cBun_1^\nabla) \arrow[dll,rounded corners,to path={ -- ([xshift=2ex]\tikztostart.east)|-(Z3) [near end]\tikztonodes-| ([xshift=-2ex]\tikztotarget.west)-- (\tikztotarget)}]\\
0 \arrow[r] 
& H^{2}(\Mxw_{1,4}\langle\kappa\rangle(M)) \arrow[r] \arrow[d, phantom, ""{coordinate, name=Z4}]
& H^{2}(M,\cBun_1^\nabla) \arrow[dll,rounded corners,to path={ -- ([xshift=2ex]\tikztostart.east)|-(Z4) [near end]\tikztonodes-| ([xshift=-2ex]\tikztotarget.west)-- (\tikztotarget)}]\\
0 \arrow[r] 
& H^{3}(\Mxw_{1,4}\langle\kappa\rangle(M)) \arrow[r] 
& H^{3}(M,\cBun_1^\nabla) \cdots\\
\end{tikzcd}.
\]
That means 
\[
H^{k}(\Mxw_{1,4}\langle\kappa\rangle(M)) \cong H^{k}(M,\cBun_1^\nabla)
\]
for $k <0$ and for $k > 1$ since the complex $\Omega^3_{\cl}[-1]$ has cohomology only in two degrees.
In particular, note that
\[
H^{-1}(\Mxw_{1,4}\langle\kappa\rangle(M)) \cong \RR/\ZZ \oplus {\rm Tors}(H^{1}(N,\ZZ))
\]
so that the global symmetries $\RR/\ZZ$ are enhanced by a finite group.
The connecting map in the long exact sequence is $\kappa \d \star \d$, so
\[
H^{0}(\Mxw_{1,4}\langle\kappa\rangle(M)) \cong \ker\left(\kappa \d \star \d|_{H^{0}(M,\cBun_1^\nabla)} \right)
\]
and thus consists of those bundles with connections whose curvatures satisfy the equation of motion $\d \star \d A = 0$.
In degrees above zero, the Deligne complex becomes integral cohomology so the connecting map acts trivially.
The last interesting cohomology group is
\[
H^{1}(\Mxw_{1,4}\langle\kappa\rangle(M)) \cong H^3(N,\ZZ) \oplus \ker\left( \d|_{\Omega^3(M)}\right)/\d\star \d\Omega^1(M),
\]
which has an interpretation in the BV formalism.

\subsection{A first and failed attempt at abelian duality}

For Maxwell theory on $\RR^4$ in a vacuum, 
one can replace $F$ with $\star F$ and get another solution to Maxwell's equations.
Roughly speaking, one might hope to replace the connection 1-form $A$ by a ``dual'' 1-form $C$ such that $F_C = \star F_A$.
In this spirit, 
we might hope to relate a $p$-form generalized Maxwell theory to an $n-p-2$-form generalized Maxwell theory,
where the curvature $F_A$ of a $p$-form $A$ is Hodge dual to the curvature $F_C$ of an $(n-p-2)$-form~$C$.

This version of duality does not hold at the level of the complexes of BV fields (or the derived stacks they present).

For the $n$-manifold $\RR^n$, the zeroth cohomology groups are isomorphic, 
but the negative cohomology groups differ for most values of~$p$: 
there is a copy of $S^1$ in degree~$-p$ for the $p$-form theory and in degree~$-(n-p-2)$ for the putative dual.

For other $n$-manifolds, even the zeroth cohomology groups differ.

\section{An alternative presentation of Maxwell theory}

In this section we offer an alternative sheaf that presents $p$-form Maxwell theory.
It makes clear why we might hope for a duality, 
but also suggests why we need to ``quantize'' electric charge to obtain duality.

\subsection{Rewriting the antifields}

Recall that on any $n$-manifold $M$, there is a cochain complex of sheaves
\begin{equation}
\label{eqn: hat omega}
\RR \to \Omega^0 \xto{\d} \cdots \xto{\d} \Omega^n
\end{equation}
concentrated in degrees~$-1$ to $n$,
and this complex is acyclic (i.e., the derived global sections are always zero) by the Poincar\'e lemma.
Let's denote this cochain complex by $\widehat{\Omega}^\bullet$ to emphasize how it is an extension of the de Rham complex.
(It is sometimes called the augmented de Rham complex.)

\begin{rmk}
In \eqref{eqn: hat omega}, the sheaf $\RR$ denotes the locally constant sheaf,
whose value on a manifold $M$ is the continuous maps from $M$ into $\RR$ with the {\em discrete} topology.
This is {\em not} the same as smooth maps into $\RR$ with the Euclidean topology;
that sheaf is denoted by $C^\infty$ or $\Omega^0$.~\hfill$\Diamond$
\end{rmk}

There is then a short exact sequence of sheaves
\[
0 \to \Omega^{\geq k} \to \widehat{\Omega}^\bullet \to \widehat{\Omega}^{< k} \to 0
\]
and hence there is a quasi-isomorphism of sheaves
\[
\widehat{\Omega}^{< k} \xto{\d} \Omega^{\geq k}[1].
\]
(Note that $\Omega^{\geq k}$ is the closed $k$-forms, up to shift.)
Let's write out the details to make the claim transparent:
\[
\begin{tikzcd}
\RR \ar[r] & \Omega^0 \ar[r,"\d"] & \cdots \ar[r,"\d"]& \Omega^{k-2} \ar[r] \ar[d] & \Omega^{k-1} \ar[d, "\d"] \ar[r] & 0 \ar[d] \ar[r] & 0 \ar[r] & \cdots & \\ 
& & \cdots \ar[r] & 0 \ar[r] & \Omega^k \ar[r, "\d"] & \Omega^{k+1} \ar[r, "\d"] & \cdots \ar[r, "\d"]& \Omega^n 
\end{tikzcd}
\]
This is just the ``boundary'' map arising from the short exact sequence above.

\begin{rmk}
The complex $\widehat{\Omega}^{< k}$, if shifted to put its rightmost term $\Omega^{k-1}$ in degree zero,
resembles the sheaf $\cBun^\nabla_{k-1}$ of Definition~\ref{dfn: bun nabla},
except that the leftmost factor is changed from $\ZZ$ to~$\RR$.
At the level of moduli spaces, the complex $\widehat{\Omega}^{< k}[k-1]$ presents the moduli of $k-1$-bundles --- with gauge group $\RR$ --- with a connection,
but where the global symmetries are trivialized (i.e., we only consider ``small'' gauge transformations\footnote{In the literature, ``small'' gauge transformations sometimes refers to the connected component of the identity in the global sections of the group of gauge transformations. Here, we use this term differently; it will do no harm if the reader imagines, for example,  gauge transformations with compact support.}).~\hfill$\Diamond$
\end{rmk}

Now we can use this quasi-isomorphism to modify $\Mxw_{p,n}\langle\kappa\rangle$ by replacing 
\[
\Omega^{n-p}_{\cl}[-1] = \Omega^{\geq n-p}[n-p-1]
\]
with $\widehat{\Omega}^{< n-p}[n-p-2]$.
In other words, consider the complex
\[
\begin{tikzcd}[row sep = tiny]
\ZZ \ar[r] & \Omega^0 \ar[r, "\d"] & \cdots \ar[r, "\d"]& \Omega^p \ar[dr, "\star \d"]& \\
&&&& \Omega^{n-p-1} \\
\RR \ar[r] & \Omega^0 \ar[r, "\d"] & \cdots \ar[r, "\d"]& \Omega^{n-p-2} \ar[ur, "-\kappa \d"]&
\end{tikzcd}
\]
where $\Omega^p$ and $\Omega^{n-p-2}$ sit in degree zero\footnote{Bear in mind that the copies of $\ZZ$ and $\RR$ do not (typically) sit in the same degree, even if the diagram might suggest that. On the other hand, $\Omega^p$ and $\Omega^{n-p-2}$ {\em do} both sit in degree zero. This proviso applies for many cochain complexes we write.}
Note that on $\RR^n$, the cohomology is straightforward to compute:
\begin{itemize}
\item it is manifestly zero in degree~2 and above,
\item it is $\RR/\ZZ$ in degree~$-p$, and
\item it is zero for any other negative degree, by the Poincar\'e lemma or manifestly.
\end{itemize}
In degree zero, a pair $(A,C) \in \Omega^p \oplus \Omega^{n-p-2}$ is closed if and only if
\[
\star F_A = \star \d A = \kappa \d C = \kappa F_C.
\]
Such a closed element $(A,C)$ satisfies the vacuum Maxwell equation
\[
\d \star F_A = \kappa \d^2 C = 0.
\]
On the other hand, the field $C$ is uniquely determined by $\star F_A$, up to exact terms, by Poincar\'e lemma.
Hence the zeroth cohomology consists of $A$ such that $\d \star F_A = 0$, up to exact terms (i.e., infinitesimal gauge equivalence).

Hodge theory tells us the first cohomology group vanishes:
every form is a unique linear combination of an exact $n-p-1$-form with a coexact form.

This complex is quasi-isomorphic to~$\Mxw_{p,n}$,
and we will use it as an alternative approach to abelian $p$-form gauge theory. The following definition fixes our notation for it.

\begin{dfn}
\label{dfn: Maxwell-tilde theory}
Let $\widetilde{\Mxw}_{p,n}\langle\kappa\rangle$ denote the sheaf on $\Riem^\ort_n$ (or $\Lrtz^\ort_n$) valued in $\cDz$ given locally by
\begin{equation}
\label{eqn: tildemax}
\begin{tikzcd}[row sep = tiny]
\ZZ \ar[r] & \Omega^0 \ar[r, "\d"] & \cdots \ar[r, "\d"]& \Omega^p \ar[dr, "\star \d"]& \\
&&&& \Omega^{n-p-1} \\
\RR \ar[r] & \Omega^0 \ar[r, "\d"] & \cdots \ar[r, "\d"]& \Omega^{n-p-2} \ar[ur, "-\kappa \d"]&
\end{tikzcd},
\end{equation}
where $\Omega^p$ and $\Omega^{n-p-2}$ are concentrated in degree zero.
\end{dfn}

This sheaf of cochain complexes presents a derived stack:
a point in this stack is a $p$-bundle with connection $A$ and an $(n-p-2)$-form $C$ such that $\star \d A - \kappa \d C = 0$, i.e.,
\[
\star F_A = \kappa F_C
\]
so that the curvatures (or field strengths) are Hodge dual, up to a scalar factor.
Due to the quasi-isomorphism of sheaves
\[
\Mxw_{p,n}\langle\kappa\rangle \xto{\simeq} \widetilde{\Mxw}_{p,n}\langle\kappa\rangle
\]
we see that $\widetilde{\Mxw}_{p,n}\langle\kappa\rangle$ also presents the derived stack~$\sMxw_{p,n}\langle\kappa\rangle$. 


\begin{eg}[Electromagnetism in four dimensions]
For $M = \RR^4$, the cochain complex is
\[
\begin{tikzcd}[row sep = tiny]
\ZZ \ar[r] & \Omega^0 \ar[r, "\d"] & \Omega^1 \ar[dr, "\star \d"]& \\
&&& \Omega^{2} \\
\RR \ar[r] & \Omega^0 \ar[r, "\d"] & \Omega^{1} \ar[ur, "-\kappa \d"]&
\end{tikzcd},
\]
with the $\Omega^1$ terms in degree zero.~\hfill$\Diamond$
\end{eg}

\begin{eg}[Compact boson in two dimensions]
For $M = \RR^2$, the cochain complex is
\[
\begin{tikzcd}[row sep = tiny]
\ZZ \ar[r] & \Omega^0 \ar[dr, "\star \d"]& \\
&& \Omega^{2} \\
\RR \ar[r] & \Omega^0 \ar[ur, "-\kappa \d"]&
\end{tikzcd},
\]
with the $\Omega^0$ terms in degree zero.
To be more explicit, note that a degree zero cocycle, locally on Euclidean $\RR^2$, is a smooth map 
\[
u: \RR^2 \to \RR
\]
such that there is another function 
\[
v: \RR^2 \to \RR
\]
where
\[
\star \d u = \d v
\]
or more explicitly
\[
\partial_x v = - \partial_y u \quad \text{ and }\quad \partial_y v = \partial_x u,
\]
which are the Cauchy-Riemann equations.
In other words, a degree zero cocycle determines a holomorphic function $u + iv$ on~$\CC$.
At the level of cohomology, we only care about $u$ up to a shift by an integer and $v$ up to a shift by a real number.~\hfill$\Diamond$
\end{eg}

\begin{eg}[Electromagnetism in three dimensions]
For $M = \RR^3$, the cochain complex is
\[
\begin{tikzcd}[row sep = tiny]
\ZZ \ar[r] & \Omega^0 \ar[r, "\d"] & \Omega^1 \ar[dr, "\star \d"]& \\
&&& \Omega^{1} \\
&\RR \ar[r] & \Omega^0 \ar[ur, "-\kappa \d"]&
\end{tikzcd},
\]
with the $\Omega^1$ terms in degree zero.~\hfill$\Diamond$
\end{eg}

\begin{eg}[Compact boson in three dimensions]
For $M = \RR^3$, the cochain complex is
\[
\begin{tikzcd}[row sep = tiny]
& \ZZ \ar[r] & \Omega^0 \ar[dr, "\star \d"]& \\
&&& \Omega^{1} \\
\RR \ar[r] & \Omega^0 \ar[r, "\d"] & \Omega^1 \ar[ur, "-\kappa \d"]&
\end{tikzcd},
\]
with the $\Omega^1$ terms in degree zero.~\hfill$\Diamond$
\end{eg}

\subsection{On electric charges}
\label{sec: elec charges}

This alternative formulation provides a slick way to describe electric charges using sheaves and moduli spaces.
The main object of interest, the field strength $F \in \Omega^{p+1}$, 
now has a potential $A$ and dual potential $C$ with $F = \d A$ and $F = \star \d C$.
We can rephrase that efficiently using our description.

First, recall how to get ``magnetic'' information. 
As we saw in Section~\ref{sec: charges take one}, there is the quotient map of sheaves (valued in~$\cDz$)
\[
\widetilde{\Mxw}_{p,n}\langle\kappa\rangle \to \cBun^\nabla_{p}
\]
by projecting onto the top row of the complex~\eqref{eqn: tildemax}.
It presents the map of moduli spaces
\[
\sMxw_{p,n}\langle\kappa\rangle \to \sBun^\nabla_p
\]
where one remembers just the underlying $p$-bundle with connection~$A$.
The magnetic charge means the integral cohomology class of the $p$-bundle.\footnote{Rather, the charge ``inside'' a closed and oriented $p+1$-manifold $\Sigma$ is obtained by pairing the homology class $[\Sigma]$ with this degree~$p+1$ cohomology class.}

Now, by projecting onto the bottom row of the complex~\eqref{eqn: tildemax}, 
we have a similar quotient map of sheaves (valued in~$\cDz$)
\[
\widetilde{\Mxw}_{p,n}\langle\kappa\rangle \to 
\widehat\Omega^{\leq n-p-2}
\]
where
\[
\widehat\Omega^{\leq n-p-2} = \RR \to \Omega^0 \xto{\d}  \cdots \xto{\d} \Omega^{n-p-2}
\]
denotes the bottom row.
This sheaf $\widehat\Omega^{\leq n-p-2}$ presents the moduli of $(n-p-2)$-bundles {\em with gauge group~$\RR$} and with connection
(but with large gauge transformations --- or global symmetries --- killed).
In concrete terms it remembers the dual connection $C$ for the field strength~$F$,
and hence its {\em real} cohomology class encodes electric charge.
These electric charges can take any real value,
i.e., they are not discrete --- or ``quantized'' --- in contrast to magnetic charge.

\section{Abelian duality}
\label{sec: main result}

Here we will offer a version of abelian duality (or electromagnetic duality) in terms of derived stacks.

\subsection{Encoding charge quantization (or discretization)}

At some level, duality is really about a theory that is different than, but very close to, abelian gauge theory as traditionally defined by a Lagrangian. 
The intuitive idea of duality is that we should be able to view either $A$ or $C$ as the ``potential'' for a field strength $F$.
But we saw in the preceding section that, for Maxwell $p$-form theory,
there is an asymmetry:
the connection $A$ lives on a $U(1)$-bundle while the connection $C$ lives on an $\RR$-bundle.
Following the path we have outlined, though, it is straightforward to put them on equal footing.

\begin{dfn}
For nonzero values $\kappa, e \in \RR$, 
let $\Mxw_{p,n}^\circ\langle\kappa,e\rangle$ denote the sheaf valued in $\cDz$ given locally by
\[
\begin{tikzcd}[row sep = tiny]
\ZZ \ar[r] & \Omega^0 \ar[r, "\d"] & \cdots \ar[r, "\d"]& \Omega^p \ar[dr, "\star \d"]& \\
&&&& \Omega^{n-p-1} \\
\ZZ \ar[r, "e"] & \Omega^0 \ar[r, "\d"] & \cdots \ar[r, "\d"]& \Omega^{n-p-2} \ar[ur, "-\kappa \d"']&
\end{tikzcd}
\]
where $e$ denotes an {\em electric} coupling constant and where $\Omega^p$ and $\Omega^{n-p-2}$ sit in degree zero.
\end{dfn}

We call this theory $\Mxw_{p,n}^\circ\langle\kappa,e\rangle$ the {\em charge-discretized} $p$-form gauge theory.
This sheaf $\Mxw_{p,n}^\circ\langle\kappa,e\rangle$ manifestly maps to $\widetilde{\Mxw}_{p,n}\langle\kappa\rangle$ by an obvious inclusion of sheaves
\begin{equation}
\label{eqn: forget mag charge}
\Mxw_{p,n}^\circ\langle\kappa,e\rangle \hookrightarrow \widetilde{\Mxw}_{p,n}\langle\kappa\rangle
\end{equation}
that forgets the electric charge.
Using 
\[
\sMxw_{p,n}^\circ\langle\kappa,e\rangle(M)
\] 
to denote the derived stack presented by the derived global sections~$\Mxw_{p,n}^\circ\langle\kappa,e\rangle(M)$,
the moduli of solutions for this theory on a Riemannian $n$-manifold~$M$,
the map \eqref{eqn: forget mag charge} presents a map 
\begin{equation}
\label{project to usual}
\sMxw_{p,n}^\circ\langle\kappa,e\rangle \to \sMxw_{p,n}\langle\kappa\rangle
\end{equation}
of derived stacks.

Physically, we are making electric charge {\em discrete}, rather than continuous, 
when we restrict to $\Mxw_{p,n}^\circ\langle\kappa,e\rangle$ inside~$\widetilde{\Mxw}_{p,n}\langle\kappa\rangle$.\footnote{People often say that electric charge is then quantized, since it comes in discrete values (``quanta''), 
but this use of ``quantize'' is different than computing a path integral or applying canonical quantization. 
Hence we avoid this terminology for clarity's sake.}
In a sense,
we are changing the generalized global symmetries by working with the theory described by this ``substack.''

When $M$ is closed,
the stack $\sMxw_{p,n}^\circ\langle\kappa,e\rangle(M)$ is $(-1)$-shifted {\em pre}\/symplectic,
where this presymplectic structure is simply the shifted symplectic structure of gauge theory pulled back along the inclusion~\eqref{eqn: forget mag charge}.
(For $M$ not closed, there is a ``local'' presymplectic structure,
as a kind of pairing valued in volume forms on~$M$.)

In short, this sheaf $\Mxw^\circ_{p,n}\langle\kappa,e\rangle$ describes a (slightly degenerate) classical field theory where both electric and magnetic charges are discrete.
We need the language of moduli spaces or derived geometry to articulate this theory carefully,
because a traditional Lagrangian (or action functional) is not sufficient to describe this condition.\footnote{A theory presented by a Lagrangian (i.e., the solutions to its equations of motion) is always $-1$-shifted symplectic when constructed in the BV formalism. 
The charge-discretized theory is presymplectic, and thus cannot be fully described by an action functional in traditional fashion.}

This derived stack $\sMxw_{p,n}^\circ\langle\kappa,e\rangle(M)$ can be seen as a fiber product in the following way.
For ease of notation, we set $e =1$.
Consider the sheaf map
\[
\Mxw_{p,n}^\circ\langle\kappa,1 \rangle \to \cBun_p^\nabla \times \cBun_{n-p-2}^\nabla
\]
and, furthermore, that $\Mxw_{p,n}^\circ\langle\kappa,1\rangle$ is a (co)cone 
\[
\Cone\left( \cBun_p^\nabla \oplus \cBun_{n-p-2}^\nabla \xto{\star\d - \kappa \d} \Omega^{n-p-1}\right).
\]
In stack-theoretic terms, we have a fiber product
\[
\begin{tikzcd}
\sMxw_{p,n}^\circ\langle\kappa,1\rangle \ar[r] \ar[d] & \sBun_p^\nabla  \ar[d, "\star\d "]  \\
\sBun_{n-p-2}^\nabla \ar[r,"- \kappa \d"] & \Omega^{n-p-1}
\end{tikzcd},
\]
identifying the charge-discretized solutions as pairs of higher bundles with connection whose curvatures are Hodge dual.
(Equivalently, there is the fiber product
\[
\begin{tikzcd}
\sMxw_{p,n}^\circ\langle\kappa,1\rangle \ar[r] \ar[d] & \sBun_p^\nabla \times \sBun_{n-p-2}^\nabla \ar[d, "\star\d - \kappa \d"]  \\
\ast \ar[r,"\text{zero}"] & \Omega^{n-p-1}
\end{tikzcd},
\]
of derived stacks.)

We emphasize that this fiber product description is a literal-minded interpretation of the idea that every field-strength $F_A$ must have a dual $F_A = \star F_C$ for some dual connection~$C$.
(Bunster and Henneaux develop closely related ideas in~\cite{BunHen}.)

We discuss examples in detail below, but it may be illuminating to see a special case now.
Take $n=4$ and $p=1$, which is ordinary electromagnetism. 
For simplicity, we take $\kappa = 1 = e$.
The complex is then
\[
\begin{tikzcd}[row sep = tiny]
\ZZ \ar[r] & \Omega^0 \ar[r, "\d"] & \Omega^1 \ar[dr, "\star \d"]& \\
&&& \Omega^{2} \\
\ZZ \ar[r] & \Omega^0 \ar[r, "\d"] &  \Omega^{1} \ar[ur, "-\d"']&
\end{tikzcd}
\]
with the 1-forms (both copies) sitting in degree zero.

This example makes apparent how abelian duality should work:
the top and bottom rows are isomorphic and hence can be swapped,
while $-\star$ is an isomorphism of $\Omega^2$.
This combination of isomorphisms (swap on the rows and $-\star$ on the far right term) 
\[
\begin{tikzcd}[row sep = tiny]
\ZZ \ar[r] \arrow[dd, leftrightarrow, dashed, bend right] & \Omega^0 \ar[r, "\d"] \arrow[dd, leftrightarrow, dashed, bend right] & \Omega^1 \ar[dr, "\star \d"] \arrow[dd, leftrightarrow, dashed, bend right]& \\
&&& \Omega^{2} \arrow[loop right, leftrightarrow, dashed, "-\star"] \\
\ZZ \ar[r] & \Omega^0 \ar[r, "\d"] &  \Omega^{1} \ar[ur, "-\d"']&
\end{tikzcd}
\]
encodes electromagnetic duality because it swaps the role of the field (or connection) and its dual.

What happens if we reintroduce the constants $\kappa$ and~$e$?
We turn to that question before articulating a version of abelian duality for all $p$-form gauge theories.
In particular, we identify when different choices of coupling constant yield equivalent stacks (or theories);
these are not dualities but merely rescalings of the field content.

\subsection{A digression on coupling constants}

There are four locations where it is natural, from our point of view, to scale the differentials: 
the two inclusions of integers and the two maps into $\Omega^{n-p-1}$.
Diagrammatically, 
\[
\begin{tikzcd}[row sep = tiny]
\ZZ \ar[r, "m"] & \Omega^0 \ar[r, "\d"] & \cdots \ar[r, "\d"]& \Omega^p \ar[dr, "\lambda \star \d"]& \\
&&&& \Omega^{n-p-1} \\
\ZZ \ar[r, "e"] & \Omega^0 \ar[r, "\d"] & \cdots \ar[r, "\d"]& \Omega^{n-p-2} \ar[ur, "-\kappa \d"']&
\end{tikzcd}
\]
where these constants are $m,e,\lambda, \kappa \in \RR^\times$.
Note that we do {\em not} scale the de Rham differentials $\d$ that lie between $\Omega^0$ and $\Omega^p$ on the top row and between $\Omega^0$ and $\Omega^{n-p-2}$ on the bottom row. 

We denote this complex as
\[
T\langle m,e,\lambda, \kappa\rangle
\]
and we wish to describe simple isomorphisms between these complexes so as to pin down how many meaningful coupling constants there are.
(We fix $n$ and $p$ for this section, so do not include in the notation here.)

For brevity, let $\Omega^{[0,p]}$ denote
\[
\Omega^0 \xto{\d} \cdots \xto{\d} \Omega^p
\]
for the noninteger part of the top row,
and similarly $\Omega^{[0,n-p-2]}$ for the noninteger part of the bottom row.

We will allow ourselves to scale (uniformly) the component $\Omega^{[0,p]}$ of the top row,
to scale the component $\Omega^{[0,n-p-2]}$ of the bottom row,
and to scale the rightmost factor~$\Omega^{n-p-1}$.
Thus we have three nonzero parameters for constructing maps between theories.

We then have the following isomorphisms:
\begin{itemize}

\item $T\langle m,e,\lambda, \kappa \rangle \xto{\cong} T \langle 1,e,m\lambda, \kappa \rangle$: 
Apply the identity to the bottom row and $\Omega^{n-p-1}$ 
but apply $(1/m) \id$ to $\Omega^{[0,p]}$ and the identity to the top factor of~$\ZZ$. 
This is a cochain isomorphism.

\item $T \langle m,e,\lambda, \kappa \rangle \xto{\cong} T \langle m,1,\lambda, e\kappa \rangle$: 
Apply the identity to the top row and $\Omega^{n-p-1}$ 
but apply $(1/e) \id$ to $\Omega^{[0,n-p-2]}$ and the identity to the bottom factor of~$\ZZ$.
This is a cochain isomorphism.

\item $T \langle m,e,\lambda, \kappa \rangle \xto{\cong} T \langle m,e,1, \kappa \lambda \rangle$: 
Apply the identity to the top and bottom rows 
but apply $(1/\lambda)\id$ to $\Omega^{n-p-1}$. 
This is a cochain isomorphism.

\end{itemize}
And by composing these types of isomorphisms, we can produce an isomorphism
\[
T \langle m,e,\lambda, \kappa \rangle \xto{\cong} T\langle 1,1,1, \tfrac{e\kappa}{m\lambda} \rangle
\]
and so we see that, up to isomorphism, there is precisely {\em one} coupling constant. 

In particular, we see
\begin{equation}
\label{eqn: iso of coupling constants}
\Mxw_{p,n}^\circ\langle\kappa,e\rangle = T \langle 1,e,1,\kappa \rangle \cong T \langle 1,1,1,e \kappa \rangle = \Mxw_{p,n}^\circ \langle e\kappa,1 \rangle
\end{equation}
which will play a role in abelian duality.

\begin{rmk}
Here is another way of describing the situation that some readers may find convivial.
We have described a family of theories parametrized by $(\RR^\times)^4$,
but it is better to consider the groupoid whose objects form the space $(\RR^\times)^4$ and whose morphisms are generated by the scaling isomorphisms (so there are $(\RR^\times)^3$ morphisms out of each object).
We have just shown that this groupoid contains a skeletal subcategory parametrized by the discrete groupoid with objects~$\RR^\times$.~\hfill$\Diamond$
\end{rmk}

\subsection{A statement of duality}

Here is the lovely consequence of our technological build-up, 
and it justifies thinking about this non-obvious variant $\Mxw_{p,n}^\circ\langle\kappa,e\rangle$ of abelian $p$-form gauge theory.

\begin{thm}
There is an isomorphism 
\[
\sMxw_{p,n}^\circ\langle\kappa,e\rangle \xto{\cong} \sMxw_{n-p-2,n}^\circ\langle 1/\kappa,1/e\rangle
\]
as sheaves on $\Riem_n^{\ort}$ (or $\Lrtz_n^{\ort}$) valued in derived stacks
for all natural numbers $n$ and $0\leq p \leq n-2$ and for all coupling constants $e, \kappa \in \RR^\times$.
\end{thm}

This map realizes ``electromagnetic duality'' in a very direct way, at the classical level.
After the proof, we describe the concrete examples of 4d electromagnetic duality,
T-duality for the compact boson,
and the 3d duality between electromagnetism and a compact boson.

\begin{rmk}
The awkwardness of the indexing by $p$ was discussed in Remark~\ref{rmk: indexing for duality}.
If one indexes by the degree of the curvature (or field strength), 
the duality is easier to remember:
for $p$-form gauge theory, the field strength $F$ has degree $p+1$ and its dual $\star F$ has degree $n-(p+1)$.
That is, duality swaps (co)dimension.~\hfill$\Diamond$
\end{rmk}

\begin{proof}
This map is the composite of two maps we have discussed already.
We start with the complex $\Mxw_{p,n}^\circ\langle\kappa,e\rangle$.
The first step is to swap the top and bottom rows and apply the Hodge dual to $\Omega^{p+1}$ to get $\Omega^{n-p-1}$.
The new complex is
\[
\begin{tikzcd}[row sep = tiny]
\ZZ \ar[r, "e"] & \Omega^0 \ar[r, "\d"] & \cdots \ar[r, "\d"]& \Omega^{n-p-2} \ar[dr, "-\kappa \d"]&\\
&&&& \Omega^{n-p-1} \\
\ZZ \ar[r, "1"] & \Omega^0 \ar[r, "\d"] & \cdots \ar[r, "\d"]& \Omega^p \ar[ur, "\star \d"']& 
\end{tikzcd}
\]
which would be $T\langle e,1,-\kappa,-1\rangle$ in the notation from the preceding section.
Now we know 
\[
T \langle e,1,-\kappa,-1\rangle \cong T\langle 1,1,1,1/\kappa e\rangle
\]
by~\eqref{eqn: iso of coupling constants} and, moreover, 
\[
T \langle 1,1,1,1/\kappa e\rangle \cong T\langle 1,1/e,1,1/\kappa\rangle
\]
That last theory is $\Mxw_{n-p-2,n}^\circ\langle 1/\kappa,1/e\rangle$ as claimed.
\end{proof}

\begin{rmk}
As a heuristic comment, note that in the limit where $e \to 0$, with $\kappa = 1$,
the lattice of electric charges $\ZZ e$ ``fills up'' the whole real line $\RR$
and hence (roughly) converges to the theory $\widetilde{\Mxw}_{p,n}\langle \kappa\rangle$. 
Under duality, the dual theory's (action functional) coupling constant is growing as $1/e \to \infty$. 

One can imagine other versions of such ``limiting procedures'' at an intuitive level. For instance, taking $e \to 0$ and $m \to \infty$, we get a family of theories for which the coupling constant tends to zero. If we imagine that the lattice of electric charges ``fills up'' a copy of $\RR$ and the lattice of magnetic charges ``disappears at infinity,'' one schematically imagines the limit as a complex quasi-isomorphic just to~$\Mxw_{p,n}^\pert\langle\kappa\rangle$ as in~\eqref{cplx: pert mxw}. Of course, this limiting procedure does not make rigorous sense, but does illustrate how our constructions dovetail with ideas about the validity of perturbation theory in the weak-coupling limit.~\hfill$\Diamond$
\end{rmk}

\begin{eg}[Electromagnetic duality in four dimensions]
The 1-form gauge theory $\Mxw_{1,4}^\circ\langle\kappa,e\rangle$ is
\[
\begin{tikzcd}[row sep = tiny]
\ZZ \ar[r] & \Omega^0 \ar[r, "\d"] & \Omega^1 \ar[dr, "\star \d"]& \\
&&& \Omega^{2} \\
\ZZ \ar[r,"e"] & \Omega^0 \ar[r, "\d"] &  \Omega^{1} \ar[ur, "-\kappa \d"']&
\end{tikzcd}
\]
while its dual is the 1-form gauge theory $\Mxw_{1,4}^\circ \langle 1/\kappa,1/e\rangle$ where
\[
\begin{tikzcd}[row sep = tiny]
\ZZ \ar[r] & \Omega^0 \ar[r, "\d"] & \Omega^1 \ar[dr, "\star \d"]& \\
&&& \Omega^{2} \\
\ZZ \ar[r,"1/e"] & \Omega^0 \ar[r, "\d"] &  \Omega^{1} \ar[ur, "-(1/\kappa) \d"']&
\end{tikzcd}
\]
generalizing our discussion from above.
In particular, when $\kappa = 1 = e$, the theory is self-dual.~\hfill$\Diamond$
\end{eg}

\begin{eg}[T-duality\footnote{The classic reference is \cite{Buscher}, but note also \cite{Sath}. Subsequent important sources includes \cite{RocVer,BouEvsMat, BunSch}. The example here is the simplest case and avoids the rich phenomena that the subsequent work explores, but one should be able to synthesize these views.} in two dimensions]
This duality looks most like the conventional statement if we take $\kappa = 1$ and $e = R$. 
The compact boson theory $\Mxw_{0,2}^\circ \langle 1,R \rangle$ is
\[
\begin{tikzcd}[row sep = tiny]
\ZZ \ar[r] & \Omega^0  \ar[dr, "\star \d"]& \\
&& \Omega^{1} \\
\ZZ \ar[r,"R"] & \Omega^0 \ar[ur, "-\d"']&
\end{tikzcd}
\]
while its dual is the compact boson theory $\Mxw_{0,2}^\circ \langle 1,1/R \rangle$ where
\[
\begin{tikzcd}[row sep = tiny]
\ZZ \ar[r] & \Omega^0  \ar[dr, "\star \d"]& \\
&& \Omega^{1}. \\
\ZZ \ar[r,"1/R"] & \Omega^0 \ar[ur, "-\d"']&
\end{tikzcd}
\]
When $R =1$, the theory is self-dual.~\hfill$\Diamond$
\end{eg}

\begin{eg}[Duality in three dimensions: gauge and $\sigma$-model]
The 1-form gauge theory $\Mxw_{1,3}^\circ\langle\kappa,e\rangle$ is
\[
\begin{tikzcd}[row sep = tiny]
\ZZ \ar[r, "1"] & \Omega^0 \ar[r, "\d"] & \Omega^1 \ar[dr, "\star \d"]& \\
&&& \Omega^{1} \\
& \ZZ \ar[r, "e"] & \Omega^0 \ar[ur, "-\kappa \d"']&
\end{tikzcd}
\]
while its dual is the compact boson theory $\Mxw_{0,3}^\circ \langle 1/\kappa,1/e\rangle$, i.e.
\[
\begin{tikzcd}[row sep = tiny]
&\ZZ \ar[r, "1"] & \Omega^0  \ar[dr, "\star \d"]& \\
&&& \Omega^{2} \\
\ZZ \ar[r, "1/e"] & \Omega^0 \ar[r, "\d"] & \Omega^1 \ar[ur, "-(1/\kappa) \d"']&
\end{tikzcd}.
\]
This is our first example where no {\em self}\/-duality is possible.~\hfill$\Diamond$
\end{eg}

\subsection{On lattices related to charges}
\label{sec: on lattices}

We revisit the decomposition of gauge theory into a topological and a dynamical part,
as in Section~\ref{sec: conceptual SES}.

\begin{dfn}
Let $\widetilde{\Mxw}{}^{\pert}_{p,n}$ denote the sheaf
\[
\begin{tikzcd}[row sep = tiny]
\Omega^0 \ar[r, "\d"] & \cdots \ar[r, "\d"]& \Omega^p \ar[dr, "\star\d"]& \\
&&& \Omega^{n-p-1} \\
\Omega^0 \ar[r, "\d"] & \cdots \ar[r, "\d"]& \Omega^{n-p-2} \ar[ur, "-\kappa \d"]&
\end{tikzcd},
\]
where $\Omega^p$ and $\Omega^{n-p-2}$ are concentrated in degree zero.
\end{dfn}

This dynamical theory looks like $\Mxw^\circ_{p,n}\langle\kappa,e\rangle$ with the $\ZZ$ factors sheared off the left side.
On a contractible spacetime, it has the same zeroth cohomology as the $p$-form gauge theories we've seen,  
so locally it seems the ``same'' if one only examines the ``physical'' zeroth cohomology.
But on a contractible spacetime, 
its cohomology in degrees $-p$ and~$-n+p+2$ is $\RR$ by the Poincar\'e lemma,
and hence differs from the cohomology of~$\widetilde{\Mxw}_{p,n}$ or~$\Mxw^\circ_{p,n}\langle\kappa,e\rangle$.

This discrepancy can be captured very cleanly using sheaf-theoretic language.
There is a short exact sequence of sheaves
\begin{equation}
\label{SES for circ mxw}
0 \to \widetilde{\Mxw}{}^{\pert}_{p,n} \to \Mxw^\circ_{p,n}\langle\kappa,e\rangle \to \ZZ[p+1] \oplus \ZZ[n-p-1] \to 0
\end{equation}
that can be interpreted as extending a purely topological field theory
\[
\ZZ[p+1] \oplus \ZZ[n-p-1],
\] 
encoding just the electric and magnetic charges,
by a dynamical theory~$\widetilde{\Mxw}{}^{\pert}_{p,n}$.

At the level of derived stacks, it says we have a map
\[
\sMxw^\circ_{p,n}\langle\kappa,e\rangle(M) \to \Map(M,B^{p+1} \ZZ) \times \Map(M,B^{n-p-1} \ZZ)
\]
or, taking connected components, 
\[
\pi_0(\sMxw^\circ_{p,n}\langle\kappa,e\rangle(M)) \to H^{p+1}(M,\ZZ) \times H^{n-p-1}(M,\ZZ).
\]
In other words, each solution in $\sMxw^\circ_{p,n}\langle\kappa,e\rangle(M)$ has underlying cohomology classes in $H^{p+1}(M,\ZZ)$ and in $H^{n-p-1}(M,\ZZ)$.
For the case of 4-dimensional electromagnetism, these two classes live in $H^2(M,\ZZ)$ and are precisely the first Chern classes of line bundles, and we have already interpreted these classes in terms of electric and magnetic charges.

After fixing a point in the base (essentially, an element of the charge lattice), 
the fiber of the map is described by the perturbative theory~$\widetilde{\Mxw}{}^{\pert}_{p,n}$,
which captures the dynamics.

\subsection{Higher form symmetries in this context}
\label{sec: higher form sym via derived stacks}

We make here an aside on how to see typical magnetic and electric higher form symmetries in this framework,
and we explain how it provides more sophisticated refinements.

Consider the forgetful map
\[
\sMxw^\circ_{p,n} \to \sBun^\nabla_p \times \sBun^\nabla_{n-p-2}
\]
that remembers the underlying bundles with connection. 
Every bundle with connection has an underlying bundle so we have a map
\begin{equation}
\label{forget to charges}
\sMxw^\circ_{p,n}(M) \to H^{p+1}(M,\ZZ) 
\end{equation}
that encodes the ``magnetic charge,'' as just discussed.
Such a map determines an observable, typically known as a higher form symmetry; 
we now explain how to get the usual formula for the magnetic $n-p-2$-form symmetry.

Notice that pulling back along~\eqref{forget to charges} determines a map
\[
U^{\rm mag}: \Hom_{\Ab}(H^{p+1}(M,\ZZ), U(1)) \to \cO(\sMxw^\circ_{p,n}(M))^\times
\]
that takes a character on the abelian group $H^{p+1}(M,\ZZ)$ and returns a nowhere-vanishing function on the derived stack $\sMxw^\circ_{p,n}(M)$.
In more concrete terms there is a character
\[
\begin{tikzcd}[row sep=tiny]
H^{p+1}(M,\ZZ) \ar[r] & U(1)\\
c \ar[r, mapsto] & e^{i\alpha \int_N c} = g^{\langle [N], c\rangle}
\end{tikzcd}
\]
for each an element $g = e^{i\alpha} \in U(1)$ and each choice of a closed $n-p-1$-dimensional submanifold $N \subset M$ (which determines a homology class $[N] \in H_{n-p-1}(M,\ZZ)$). 
Here $c$ denotes an integral cohomology class and $\langle [N], c\rangle$ denotes the co/homology pairing.
This character becomes the observable
\[
U^{\rm mag}_g(N^{(n-p-1)}): (P \to M, A) \mapsto e^{i\alpha \int_N F_A} = g^{\langle N, c(P)\rangle}
\]
where $P \to M$ is a $p$-bundle and $A$ is its connection.
Here $c(P)$ is the characteristic class of $P$ and hence can be identified with the curvature $F_A$ by Chern-Weil theory.
Note that this formula is precisely formula~(4.2) from~\cite{GKSW}, if we take~$p=1$.

There is a parallel construction for the electric $p$-form symmetry, recovering formula~(4.1) of~\cite{GKSW}.

Finally, we observe that this construction admits a substantial improvement,
by replacing $H^{p+1}(M,\ZZ)$ with the corresponding differential cohomology group or, even better, the derived stack~$\sBun^\nabla_p$.
The forgetful map
\[
f: \sMxw^\circ_{p,n} \to \sBun^\nabla_p
\]
induces, by pullback, a map
\[
f^*: \cO(\sBun^\nabla_p) \to \cO(\sMxw^\circ_{p,n}(M))
\]
between algebras of functions.\footnote{A character is a special kind of function, which is nowhere vanishing, and hence the pullback $f^*$ maps $\Hom(\sBun^\nabla_p, U(1))$ to~$\cO(\sMxw^\circ_{p,n}(M))^\times$. Such characters are particularly useful when quantizing these abelian gauge theories.}
The observables produced by pullback can be viewed as  nontrivial refinements of the usual higher form symmetry operators.
In particular, these include functions that depend on the differential cohomology class of the bundle with connection,
not just the underlying bundle.

\subsection{Moduli-theoretic meaning of charge discretization}

In mathematical terms, the quotient sheaf 
\[
\widetilde{\Mxw}_{p,n}\langle\kappa\rangle/\Mxw_{p,n}^\circ\langle\kappa,e\rangle
\]
is the dg abelian group $(\RR/\ZZ e)[n-p-1]$ corresponding to a classifying space $B^{n-p-1}(\RR/\ZZ e)$ where $\ZZ e \subset \RR$.
Note that this quotient group $\RR/\ZZ e$ has the discrete topology
(it is not isomorphic to $U(1)$ with its manifold structure),
so this dg abelian group models $B^{n-p-1}(\RR/\ZZ e)$.
In other words, it presents the classifying stack for a higher abelian group $B^{n-p-2}(\RR/\ZZ e)$ and hence parametrizes bundles for this higher group.
We have a map of stacks
\[
q: \widetilde{\sMxw}_{p,n}\langle\kappa\rangle \to \widetilde{\sMxw}_{p,n}\langle\kappa\rangle/\sMxw_{p,n}^\circ\langle\kappa,e\rangle \simeq \sBun_{B^{n-p-2}(\RR/\ZZ e)}
\]
so that a solution to $p$-form gauge theory has an underlying $B^{n-p-2}(\RR/\ZZ e)$-bundle, as kind of ``electric charge.''
Moreover, we can view $\sMxw_{p,n}^\circ\langle\kappa,e\rangle$ as a pullback
\begin{equation}
\label{fiber bundle}
\begin{tikzcd}
\sMxw_{p,n}^\circ\langle\kappa,e\rangle \ar[r] \ar[d] & \sMxw_{p,n}\langle\kappa\rangle \ar[d, "q"] \\
\ast \ar[r] & \sBun_{B^{n-p-2}(\RR/\ZZ e)}
\end{tikzcd},
\end{equation}
so that $\sMxw_{p,n}^\circ\langle\kappa,e\rangle$ is the fiber of the theory $\sMxw_{p,n}\langle\kappa\rangle$ over $\sBun_{B^{n-p-2}(\RR/\ZZ e)}$.
Such a pullback amounts to giving the data of a trivialization of the charge,\footnote{As an analogy, consider how lifting an $SO(3)$-bundle to an $SU(2)$-bundle amounts to trivializing its Stiefel-Whitney class~$w_2$, a similar construction appearing in the center symmetry of nonabelian gauge theory.}
and the choice is not unique:
there is a space of trivializations described by the derived abelian group stack~$\sBun_{B^{n-p-2}(\RR/\ZZ e)}$. 

\begin{rmk}
Alternatively, swapping the corners of this diagram, we have a pullback square
\[
\begin{tikzcd}
\sMxw_{p,n}^\circ\langle\kappa,e\rangle \ar[r] \ar[d] & \ast \ar[d]  \\
\sMxw_{p,n}\langle\kappa\rangle \ar[r, "q"] & \sBun_{B^{n-p-2}(\RR/\ZZ e)}
\end{tikzcd},
\]
and so view $\sMxw_{p,n}^\circ\langle\kappa,e\rangle$ as the total space of a bundle with fiber~$\sBun_{B^{n-p-3}(\RR/\ZZ e)}$.~\hfill$\Diamond$
\end{rmk}

Let's unpack our story on a spacetime manifold $M$.
A point in $\sMxw_{p,n}^\circ\langle\kappa,e\rangle(M)$\footnote{That is, a map from the derived manifold $\ast$ (with $\RR$ as ring of functions) to~$\sMxw_{p,n}^\circ\langle\kappa,e\rangle(M)$.} amounts to 
an element $\alpha$ of $\sMxw_{p,n}\langle\kappa\rangle(M)$ along with a trivialization of its charge
\[
q(\alpha) \in \Map(M,B^{n-p-1} (\RR/\ZZ e)).
\]
That mapping space has connected components
\[
\pi_0 \Map(M,B^{n-p-1} (\RR/\ZZ e))
= H^{n-p-1}(M,\RR/\ZZ e)
\]
which sits in a short exact sequence
\[
0 \to H^{n-p-1}(M,\RR)/eH^{n-p-1}_\free(M,\ZZ) \to H^{n-p-1}(M,\RR/\ZZ e) \to H^{n-p}_\tors(M,\ZZ) \to 0
\]
arising from the Bockstein sequence for $0 \to e\ZZ \to \RR \to \RR/e\ZZ \to 0$.
So the charge $q(\alpha)$ determines a cohomology class in $H^{n-p-1}(M,\RR/\ZZ e)$, which combines continuous data (from real cohomology) with a discrete torsion invariant in $H^{n-p}_\tors(M,\ZZ)$.
The charge $q(\alpha)$ itself lives in the mapping space,
which we view as an enhancement of this cohomology group.

The space of such trivializations is parametrized by $\Map(M,B^{n-p-2}(\RR/\ZZ e))$,
whose connected components sit in the analogous short exact sequence
\[
0 \to H^{n-p-2}(M,\RR)/eH^{n-p-2}_\free(M,\ZZ) \to H^{n-p-2}(M,\RR/\ZZ e) \to H^{n-p-1}_\tors(M,\ZZ) \to 0.
\]
In this sense, the field theory $\sMxw_{p,n}^\circ\langle\kappa,e\rangle(M)$ can be viewed as consisting of the charge-zero fields of $\widetilde{\sMxw}_{p,n}\langle\kappa\rangle(M)$ and background fields given by $(n-p-2)$-bundles for group~$\RR/\ZZ e$.
This description offers another mathematically precise way of saying it is gauge theory equipped with a ``higher form symmetry.''

\subsection{On correspondences}
\label{sec: correspondence}

By using the duality isomorphism and the maps of forgetting charges, 
we get a correspondence
\[
\begin{tikzcd}
& \sMxw_{p,n}^\circ\langle\kappa,e\rangle \ar[dl, swap, "\pi_{L}"] \ar[dr, "\pi_{R}"] & \\
\sMxw_{p,n}\langle\kappa \rangle & &\sMxw_{n-p-2,n}\langle 1/\kappa \rangle
\end{tikzcd}
\]
of derived stacks,
where each leg is a version of~\eqref{project to usual}.
Moreover, in light of~\eqref{fiber bundle}, 
this correspondence realizes $\sMxw_{p,n}^\circ\langle\kappa,e\rangle$ as the total space of two distinct ``fiber bundles.''
Locally in spacetimes, the left leg $\pi_L$ realizes $\sMxw_{p,n}^\circ\langle\kappa,e\rangle$ as a bundle for $B^{n-p-2} (\RR/e\ZZ)$,
while the right leg $\pi_R$ realizes it as a bundle for~$B^p (\RR/\ZZ)$.

Note that this statement about bundles is formulated in terms of sheaves on spacetimes valued in derived stacks.
If we evaluate on a spacetime manifold~$M$,
the fiber is essentially $\Map(M,B^{n-p-2}(\RR/e\ZZ))$, 
a moduli of higher $\RR/e\ZZ$-bundles on~$M$,
itself a derived abelian group stack.

This correspondence yields integral transforms between the two theories, 
relating path integrals (at a physical level) and categories of sheaves (at a mathematical level, e.g., Fourier-Mukai transforms).

For instance, it seems to match with the argument from \cite{WittenIAS} for these dualities: 
\begin{itemize}
\item an observable $f$ for $p$-form Maxwell theory can be pulled back to an observable $\pi_{L}^* f$ for the charge-discretized theory, and
\item by integrating over the fiber, one obtains an observable ${\pi_R}_! \pi_{L}^* f$ for the $n-p-2$-form Maxwell theory.
\end{itemize}
In other words, we only need to treat the fiber direction as ``quantum'' (i.e., it is the direction along which we path-integrate).

It is interesting to wonder how categorical analogs --- pushing and pulling sheaves along this correspondence,
following work such as \cite{DonPan, AriBloPan, BunSchSpiTho, GamHilMaz} (among many others) --- 
might illuminate features of duality.

\begin{rmk}
In \cite{Elliott}, a story is articulated about transferring observables between generalized Maxwell theories, 
but it involves a correspondence of {\it observables} rather than field theories.
It uses, as an intermediary, the theory encoded by the cochain complex
\[
\Omega^p \xto{\d \star \d} \Omega^{n-p}
\]
that is ignorant of the (higher) gauge symmetries and their antifields.~\hfill$\Diamond$
\end{rmk}

\subsection{Duality as a $p$-form theory coupled to a topological $BF$ theory}

It is natural to wonder if there is a form of abelian duality {\em without} imposing charge discretization.
After all, the correspondence shows that $p$-form and $n-p-2$-form Maxwell theories differ only by discrete gauge fields,
so one might hope to turn these extra fields on or off.

Such a relationship should hold, albeit with subtleties: 
as shown in \cite[\S 47.2]{MooSax},
a careful computation of the partition functions shows that they agree up to an explicit factor depending on Ray-Singer torsion and the torsion in the integral cohomology.
One might hope that this factor is the partition function of a topological field theory, of BF-type.

In the formulation of our paper, we ask: can one describe $\sMxw_{p,n}\langle\kappa\rangle$ in terms of $\sMxw_{n-p-2,n}\langle 1/\kappa\rangle$, 
but perhaps coupled to another theory?

We offer here a formulation in terms of the model $\widetilde{\Mxw}_{p,n}$ as follows.
Recall that $\widetilde{\Mxw}_{p,n}$ is
\[
\begin{tikzcd}[row sep = tiny]
\ZZ \ar[r] & \Omega^0 \ar[r, "\d"] & \cdots \ar[r, "\d"]& \Omega^p \ar[dr, "\star \d"]& \\
&&&& \Omega^{n-p-1} \\
\RR \ar[r] & \Omega^0 \ar[r, "\d"] & \cdots \ar[r, "\d"]& \Omega^{n-p-2} \ar[ur, "-\kappa \d"]&
\end{tikzcd}
\]
and it is very similar to $\widetilde{\Mxw}_{n-p-2,n}$ except that the placement of $\ZZ$ and $\RR$ must be swapped.
It is possible to implement that swap by adding some extra terms to~$\widetilde{\Mxw}_{p,n}$:
consider the complex
\begin{equation}
\label{eqn: ugly duality}
\begin{tikzcd}[row sep = tiny]
\ZZ \ar[r, hook] \ar[dr, "-1"'] & \RR \ar[dr, hook] &&&&\\
&\ZZ \ar[r] & \Omega^0 \ar[r, "\d"] & \cdots \ar[r, "\d"]& \Omega^p \ar[dr, "\star \d"]& \\
&&&&& \Omega^{n-p-1} \\
&\RR \ar[r] \ar[rd] & \Omega^0 \ar[r, "\d"] & \cdots \ar[r, "\d"]& \Omega^{n-p-2} \ar[ur, "-\kappa \d"]&\\
&\ZZ \ar[r] \ar[ur, crossing over, very near end, "-1"'] & \RR &&&
\end{tikzcd}
\end{equation}
where in the bottom left, we send $(x,n) \in \RR \oplus \ZZ$ to $(x-n,x+n) \in \Omega^0 \oplus \RR$.
Observe that this complex~\eqref{eqn: ugly duality} contains
\[
\begin{tikzcd}[row sep = tiny]
\RR \ar[r] & \Omega^0 \ar[r, "\d"] & \cdots \ar[r, "\d"]& \Omega^p \ar[dr, "\star \d"]& \\
&&&& \Omega^{n-p-1} \\
\ZZ \ar[r] & \Omega^0 \ar[r, "\d"] & \cdots \ar[r, "\d"]& \Omega^{n-p-2} \ar[ur, "-\kappa \d"]&
\end{tikzcd}
\]
as a quasi-isomorphic subcomplex.
This complex~\eqref{eqn: ugly duality} looks like $\widetilde{\Mxw}_{p,n}$ coupled to a locally constant sheaf,
which can be viewed as a version of a TFT.

There is a version that is even nicer and manifestly related to BF theory,
by replacing each copy of $\RR$ in the topmost and bottommost rows of~\eqref{eqn: ugly duality} by a de Rham complex (inspired by the Poincar\'e lemma).

\begin{dfn}
Let $\mathrm{MxwBF}_{p,n}\langle\kappa\rangle$ denote the sheaf on $\Riem^\ort_n$ (or $\Lrtz^\ort_n$) valued in $\cDz$ given locally by
\begin{equation}
\begin{tikzcd}[row sep = tiny]
\ZZ \ar[r, hook] \ar[dr, "-1"'] & \Omega^0 \ar[dr, hook] \ar[r, "\d"] & \cdots \ar[r, "\d"] & \Omega^p \ar[dr, "\pm 1"] \ar[r, "\d"] & \Omega^{p+1} \ar[r, "\d"] \ar[rdd, bend left, "\mp \star"'] & \cdots \ar[r, "\d"] & \Omega^n&\\
&\ZZ \ar[r] & \Omega^0 \ar[r, "\d"] & \cdots \ar[r, "\d"]& \Omega^p \ar[dr, "\star \d"]& &&\\
&&&&& \Omega^{n-p-1} &\\
&\RR \ar[r] \ar[dr]  & \Omega^0 \ar[r, "\d"] & \cdots  \ar[r, "\d"]& \Omega^{n-p-2} \ar[ur, "-\kappa \d"] & & &\\
&\ZZ \ar[r] \ar[ur, crossing over, very near end, "-1"'] & \Omega^0 \ar[r,"\d"] & \cdots \ar[r,"\d"] &\Omega^{n-p-2} \ar[r, "\d"] & \Omega^{n-p-1} \ar[r,"\d"]& \cdots \ar[r,"\d"]& \Omega^n
\end{tikzcd}
\end{equation}
where, row by row, $\Omega^{p+1}$, $\Omega^p$,  $\Omega^{n-p-2}$, and $\Omega^{n-p-2}$ are concentrated in degree zero.
\end{dfn}

This complex contains $\widetilde{\Mxw}_{p,n}$ as its middle two rows but has a topological BF theory as its outer rows.
More accurately, the outer rows form a $p+1$-form topological BF theory with {\em discrete} electric charge.
Thus the full complex $\mathrm{MxwBF}_{p,n}\langle\kappa\rangle$ is an interesting coupling of a Maxwell theory to a charge-discretized BF theory.

\begin{rmk}
It is straightforward to write an action functional for this theory, 
via the method described for constructing the action functional~\eqref{eqn: action for mxw},
and it amounts to including a higher-form topological BF theory along with quadratic couplings to the Maxwell theory.
But the action functional does not witness the discretization of the charges;
that has to be spelled out as a constraint,
which is visible in the sheaf description.~\hfill$\Diamond$
\end{rmk}

By construction, $\mathrm{MxwBF}_{p,n}\langle\kappa\rangle$ is quasi-isomorphic to $\widetilde{\Mxw}_{n-p-2,n}\langle 1/\kappa\rangle$.
Hence we have the following form of duality.

\begin{prp}
There is an isomorphism of moduli stacks
\[
\MM\mathrm{xwBF}_{p,n}\langle\kappa\rangle \simeq \sMxw_{n-p-2,n}\langle 1/\kappa\rangle
\]
where the left hand side is the moduli stack presented by~$\mathrm{MxwBF}_{p,n}\langle\kappa\rangle$.
\end{prp}

It would be enlightening to compute the partition function of $\mathrm{MxwBF}_{p,n}\langle\kappa\rangle$ so as to relate to the discussion in~\cite{MooSax}.

\section{Compactification of abelian gauge theories}

We now turn to studying the compactifications (or pushforward) of $n$-dimensional $p$-form gauge theories to produce theories on lower dimensional manifolds.

The basic situation is to study a theory on a product manifold $X \times Y$, 
where $Y$ is closed and oriented.
We might then think of theory as living on~$X$.
Using the sheaf-theoretic language of this paper, 
we have a projection map
\[
\pi: X \times Y \to X
\]
and we want to understand the pushforward~$\pi_* \cF$ of the sheaf $\cF$ on $X \times Y$ that encodes the theory.
(Recall that pushforward for us means the derived pushforward.)

More precisely, we expect that the pushforward is equivalent to theories that naturally occur on $X$.
Here, we expect that a generalized Maxwell theory compactifies to a collection of generalized Maxwell theories,
albeit with abelian groups other than $U(1)$.
Before we state our main result,  
we quickly describe generalized Maxwell theories for higher-dimensional tori.

\subsection{Higher rank abelian gauge theories}

Above we have discussed $U(1)$-gauge theories,
but it is straightforward to generalize to any torus:
we simply take the $r$-fold direct sum of each component in the $U(1)$ case and generalize $\kappa$ from a real number to a $r \times r$ matrix~$K$.

\begin{dfn}
Given a matrix $K \in \GL(r,\RR)$,
let $\Mxw_{p,n}\langle \ZZ^r, K \rangle$ be the sheaf defined on $\Riem^{\ort}_n$ and valued in $\cDz$, 
whose value is
\[
\ZZ^r \to (\Omega^0)^{\oplus r} \xto{\d} \cdots (\Omega^p)^{\oplus r} \xto{K \d \star \d} (\Omega^{n-p})^{\oplus r} \xto{\d} \cdots \xto{\d} (\Omega^n)^{\oplus r}
\] 
on any contractible such manifold.
Its derived global sections $\Mxw_{p,n}\langle \ZZ^r, K \rangle(M)$
presents a derived stack 
\[
\sMxw_{p,n}\langle \ZZ^r, K \rangle(M)
\]
for each Riemannian manifold~$M$.
The sheaf is equipped with a natural local $(-1)$-symplectic structure. 
\end{dfn}

We refer to $\Mxw_{p,n}\langle \ZZ^r,K\rangle$ as {\em $p$-form gauge theory with group $U(1)^r$} and coupling constant $K \in \GL(r,\RR)$.
We have the charge-discretized version too.

\begin{dfn}
Given a matrix $E \in \GL(r,\RR)$,
{\em the charge-discretized $p$-form gauge theory with group $U(1)^r$}, denoted by $\Mxw^\circ_{p,n}\langle \ZZ^r, K,E \rangle$, is encoded by the sheaf defined on $\Riem^{\ort}_n$ and valued in $\cDz$, 
whose value is
\[
\begin{tikzcd}
\ZZ^r \ar[r] & (\Omega^0)^{\oplus r} \ar[r, "\d"] & \cdots \ar[r, "\d"]& (\Omega^p)^{\oplus r} \ar[dr, "\star \d"]& \\
&&&& (\Omega^{n-p-1})^{\oplus r} \\
\ZZ^r \ar[r, "E"] & (\Omega^0)^{\oplus r} \ar[r, "\d"] & \cdots \ar[r, "\d"]& (\Omega^{n-p-2})^{\oplus r} \ar[ur, "-K \d"']&
\end{tikzcd}
\]
on any contractible such manifold.
Its derived global sections on $M$ presents a derived stack $\sMxw^\circ_{p,n}\langle \ZZ^r, K,E \rangle(M)$
that is $-1$-presymplectic on a closed Riemannian manifold~$M$.
\end{dfn}

Just as earlier, there are isomorphisms
\[
\Mxw^\circ_{p,n}\langle \ZZ^r, K,E \rangle \cong \Mxw^\circ_{p,n}\langle \ZZ^r, KE,1 \rangle
\]
and dualities
\[
\Mxw^\circ_{p,n}\langle \ZZ^r, K,E \rangle \cong \Mxw^\circ_{n-p-2,n}\langle \ZZ^r, K^{-1}, E^{-1} \rangle
\]
by parallel arguments.
(In fact, the situation gets more interesting because we can ask about symmetries arising from $\GL(r,\ZZ)$ acting on the ``discrete charges.'')

We note that, similarly, we can introduce a higher rank variant of the Deligne complex of Definition~\ref{Deligne cplx}.

\begin{dfn}
Given a matrix $L \in \GL(r,\RR)$,
the {\em Deligne complex for torus $\RR^r/L(\ZZ^r)$}, 
denoted by $\Del_n\langle \ZZ^r, L \rangle$, 
is encoded by the sheaf defined on $\Riem^{\ort}_n$ and valued in $\cDz$ 
whose value is
\[
\ZZ^{\oplus r}[1] \xto{L} (\Omega^0)^{\oplus r} \xto{\d} \cdots \xto{\d} (\Omega^n)^{\oplus r}
\]
on any contractible such manifold.
\end{dfn}

\subsection{The main result and some special cases}

Recall that we can view charge-discretized Maxwell theory $\Mxw_{p,n}^\circ\langle\kappa,e\rangle$ as assembled from a topological theory and a perturbative theory,
following Section~\ref{sec: on lattices}.
In explicit terms, we have 
\[
0 \to \widetilde{\Mxw}^{\pert}_{p,n} \to \Mxw_{p,n}^\circ\langle\kappa,e\rangle \to \ZZ[p+1] \oplus \ZZ[n-p-1] \to 0
\]
already seen as the short exact sequence~\eqref{SES for circ mxw}.
To describe the pushforward $\pi_* \Mxw_{p,n}^\circ\langle\kappa,e\rangle$,
we can work out the pushforwards of the topological and perturbative sheaves and then assemble them.
(We discuss how this process works in Section~\ref{sec: comp by triangles}.)
Describing a pushforward like $\pi_* \ZZ$ is a classic problem of topology,
and it is quasi-isomorphic to the constant sheaf valued in~$H^\bullet(Y,\ZZ)$.
Describing the pushforward of the perturbative theory in convenient terms involves Hodge theory and the homological perturbation lemma; 
we undertake it in Section~\ref{sec: comp by HPL}.

When these ingredients are assembled, we obtain the result stated below.
It involves some interplay between the dimensions of the manifolds $X$ and $Y$ and the value of $p$,
so the statement is somewhat involved. 
Below we unpack this general statement for 4-dimensional Maxwell theories,
running over all possible cases.

\begin{thm}
Let $X$ be an oriented Riemannian manifold of dimension~$x$ and let $Y$ be a closed, oriented Riemannian manifold of dimension~$y$.
Let $\pi: X \times Y \to X$ be the projection.
Let $E_k: H^k_\free(Y,\ZZ) \to H^k(Y,\RR)$ denote the inclusion of integral cohomology into real cohomology.

The pushforward sheaf $\pi_* \Mxw^\circ_{p,n}\langle\kappa,e\rangle$ on $X$ is quasi-isomorphic to the direct sum of four sheaves
\begin{itemize}
\item $\bigoplus_{\ell = \max(p-x+1,0)}^{\min(y,p)} \Mxw^\circ_{p-\ell, x}\langle H^\ell(Y,\ZZ), \kappa\, {\rm Id}, e E_\ell\rangle$ 
\item $\bigoplus_{\ell = 0}^{p-x} \Del_x\langle H^\ell(Y,\ZZ), E_\ell \rangle[p-\ell]$ 
\item $\bigoplus_{k = 0}^{y-p-1} \Del_x\langle H^k(Y,\ZZ), E_k \rangle[n-p-2-k] $
\item $H^\bullet_\tors(Y,\ZZ)[p+1] \oplus H^\bullet_\tors(Y,\ZZ)[n-p-1]$
\end{itemize}
where $H^\bullet_\tors(Y,\ZZ)$ denotes the torsion subgroup of the integral cohomology of~$Y$.
\end{thm}

The first contribution is a product of generalized Maxwell theories on~$X$ of higher rank.

\begin{rmk}
Our results lift to statements about the derived stacks presented by these complexes because our methods work in families of test objects:
we are not varying the metric on $Y$ and hence need not grapple with subtleties about the Hodge decomposition varying in families.\hfill$\Diamond$
\end{rmk}

\begin{rmk}
We have considered products of Riemannian manifolds here,
but the techniques can accommodate fiber bundles $M \to X$ with compact fibers\footnote{That is, a proper surjective submersion.}
where the metric varies 
(such as warped products, for a simple type of example).
The arguments involve more elaborate bookkeeping for the perturbative factor,
as one must deal with Hodge decompositions for the fibers that vary along the base.
The theories on $X$ then have coupling ``constants'' that depend in an interesting way on the metric of the total space.~\hfill$\Diamond$
\end{rmk}

We now consider the situation where $\dim(X \times Y) = 4$.
There are six cases to examine:
we let $Y$ have dimensions 1, 2, or~3 and let $p = 0$ or~1,
because the $2$-form theory is dual to $0$-form theory.
Note that torsion in cohomology can only appear when~$y = 3$, 
since $S^1$ and oriented surfaces have no torsion.

\begin{eg}[$x+y = 4$, $p=0$, $y=1$]
The sheaf on the 3-manifold~$X$ is locally given by evaluating the following
\[
\begin{tikzcd}[row sep = tiny]
&&H^0(Y,\ZZ) \ar[r] & \Omega^0_X \otimes H^0(Y,\RR) \ar[dr, "\star \d"]& \\
&&&& \Omega^{2}_X \otimes H^1(Y,\RR) \\
&H^1(Y,\ZZ) \ar[r,"e"] & \Omega^0_X \otimes H^1(Y,\RR) \ar[r, "\d"] &  \Omega^{1}_X \otimes H^1(Y,\RR) \ar[ur, "-\kappa \d"']&\\
H^0(Y,\ZZ) \ar[r,"e"] & \Omega^0_X \otimes H^0(Y,\RR)\ar[r, "\d"] &  \Omega^{1}_X\otimes H^0(Y,\RR) \ar[r, "\d"] &  \Omega^{2}_X \otimes H^0(Y,\RR)\ar[r, "-\kappa\d"] &  \Omega^{3}_X\otimes H^0(Y,\RR)
\end{tikzcd}
\]
where in degree zero we have
\[
\Omega^0_X \otimes H^0(Y,\RR) 
\oplus \Omega^{1}_X \otimes H^1(Y,\RR)
\oplus \Omega^{2}_X \otimes H^0(Y,\RR).
\]
Note that $\star$ means $\star_X \otimes \star_Y$ (up to sign) and $\d$ denotes the exterior derivative on~$X$.~\hfill$\Diamond$
\end{eg}

\begin{eg}[$x+y = 4$, $p=0$, $y=2$]
The sheaf on the 2-manifold~$X$ is locally given by evaluating the following
\[
\begin{tikzcd}[row sep = tiny]
&&H^0(Y,\ZZ) \ar[r] & \Omega^0_X \otimes H^0(Y,\RR) \ar[dr, "\star \d"]& \\
&&&& \Omega^{1}_X \otimes H^2(Y,\RR) \\
&&H^2(Y,\ZZ) \ar[r,"e"] & \Omega^0_X \otimes H^2(Y,\RR) \ar[ur, "-\kappa \d"']&\\
H^0(Y,\ZZ) \ar[r,"e"] & \Omega^0_X \otimes H^0(Y,\RR)\ar[r, "\d"] &  \Omega^{1}_X\otimes H^0(Y,\RR) \ar[r, "\d"] &  \Omega^{2}_X \otimes H^0(Y,\RR)\\
&H^1(Y,\ZZ) \ar[r,"e"] & \Omega^0_X \otimes H^1(Y,\RR)\ar[r, "\d"] &  \Omega^{1}_X\otimes H^1(Y,\RR) \ar[r, "-\kappa\d"] &  \Omega^{2}_X \otimes H^1(Y,\RR)
\end{tikzcd}
\]
where in degree zero we have
\[
\Omega^0_X \otimes H^0(Y,\RR) 
\oplus \Omega^{0}_X \otimes H^2(Y,\RR)
\oplus \Omega^{2}_X \otimes H^0(Y,\RR)
\oplus \Omega^{1}_X \otimes H^1(Y,\RR).
\]
Note that $\star$ means $\star_X \otimes \star_Y$ (up to sign) and $\d$ denotes the exterior derivative on~$X$.~\hfill$\Diamond$
\end{eg}

\begin{eg}[$x+y = 4$, $p=0$, $y=3$]
The sheaf on the 1-manifold~$X$ is locally given by evaluating the following
\[
\begin{tikzcd}[row sep = tiny]
&&H^0(Y,\ZZ) \ar[r] & \Omega^0_X \otimes H^0(Y,\RR) \ar[dr, "\star \d"]& \\
&&&& \Omega^{0}_X \otimes H^3(Y,\RR) \\
H^0(Y,\ZZ) \ar[r,"e"] & \Omega^0_X \otimes H^0(Y,\RR)\ar[r, "\d"] &  \Omega^{1}_X\otimes H^0(Y,\RR)\\
&H^1(Y,\ZZ) \ar[r,"e"] & \Omega^0_X \otimes H^1(Y,\RR)\ar[r, "\d"] &  \Omega^{1}_X\otimes H^1(Y,\RR)\\
&&H^2(Y,\ZZ) \ar[r,"e"] & \Omega^0_X \otimes H^2(Y,\RR)\ar[r, "-\kappa\d"] &  \Omega^{1}_X\otimes H^2(Y,\RR)
\end{tikzcd}
\]
where in degree zero we have
\[
\Omega^0_X \otimes H^0(Y,\RR) 
\oplus \Omega^{1}_X \otimes H^1(Y,\RR)
\oplus \Omega^{0}_X \otimes H^2(Y,\RR).
\]
Note that the factor $H^2(Y,\ZZ)$ might have torsion!
Note that $\star$ means $\star_X \otimes \star_Y$ (up to sign) and $\d$ denotes the exterior derivative on~$X$.~\hfill$\Diamond$
\end{eg}

\begin{eg}[$x+y = 4$, $p=1$, $y=1$]
The sheaf on the 3-manifold~$X$ is locally given by evaluating the following
\[
\begin{tikzcd}[row sep = tiny]
&H^0(Y,\ZZ) \ar[r] & \Omega^0_X \otimes H^0(Y,\RR) \ar[r, "\d"] & \Omega^1_X \otimes H^0(Y,\RR) \ar[dr, "\star \d"]& \\
&&&& \Omega^{1}_X \otimes H^1(Y,\RR) \\
&&H^1(Y,\ZZ) \ar[r,"e"] & \Omega^0_X \otimes H^1(Y,\RR) \ar[ur, "-\kappa \d"']&\\
&&H^1(Y,\ZZ) \ar[r] & \Omega^0_X \otimes H^1(Y,\RR) \ar[dr, "\star \d"]&\\
&&&& \Omega^{2}_X \otimes H^0(Y,\RR) \\
&H^0(Y,\ZZ) \ar[r,"e"] & \Omega^0_X \otimes H^0(Y,\RR) \ar[r, "\d"] & \Omega^1_X \otimes H^0(Y,\RR) \ar[ur, "-\kappa \d"']&
\end{tikzcd}
\]
where in degree zero we have
\[
\Omega^1_X \otimes H^0(Y,\RR) 
\oplus \Omega^{0}_X \otimes H^1(Y,\RR)
\oplus \Omega^{0}_X \otimes H^1(Y,\RR)
\oplus \Omega^1_X \otimes H^0(Y,\RR).
\]
Note that $\star$ means $\star_X \otimes \star_Y$ (up to sign) and $\d$ denotes the exterior derivative on~$X$.~\hfill$\Diamond$
\end{eg}

\begin{eg}[$x+y = 4$, $p=1$, $y=2$]
The sheaf on the 2-manifold~$X$ is locally given by evaluating the following
\[
\begin{tikzcd}[row sep = tiny]
&H^0(Y,\ZZ) \ar[r] & \Omega^0_X \otimes H^0(Y,\RR) \ar[r, "\d"] & \Omega^1_X \otimes H^0(Y,\RR) \ar[dr, "\star \d"]& \\
&&& H^2(Y,\ZZ) \ar[r] &\Omega^{0}_X \otimes H^2(Y,\RR)\\
&&H^1(Y,\ZZ) \ar[r] & \Omega^0_X \otimes H^1(Y,\RR) \ar[dr, "\star \d"]&\\
&&&& \Omega^{1}_X \otimes H^1(Y,\RR) \\
&&H^1(Y,\ZZ) \ar[r,"e"] & \Omega^0_X \otimes H^1(Y,\RR) \ar[ur, "-\kappa \d"']&\\
&H^0(Y,\ZZ) \ar[r,"e"] & \Omega^0_X \otimes H^0(Y,\RR) \ar[r, "\d"] & \Omega^1_X \otimes H^0(Y,\RR) \ar[r, "-\kappa\d"]& \Omega^2_X \otimes H^0(Y,\RR)
\end{tikzcd}
\]
where in degree zero we have
\[
\Omega^1_X \otimes H^0(Y,\RR) 
\oplus \Omega^{0}_X \otimes H^1(Y,\RR)
\oplus \Omega^{0}_X \otimes H^1(Y,\RR)
\oplus \Omega^1_X \otimes H^0(Y,\RR).
\]
Note that $\star$ means $\star_X \otimes \star_Y$ (up to sign) and $\d$ denotes the exterior derivative on~$X$.~\hfill$\Diamond$
\end{eg}

\begin{eg}[$x+y = 4$, $p=1$, $y=3$]
The sheaf on the 1-manifold~$X$ is locally given by evaluating the following
\[
\begin{tikzcd}[row sep = tiny]
&&H^1(Y,\ZZ) \ar[r] & \Omega^0_X \otimes H^1(Y,\RR) \ar[dr, "\star \d"]&\\
&&&H^2(Y,\ZZ) \ar[r] &\Omega^{0}_X \otimes H^2(Y,\RR) \\
&H^0(Y,\ZZ) \ar[r] & \Omega^0_X \otimes H^0(Y,\RR) \ar[r, "\d"] & \Omega^1_X \otimes H^0(Y,\RR) & \\
&H^0(Y,\ZZ) \ar[r,"e"] & \Omega^0_X \otimes H^0(Y,\RR) \ar[r, "\d"] & \Omega^1_X \otimes H^0(Y,\RR) & \\
&&H^1(Y,\ZZ) \ar[r,"e"] & \Omega^0_X \otimes H^1(Y,\RR) \ar[r, "-\kappa\d"]&\Omega^{0}_X \otimes H^1(Y,\RR)
\end{tikzcd}
\]
where in degree zero we have
\[
\Omega^1_X \otimes H^0(Y,\RR) 
\oplus \Omega^{0}_X \otimes H^1(Y,\RR)
\oplus \Omega^{0}_X \otimes H^1(Y,\RR)
\oplus \Omega^1_X \otimes H^0(Y,\RR).
\]
Note that $\star$ means $\star_X \otimes \star_Y$ (up to sign) and $\d$ denotes the exterior derivative on~$X$.~\hfill$\Diamond$
\end{eg}

\subsection{Compactification, via triangulation}
\label{sec: comp by triangles}

Roughly speaking, a compactification of a $p$-form theory is a product of abelian $k$-form gauge theories with $k$ running from $0$ to $p$.
That is, if $\pi: X \times Y \to X$ is the projection, 
then the pushforward $\pi_* \Mxw^\circ_{p,n}\langle\kappa,e\rangle$ is approximately
\begin{equation}
\label{eqn: rough approx}
\Mxw^\circ_{0,k}\langle H^p_{\free}(Y), \kappa, eE_p\rangle \times \Mxw^\circ_{1,k}\langle H^{p-1}_{\free}(Y), \kappa, eE_{p-1}\rangle \times \cdots \times \Mxw^\circ_{p,k}\langle H^0_{\free}(Y), \kappa, eE_0\rangle
\end{equation}
where we write $H^k_{\free}(Y)$ for the integral lattice given by the {\em free} part of the abelian group~$H^k(Y,\ZZ)$ and where $E_k$ encodes the inclusion of this lattice into the real cohomology group.

In~\eqref{eqn: rough approx}, we are missing the contribution of the torsion in the cohomology of~$Y$.
This torsion contributes Dijkgraaf-Witten-like topological gauge theories that couple to the gauge theories appearing on the right hand side of~\eqref{eqn: rough approx}.

Before jumping into details,
we point out what can be seen from general principles.
We will describe here the situation with $\Mxw^\circ_{p,n}\langle \kappa, e\rangle$,
but the same techniques work for any variation on generalized Maxwell theory.

Note that any short exact sequence of sheaves on $M = X \times Y$ determines a distinguished triangle in the derived category of sheaves on~$M$,
and the (derived) pushforward along $\pi$ of this distinguished triangle is a distinguished triangle in the derived category of sheaves on~$X$.

We now apply that observation to 
\[
0 \to \widetilde{\Mxw}^{\pert}_{p,n} \to \Mxw_{p,n}^\circ\langle\kappa,e\rangle \to \ZZ[p+1] \oplus \ZZ[n-p-1] \to 0
\]
the short exact sequence~\eqref{SES for circ mxw} from Section~\ref{sec: on lattices}.
We obtain a distinguished triangle
\[
\pi_* \widetilde{\Mxw}^{\pert}_{p,n} \to \pi_* \Mxw^\circ_{p,n}\langle\kappa,e\rangle \to \pi_* \left(\ZZ[p+1]\oplus \ZZ[n-p-1]\right) \to \pi_* \widetilde{\Mxw}^{\pert}_{p,n}[1]
\]
in the derived category of sheaves for~$X$.
(Recall that we use $\pi_*$ to denote the derived pushforward.)

Now the sheaf $\pi_* \ZZ[p+1]$ is well-studied in algebraic topology,
where it is typically studied via the Leray--Serre spectral sequence.
Recall that $\pi_* \ZZ$ is a locally constant sheaf whose value on a contractible open $U \subset X$ is quasi-isomorphic to the singular cochain complex~$C^\bullet(Y,\ZZ)$. 
But it is a classical fact\footnote{For a proof, see \cite[Corollary 13.1.20]{KasSch}.} that $C^\bullet(Y,\ZZ)$ is quasi-isomorphic to its cohomology $H^\bullet(Y,\ZZ)$, 
so we find
\[
\pi_* \ZZ[p+1] \simeq H^\bullet(Y,\ZZ)[p+1]
\]
where the right hand side is viewed as a locally constant sheaf on $X$.
In short, we have good control of this term in the distinguished triangle.

The sheaf $\pi_* \Mxw^{\pert}_{p,n}$ is also well-understood:
\begin{itemize}
\item it is built from smooth sections of vector bundles, 
so the derived pushforward agrees with the naive pushforward and
\item so long as $Y$ is closed and oriented, Hodge theory allows us, in essence, to replace differential forms on $Y$ with harmonic forms on~$Y$.\footnote{Using $L^2$ cohomology, one can work with $Y$ not closed and our arguments below carry over nicely.}
\end{itemize}
In consequence we can view $\pi_*\Mxw^{\pert}_{p,n}$ as equivalent to a field theory on~$X$. 
We will see that it is, in essence, a product of perturbative Maxwell theories on~$X$ itself.

Putting our observations together,
we see that 
\[
\pi_* \Mxw_{p,n}^\circ\langle\kappa,e\rangle
\]
can be described as the (co)cone of a sheaf map 
\[
H^\bullet(Y,\ZZ)[p+1] \oplus H^\bullet(Y,\ZZ)[n-p-1] \to \pi_* \widetilde{\Mxw}^{\pert}_{p,n}
\]
on~$X$.\footnote{This situation is similar to the setting of variations of Hodge structures,
as we are working with a cousin of the Deligne complexes.
Many results appearing in that subject ought to admit analogues here. 
In the setting of holomorphically twisted supersymmetric theories,
our methods become extremely closely intertwined with variations of Hodge structures.}
Since the codomain of that map is built of vector spaces over~$\RR$,
it cannot receive nonzero maps from torsion groups, 
so the map factors through the free part $H^\bullet_\free(Y,\ZZ)[p+1]$ of integral cohomology.

Thus, roughly, the compactification is the product of two kinds of theories:
\begin{itemize}
\item a {\em topological} field theory arising from torsion in the integral cohomology~$H^\bullet(Y,\ZZ)$ and
\item a nonperturbative generalized Maxwell theories, determined by the free part of the integral cohomology~$H^\bullet(Y,\ZZ)$.
\end{itemize}
We now turn to pinning down the second factor precisely.

\begin{rmk}
A general mechanism we have not exploited here is the Leray spectral sequence,
but our arguments above make this rather computable,
inasmuch as one knows the Hodge theory and integral cohomology of~$Y$.
It just gives less sharp statements than we can obtain via homological perturbation lemma.~\hfill$\Diamond$
\end{rmk}

\subsection{Compactification, via homological perturbation}
\label{sec: comp by HPL}

We now work out a convenient model for $\pi_* \widetilde{\Mxw}^\pert$ using the homological perturbation lemma,
which is reviewed in Section~\ref{sec: HPL}.
This section provides a proof of the following result.

\begin{prp}
\label{prp: pert compactification}
Let $X$ be an oriented Riemannian manifold of dimension~$x$ and let $Y$ be closed, oriented Riemannian manifold of dimension~$y$.

The pushforward $\pi_* \widetilde{\Mxw}^\pert$ is quasi-isomorphic to the sheaf {\em on~$X$} 
\[
\bigoplus_{\ell = \max(p-x+1,0)}^{\min(y,p)} \widetilde{\Mxw}^\pert_{p-\ell, x} \otimes H^\ell(Y) 
\oplus \bigoplus_{\ell = 0}^{\min(p-x,y)} \Omega^\bullet_X \otimes H^\ell(Y)[p-\ell]
\oplus \bigoplus_{k = 0}^{y-p-1} \Omega^\bullet_X \otimes H^k(Y)[n-p-2-k]
\]
describing a product of perturbative generalized Maxwell theories {\em on $X$} and some de Rham complexes on $X$ valued in the cohomology of~$Y$.
\end{prp}

In this section we use $H^\ell(Y)$ to denote the de Rham cohomology $H^\ell(Y,\RR)$.
Note that we get as copies of the $p-\ell$-form gauge theory on $X$ with multiplicity~$\dim(H^\ell(Y))$.

\subsubsection{Strategy of proof}

The perturbative complex is the truncation of $\Mxw^\circ_{p,n}$ where we slice off the copies of $\ZZ$ on the far left.
That is, consider the complex
\[
\begin{tikzcd}[row sep = tiny]
\Omega^0 \ar[r, "\d"] & \cdots \ar[r, "\d"]& \Omega^p \ar[dr, "\star \d"]& \\
&&& \Omega^{n-p-1} \\
\Omega^0 \ar[r, "\d"] & \cdots \ar[r, "\d"]& \Omega^{n-p-2} \ar[ur, "-\kappa \d"]&
\end{tikzcd}.
\]
It can be seen as a deformation of the complex
\[
\begin{tikzcd}[row sep = tiny]
\Omega^0 \ar[r, "\d"] & \cdots \ar[r, "\d"]& \Omega^p & \\
\Omega^0 \ar[r, "\d"] & \cdots \ar[r, "\d"]& \Omega^{n-p-2} \ar[r, "-\kappa \d"]&\Omega^{n-p-1}
\end{tikzcd},
\]
which is a direct sum of two truncated de Rham complexes.\footnote{Bear in mind that the copies of $\Omega^0$ do not (typically) sit in the same degree, even if the diagram might suggest that. On the other hand, $\Omega^p$ and $\Omega^{n-p-2}$ {\em do} both sit in degree zero.}
This deformation by $\star \d$ is small,
in the sense of the homological perturbation lemma.

Hence, we break up the proof into a few steps:
\begin{enumerate}
\item We analyze pushforwards of truncated de Rham complexes, using a combination of Hodge theory and homological perturbation theory.
We find that this pushforward is quasi-isomorphic to a small complex with two components: 
a piece that depends on cohomology $H^\ell(Y)$ and a piece that depends on the coclosed forms 
\[
K^\ell(Y) = {\rm im}(\d^\star_Y) \subset \Omega^\ell(Y),
\]
due to the Hodge decomposition.
\item We view $\widetilde{\Mxw}^\pert$ as a perturbation of a sum of two truncated de Rham complexes, 
and we apply the homological perturbation lemma to find a small model built from these kinds of two components.
\item We show that the first component, 
which involves the cohomology~$H^\bullet(Y)$, 
yields generalized Maxwell theories on~$X$,
along with some topological theories
\item We analyze the second component, 
which involves the coclosed forms on $Y$,
and we show it is acyclic, 
so it does not contribute (nontrivially) to our main result.
\end{enumerate}
The arguments are lengthy but straightforward manipulations,
so we go through them discursively in the two subsections below.
The conclusion is summarized as follows.

\subsubsection{Half of the perturbative complex}

As a first step towards the perturbative Maxwell theory, consider a truncated de Rham complex on $X \times Y$:
\[
\Omega^0 \xto{\d} \cdots \xto{\d} \Omega^k.
\]
Note that differential forms on $X \times Y$ admit a decomposition
\begin{equation}
\label{bigrading}
\Omega^k(X\times Y) =  \bigoplus_{j =0}^k \Omega^j(X) \,\widehat{\otimes} \, \Omega^{k-j} (Y) 
\end{equation}
where $\widehat{\otimes}$ denotes the completed projective tensor product with respect to the standard Fr\'echet topology on these vector spaces.
We abbreviate these summands as
\[
\Omega^{[j,l]} =  \Omega^j(X) \,\widehat{\otimes} \, \Omega^l(Y)
\]
for concision.
In this subsection, we will view $k$-forms as having cohomological degree~$k$, 
so that the total degree of an element in $\Omega^{[j,l]}$ is~$j+l$.

In terms of the bigrading above, 
the truncated de Rham complex is the totalization of a double complex
\[
\begin{tikzcd}
0 & \cdots&&&&\\
\Omega^{[0,k]} 
& 0  
& \cdots 
&&&\\
\Omega^{[0,k-1]} \ar[r, "\d_X"] \ar[u, "\d_Y"]
& \Omega^{[1,k-1]}  
& 0
& \cdots
&&\\
\vdots \ar[u, "\d_Y"]
& \vdots \ar[u, "\d_Y"]  
& \ddots & 0
&
\cdots
&\\
\Omega^{[0,1]} \ar[r, "\d_X"] \ar[u, "\d_Y"]
& \Omega^{[1,1]} \ar[r, "\d_X"] \ar[u, "\d_Y"]
& \cdots \ar[r, "\d_X"] 
& \Omega^{[k-1,1]} 
& 0 & \cdots\\
\Omega^{[0,0]} \ar[u, "\d_Y"] \ar[r, "\d_X"] 
& \Omega^{[1,0]} \ar[r, "\d_X"] \ar[u, "\d_Y"] 
& \cdots \ar[r, "\d_X"] 
& \Omega^{[k-1,0]} \ar[r, "\d_X"] \ar[u, "\d_Y"] 
& \Omega^{[k,0]} 
& 0 & \cdots
\end{tikzcd}
\]
Each column corresponds to a truncated de Rham complex for $Y$ tensored with some vector space $\Omega^j(X)$.

As $Y$ is closed, Hodge theory provides an explicit decomposition
\[
\Omega^k(Y) = H^k(Y) \oplus {\rm im}(\d)^k \oplus {\rm im}(\d^\star)^k
\]
where $\d^\star = \pm \star \d \star$ (the sign depends on the form degree and the metric signature) and $H^k(Y)$ here consists of the harmonic $k$-forms.
We include the superscript $k$ on the images to make the form-degree manifest.
This decomposition determines a deformation retract for the truncated de Rham complex
\[
\Omega^0(Y) \xto{\d_Y} \cdots \xto{\d_Y} \Omega^{\ell-1}(Y) \xto{\d_Y} \Omega^\ell(Y)
\]
onto its cohomology
\[
H^0(Y) \xto{0} \cdots \xto{0} H^{\ell-1}(Y) \xto{0}  \Omega^\ell(Y)/\d_Y(\Omega^{\ell-1}(Y)).
\]
We can (and do) require that the cochain homotopy $\eta$ vanishes on each summand~$H^k(Y)$.
Note the rightmost term is very large, 
but it decomposes as
\[
\Omega^\ell(Y)/\d_Y(\Omega^{\ell-1}(Y)) = H^\ell(Y) \oplus {\rm im}(\d_Y^\star)^\ell
\]
where ${\rm im}(\d_Y^\star)^\ell$ denotes the image of $\d^\star_Y$ in~$\Omega^\ell(Y)$.
Set 
\begin{equation}
\label{coclosed forms}
K^\ell(Y) = {\rm im}(\d_Y^\star)^\ell
\end{equation}
for the sake of concise notation.
Note that~$K^y(Y) = 0$.

Tensoring with $\Omega^j(X)$,
we obtain a deformation complex for the cochain complex $\Omega^j(X) \widehat{\otimes}\, \Omega^{\leq \ell}(Y)$ as well.
Corollary~\ref{crl: hpl for double} (of the homological perturbation lemma) then guarantees that we have a homotopy equivalence
\[
\begin{tikzcd}
(H^\bullet(\Omega^{\leq k}(X \times Y),\d_Y), \pi \alpha \iota) \ar[r,shift left, "\iota + \eta \alpha \iota"] & (\Omega^{\leq k}(X \times Y), \d_Y + \d_X) \ar[l, shift left, "\pi + \pi \alpha \eta"] \ar[loop right, "\eta + \eta \alpha \eta"] 
\end{tikzcd}
\]
where 
\[
\alpha = (1-\d_X \eta)^{-1} \d_X = \sum_{n \geq 0} (\d_X \eta)^n \d_X
\] 
and $\eta$ is the cochain homotopy from the Hodge theory for~$Y$.
In explicit terms, the left hand complex is the totalization of the bigraded vector space
\[
\begin{tikzcd}[column sep = tiny]
\begin{matrix} 
\Omega^{0}(X) \otimes H^k(Y) \\ 
\,\oplus \, \Omega^{0}(X) \widehat{\otimes} \, K^k(Y)
\end{matrix}
& 0  
& \cdots 
&&&\\
\Omega^{0}(X) \otimes H^{k-1}(Y)  
& 
\begin{matrix} 
\Omega^{1}(X) \otimes H^{k-1}(Y) \\ 
\,\oplus \, \Omega^{1}(X) \widehat{\otimes} \, K^{k-1}(Y)
\end{matrix}
& 0
& \cdots
&&\\
\vdots 
& \vdots   
& \ddots & 0
&
\cdots
&\\
\Omega^{0}(X) \otimes H^{1}(Y)  
& \Omega^{1}(X) \otimes H^{1}(Y)  
& \cdots  
& 
\begin{matrix} 
\Omega^{k-1}(X) \otimes H^{1}(Y) \\ 
\oplus \, \Omega^{k-1}(X) \widehat{\otimes} \, K^1(Y)
\end{matrix}
& 0 & \cdots\\
\Omega^{0}(X) \otimes H^{0}(Y)   
& \Omega^{1}(X) \otimes H^{0}(Y)   
& \cdots 
& \Omega^{k-1}(X) \otimes H^{0}(Y)  
& \Omega^{k}(X) \widehat{\otimes} \, H^{0}(Y)
& 0 & \cdots
\end{tikzcd}
\]
with the transferred differential~$\pi\alpha\iota$ (given by ``descending the staircase'').
But because $\eta$ vanishes on the summand $\Omega^j(X) \widehat{\otimes} H^{k-j}(Y)$,\footnote{The notation $\widehat{\otimes}$ denotes the completed projective tensor product of these Fr\'echet vector spaces. Recall that $C^\infty(X \times Y) = C^\infty(X) \widehat{\otimes}\, C^\infty(Y)$ for any smooth manifolds $X$ and $Y$, and a similar statement holds for differential forms.}
this differential simplifies to~$\pi \d_X \iota$,
i.e., only the horizontal differential survives.

We record the conclusion.

\begin{lmm}
\label{lmm: decomp for trunc DR}
When $Y$ is closed and oriented of dimension~$y$,
the truncated de Rham complex
\[
(\Omega^{\leq k}(X \times Y), \d_Y + \d_X)
\]
is quasi-isomorphic to the direct sum of cochain complexes
\begin{equation}
\label{decomp for trunc DR}
C_H \oplus C_K
\end{equation}
where
\[
C_H = \left(\bigoplus_{\ell = 0}^{\min(y,k)} \Omega^{\leq k-\ell}(X) \otimes H^\ell(Y) \; , \; \d_X \otimes \id_{H^\ell(Y)}\right)
\]
and
\[
C_K = \bigoplus_{\ell = 0}^{\min(y-1,k)} \Omega^{k - \ell}(X) \widehat{\otimes}\, K^\ell(Y).
\]
Note that $C_K$ is concentrated wholly in cohomological degree~$k$ while $C_H$ lives between degrees 0 and~$k$.
\end{lmm}

The first upper bound is $\min(y,k)$ because $H^\ell(Y) = 0$ if $\ell > y$.
Note that $K^y(Y) = 0$, yielding the second upper bound.

\subsubsection{The whole perturbative complex}

Recall that the perturbative complex
\begin{equation}
\label{cplx: mxwcircpert}
\begin{tikzcd}[row sep = tiny]
&\Omega^0 \ar[r, "\d"] & \cdots \ar[r, "\d"]& \Omega^p \ar[dr, "\star \d"]& \\
&&&& \Omega^{n-p-1} \\
\Omega^0 \ar[r, "\d"] & \cdots & \cdots \ar[r, "\d"]& \Omega^{n-p-2} \ar[ur, "-\kappa \d"]&
\end{tikzcd}
\end{equation}
can be seen as a small deformation of the complex
\begin{equation}
\label{cplx: sum of truncs}
\begin{tikzcd}[row sep = tiny]
& \Omega^0 \ar[r, "\d"] & \cdots \ar[r, "\d"]& \Omega^p & \\
\Omega^0 \ar[r, "\d"] & \cdots & \cdots \ar[r, "\d"]& \Omega^{n-p-2} \ar[r, "-\kappa \d"]&\Omega^{n-p-1}
\end{tikzcd}
\end{equation}
which is a direct sum of two truncated de Rham complexes.

The complex~\eqref{cplx: sum of truncs} admits a simple description by our work in the preceding section:
it admits a homotopy equivalence with a direct sum $T \oplus B$ where
\begin{equation}
\label{cplx: top row}
T = \bigoplus_{\ell = 0}^{\min(y,p)} (\Omega^{\leq p-\ell}(X) \otimes H^\ell(Y)[p], \d_X) 
\oplus  \bigoplus_{\ell = 0}^{\min(y-1,p)} \Omega^{p-\ell}(X)[p] \widehat{\otimes}\, K^\ell(Y)
\end{equation}
and
\begin{equation}
\label{cplx: bottom row}
B = \bigoplus_{k = 0}^{\min(y,n-p-1)} (\Omega^{\leq (n-p-1)-k}(X) \otimes H^k(Y)[n-p-2], \d_X) 
\oplus  
\bigoplus_{k = 0}^{\min(y-1,n-p-1)} \Omega^{(n-p-1)-k}(X) \widehat{\otimes}\, K^k(Y)[n-p-2]
\end{equation}
by using the description of~\eqref{decomp for trunc DR} for the top and bottom rows of~\eqref{cplx: sum of truncs}.\footnote{Note that we included downward shifts, relative to the preceding section where we used the total degree of the differential forms.}

Recall the further decomposition $T = T_H \oplus T_K$ where
\[
T_H = \bigoplus_{\ell = 0}^{\min(y,p)} (\Omega^{\leq p-\ell}(X) \otimes H^\ell(Y), \d_X)
\]
and
\[
T_K = \bigoplus_{\ell = 0}^{\min(y-1,p)} \Omega^{p-\ell}(X) \widehat{\otimes}\, K^\ell(Y).
\]
Likewise, we decompose~$B = B_H \oplus B_K$.

Since the complex~\eqref{cplx: mxwcircpert} is a small deformation of~\eqref{cplx: sum of truncs},
it remains to apply the homological perturbation lemma to obtain a simple, smaller replacement of the complex~\eqref{cplx: mxwcircpert} of fundamental interest.

Here we need to describe $\star \d$ in terms of the bigrading of~\eqref{bigrading}.
It suffices to understand $\star$ on~$X \times Y$.
Let $x = \dim(X)$ and $y = \dim(Y)$ so that $n = x + y$.
Then $\star$ sends $\Omega^{[j,k]}$ to $\Omega^{[x-j, y-k]}$,
and for any $\alpha \in \Omega^j(X)$ and $\beta \in \Omega^k(Y)$, one finds
\[
\star( \alpha \wedge \beta) = \pm (\star_X \alpha) \wedge (\star_Y \beta)
\]
where we postpone pinning down the sign. 
(It depends on the signature of the metrics, as well as~$j, k$.)

Now $\star \d = \star(\d_X + \d_Y)$ and it acts differently on the two summands of~$T$.
In fact, the total complex 
\[
T \xto{\star \d} B
\]
decomposes as a direct sum 
\[
(T_H \xto{\star \d_X} B_H) \oplus (T_K \xto{\star \d} B_K).
\]
We now show this and examine its consequence.

On the summand $T_H$ of~\eqref{cplx: top row}, 
$\star \d$ acts only by~$\star \d_X$ since $\d_Y$ annihilates $H^\bullet(Y)$.
The contribution of the summand $\Omega^{\leq p-\ell}_X \otimes H^\ell(Y)$ depends on how $\ell$ relates to $x = \dim(X), y = \dim(Y)$, and~$p$.
Note a basic dichotomy:
\begin{itemize}
\item[(i)] If $p - \ell \geq x$, then $\d_X$ acts by zero on the final term~$\Omega^{p-\ell}_X$.
\item[(ii)] If $p-\ell < x$, then $\d_X$ acts nontrivially. (Note this implies $x - (p-\ell) -1 \geq 0$ so its codomain $\Omega^{x-(p-\ell)-1}_X$ is nonzero too.)
\end{itemize}
In case~(i) with $\ell \leq p-x$, we find a copy of the full de Rham complex
\[
\Omega^0_X \otimes H^\ell(Y) \xto{\d_X} \cdots \xto{\d_X} \Omega^x_X \otimes H^\ell(Y)
\]
and so this subcomplex of $T_H$ is not deformed.
In case~(ii) with $\ell > p-x$,
for each value of $p-x +1 \leq \ell \leq \min(y,p)$, we get a copy of the complex
\begin{equation}
\label{cplx ofr gen mxw}
\begin{tikzcd}[row sep = tiny]
\Omega^0_X \otimes H^{\ell}(Y) \ar[r, "\d_X"] & \cdots \ar[r, "\d_X"]& \Omega_X^{p -\ell} \otimes H^{\ell}(Y) \ar[dr, "\star \d_X"]& \\
&&& \Omega_X^{x-(p-\ell)-1} \otimes H^{y-\ell}(Y) \\
\Omega_X^0 \otimes H^{y-\ell}(Y) \ar[r, "\d_X"] & \cdots \ar[r, "\d_X"]& \Omega_X^{x-(p-\ell)-2} \otimes H^{y-\ell}(Y) \ar[ur, "-\kappa \d_X"]&
\end{tikzcd}
\end{equation}
which is perturbative $p-\ell$-form Maxwell theory on~$X$, for the abelian Lie algebra~$H^{\ell}(Y)$.
(Poincar\'e duality on $Y$ guarantees $H^{y-\ell}(Y) \cong H^\ell(Y)$.)
Thus we have
\[
\bigoplus_{\ell = \max(p-x+1,0)}^{\min(y,p)} \widetilde{\Mxw}^\pert_{p-\ell, x} \otimes H^\ell(Y)
\]
{\it recovers the claim of~\eqref{eqn: rough approx},
identifying the pushforward with Maxwell theories valued in the cohomology of~$Y$.}

A similar dichotomy appears for the summand~$B_H$:
\begin{itemize}
\item[(iii)] If $x+y- p -1 -k \geq x$ (so $y-p-1 \geq k$), 
then the final term~$\Omega^{x+y- p -1 -k}_X$ is zero, 
so $\star \d$ must act trivially.
\item[(iv)] If $x+y- p -1 -k < x$ (so $y-p \leq k$), 
then $\star \d_X$ has a nontrivial codomain. 
\end{itemize}
In case~(iii) we find a copy of the full de Rham complex
\[
\Omega^0_X \otimes H^k(Y) \xto{\d_X} \cdots \xto{\d_X} \Omega^x_X \otimes H^k(Y)
\]
and so this subcomplex of $B_H$ is not deformed.
Case~(iv) provides the bottom row of~\eqref{cplx ofr gen mxw},
where the range of values for $\ell$ on the top row matches with the range of values for $k = y - \ell$ on the bottom row.

Putting these cases together, 
we find that $T_H \to B_H$ is isomorphic to the direct sum
\[
\bigoplus_{\ell = \max(p-x+1,0)}^{\min(y,p)} \widetilde{\Mxw}^\pert_{p-\ell, x} \otimes H^\ell(Y) 
\oplus \bigoplus_{\ell = 0}^{\min(p-x,y)} \Omega^\bullet_X \otimes H^\ell(Y)[p-\ell]
\oplus \bigoplus_{k = 0}^{y-p-1} \Omega^\bullet_X \otimes H^k(Y)[n-p-2-k]
\]
as claimed in Proposition~\ref{prp: pert compactification}.

We now turn to analyzing the complex $T_K \to B_K$.

Note that by Hodge theory $\d_Y$ is an isomorphism on the image of $\d^\star_Y$ inside~$\Omega^\ell(Y)$:
\[
\d_Y: K^\ell(Y) = {\rm im}(\d_Y^\star)^\ell \xto{\cong} {\rm im}(\d_Y)^{\ell+1}.
\]
Thus
\[
\star_Y \d_Y: K^\ell(Y) \xto{\cong} K^{y-(\ell+1)}(Y)
\]
is an isomorphism as well.

Hence, from the summand $T_K$ of~\eqref{cplx: top row},  
we get a two-term cochain complex
\[
\bigoplus_{\ell = 0}^p \Omega^{p - \ell}(X) \widehat{\otimes}\, K^\ell(Y) 
\xto{\star \d} 
\bigoplus_{\ell = 0}^{n-p-1} \Omega^{(n-p-1)-\ell}(X) \widehat{\otimes}\, K^\ell(Y)
\]
concentrated in degrees 0 and~1.
Since there are no ghost fields, 
{\it this complex can be viewed as a kind of scalar field theory}, albeit valued in an enormous vector space.

In this complex~$B$, both operators $\d_X$ and $\d_Y$ might act nontrivially.
On each summand, the operator $\star \d_X$ is a map
\[
\Omega^{p - \ell}(X) \widehat{\otimes}\, K^\ell(Y) \xto{\star \d_X} \Omega^{x- (p-\ell)+1}(X) \widehat{\otimes}\, K^{y-\ell}(Y)
\]
while the operator $\star \d_Y$ is a map
\[
\Omega^{p - \ell}(X) \widehat{\otimes}\, K^\ell(Y) \xto{\star \d_Y} \Omega^{x- (p-\ell)}(X) \widehat{\otimes}\, K^{y-(\ell+1)}(Y)
\]
We want to understand the cohomology of the total operator~$\star \d$.

\begin{lmm}
The complex $T_K \xto{\star \d} B_K$ has trivial cohomology (i.e., it is acyclic).
\end{lmm}

This lemma ensures that only the complex $T_H \to B_H$ contributes nontrivially to the pushforward theory.

\begin{proof}
We consider two cases separately: $p < y= \dim(Y)$ and~$y \leq p$.

Let $p < y $. 
Recall that $K^y(Y) = 0$.
Then the complex can be written as
\[
\begin{tikzcd}
\Omega^p(X) \otimes K^0(Y) \ar[r, "\star \d_Y"] 
& \Omega^{x-p}(X) \otimes K^{y-1}(Y) \\
\Omega^{p-1}(X) \otimes K^1(Y) \ar[r, "\star \d_Y"] \ar[ur, "\star \d_X"] 
& \Omega^{x-p+1}(X) \otimes K^{y-2}(Y) \\
\Omega^{p-2}(X) \otimes K^2(Y) \ar[r, "\star \d_Y"]
 \ar[ur, "\star \d_X"] & \vdots \\
\vdots \ar[r, "\star \d_Y"] & \Omega^{x-1}(X) \otimes K^{y-p}(Y)\\
\Omega^{0}(X) \otimes K^p(Y) \ar[r, "\star \d_Y"] \ar[ur, "\star \d_X"] 
& \Omega^{x}(X) \otimes K^{y-(p+1)}(Y)
\end{tikzcd}
\]
where the left hand column is $T_K$ and the right hand is~$B_K$.
(Note that $y \geq p+1$ so the bottom right term is nonzero.)
As $\star \d_Y$ is an isomorphism, 
this complex is acyclic.

Now let $y \leq p$. 
Then the complex can be written as
\[
\begin{tikzcd}
\Omega^p(X) \otimes K^0(Y) \ar[r, "\star \d_Y"] 
& \Omega^{x-p}(X) \otimes K^{y-1}(Y) \\
\Omega^{p-1}(X) \otimes K^1(Y) \ar[r, "\star \d_Y"] \ar[ur, "\star \d_X"] 
& \Omega^{x-p+1}(X) \otimes K^{y-2}(Y) \\
\Omega^{p-2}(X) \otimes K^2(Y)
\vdots \ar[ur, "\star \d_X"] \ar[r, "\star \d_Y"] & \vdots \\
\vdots \ar[r, "\star \d_Y"] & \Omega^{x-(p-y-2)}(X) \otimes K^{1}(Y)\\
\Omega^{0}(X) \otimes K^y(Y) \ar[r, "\star \d_Y"] \ar[ur, "\star \d_X"] 
& \Omega^{x-(p-y+1)}(X) \otimes K^{0}(Y)
\end{tikzcd}
\]
since $x-(p-y+1) = x+y -p-1 = n-p-1$.

As $\star \d_Y$ is an isomorphism, 
this complex is also acyclic.
\end{proof}

\appendix

\section{Moduli stacks and derived geometry in more detail}
\label{sec: dag stuff}

In this section we describe the mathematical home for the constructions used in this paper,
namely sheaves on a site of spacetimes valued in derived stacks.
The spacetimes are simple enough, such as $\Riem_n$, the site of Riemannian $n$-manifolds and isometric embeddings; 
but we need to clarify what we mean by {\em derived stack}.
Our needs for this paper are minimal,
since we only obtain them from sheaves built with differential forms $\Omega^k$ or with a constant sheaf~$\ZZ$.
In consequence, we do not examine geometry here closely,
although undoubtedly there are interesting phenomena to explore.

We note that these objects, as sheaves on the site of smooth manifolds, are already well-explored.\footnote{
Freed and Hopkins~\cite{FreHopCW} offers a lovely exposition for a broad audience; it is
a useful companion to their work on differential cohomology for physics.
The book \cite{AmaDebHai} is a very accessible and more extensive treatment of differential cohomology and sheaves on manifolds.
Finally, \cite{MooSax} is a physics-oriented exposition.}
For instance, the $k$-forms provide a functor
\[
\Omega^k: \Mfld^\opp \to \Set
\]
assigning $\Omega^k(M)$ to each smooth manifold~$M$,
and (truncated) de Rham and Deligne complexes likewise provide functors like
\[
\Omega^{\leq k}[k]: \Mfld^\opp \to \Ch
\]
that can then be further mapped to $\Spaces$.
Hence the field content (in the BRST sense) of abelian gauge theories is naturally accommodated by this setting.

This setting, however, has some limitations:
\begin{itemize}
\item it is not easy to extract the perturbative features and
\item it is not easy to encode the equations of motion (and antifield content of the BV formalism).
\end{itemize}
A modest enhancement of $\Mfld$, however, can obviate these issues,
and that is by using derived manifolds instead of manifolds as the test objects.

Our discussion here will thus emphasize the {\em derived} direction to complement the {\em higher} direction already developed in the community living between physics and homotopy theory.

Our approach is to view a derived stack as a moduli space,
so we define it as functor.
As is typical currently,
such a functor takes values in $\Spaces$,
the $\infty$-category of $\infty$-groupoids (or homotopy types).
The domain of such a functor is a site,
where we view objects in the site as parameter spaces.
For our purposes, it is useful for this site to contain both formal directions and derived directions;
we also wish to incorporate differential geometry, 
since we work in field theory.
Thus we will sketch the rudiments of derived differential geometry, following \cite{CarSte, CarDG, SteDDG, Tar1, Tar2, AlfYou}.
Indeed, our overview is merely a riff on that from~\cite{Tar2}.

First, we need to build up to the site,
and we do that by defining {\em derived manifolds} as enhancements of ordinary manifolds.
Then we will characterize derived stacks among the $\Spaces$-valued functors on this site.
It will take a few definitions.

\subsection{Derived rings for smooth geometry}

We start with identifying the natural class of algebras,
guided by algebraic geometry, 
where the model spaces are the affine schemes and those are formally dual to commutative rings.

\begin{dfn}
Let $\CartSp$ denote the category whose objects are the smooth manifolds $\RR^n$ for each natural number $n \in \NN$ and whose morphisms are smooth maps between them.
\end{dfn}

It is the full subcategory of the category of smooth manifolds $\Mfld$ given by the {\em Cartesian} spaces~$\RR^n$.
Every smooth manifold $M$ determines a functor
\[
C^\infty_M: \CartSp \to \Set
\]
where
\[
C^\infty_M(\RR) = C^\infty(M) = \Map_{C^\infty}(M,\RR)
\]
and
\begin{align*}
C^\infty_M(\RR^n) &= \Map_{C^\infty}(M,\RR^n)\\
&= C^\infty(M)^{\times n} = C^\infty(M) \times \cdots \times C^\infty(M) 
\end{align*}
where we take $n$ copies on the right.
Any smooth map $\phi: \RR^m \to \RR^n$ induces a map
\[
C^\infty_M(\phi): \Map_{C^\infty}(M,\RR^m) \to \Map_{C^\infty}(M,\RR^n)
\]
by postcomposition.

For example, the multiplication map for the algebra of real numbers $m_\RR: \RR^2 \to \RR$ induces the multiplication map on functions
\[
C^\infty(M) \times C^\infty(M) = C^\infty_M(\RR^2) \xto{m_\RR \circ -} C^\infty_M(\RR) = C^\infty(M)
\]
by postcomposition.
In other words, this functor $C^\infty_M$ encodes the algebra of {\em smooth} functions on~$M$.

Note that by construction this functor $C^\infty_M$ preserves finite products.
With a little work, one sees that there is a fully faithful functor
\[
\Mfld \hookrightarrow \Fun^{\rm prod}(\CartSp, \Set)
\]
assigning a manifold $M$ its functor~$C^\infty_M$;
the subscript denotes functors that preserve finite products.

This larger category of functors encodes rings that behave like smooth functions on manifolds.

\begin{dfn}
A {\em $C^\infty$-ring} is a functor 
\[
R: \CartSp \to \Set
\] 
that preserves finite products.

Let $\Ring_{C^\infty}$ denote the category~$\Fun^{\rm prod}(\CartSp, \Set)$.
\end{dfn}

We have seen how every smooth manifold gives such a ring, 
but it is easy to produce new examples as quotients.

\begin{constr}
Let $f_1, \ldots, f_k$ be a collection of smooth functions on a manifold~$M$.
Let $I = (f_1, \ldots, f_k) \subset C^\infty(M)$ be the ideal they generate.
Then the quotient ring $C^\infty(M)/I$ determines a $C^\infty$-ring.

More explicitly, these $f_j$ determine a map of smooth manifolds $F: M \to \RR^k$ and hence a map of $C^\infty$-rings 
\[
F^*: C^\infty_{\RR^k} \to C^\infty_M
\]
and we define the ``quotient $C^\infty$-ring'' $R$ to be the pushout, as product-preserving functors,
\[
\begin{tikzcd}
C^\infty_{\RR^k} \ar[d,"F^*"] \ar[r] & C^\infty_{*} \ar[d] \\
C^\infty_M \ar[r] & R
\end{tikzcd}
\]
where $* = \RR^0$ is a point so $C^\infty_{*}$ is simply the ring~$\RR$ of real numbers.~\hfill$\Diamond$
\end{constr}

Artinian algebras, and hence deformation theory, naturally arise this way.

\begin{eg}
The dual numbers\footnote{This algebra is useful for computing tangent spaces, in the functorial style.} $\RR[x]/(x^2)$ appear as special case of the construction above.  
Take $M =\RR$ with coordinate $x$ and the function $f(x) = x^2$,
so that  the $C^\infty$-ring associated to $C^\infty(\RR)/(x^2)$ has $\RR[x]/(x^2)$ as its underlying commutative ring.~\hfill$\Diamond$
\end{eg}

Such a functor $R$ automatically determines a commutative ring but this ring admits more operations.
If one views a $C^\infty$-ring $R$ as functions on some space (a generalized manifold),
then an element $r \in R(\RR)$ should be viewed as encoding a map from this space to the real line~$\RR$,
so one should be able to compose with smooth functions on the real line.
For example, there is some composite $\exp \circ \,r$ with $\exp$ the exponential map.
For a typical ring, by contrast, there is no way to talk about~$\exp\circ\, r$.

\begin{rmk}
To make this idea precise, note that an element $r \in R(\RR)$ corresponds to a natural transformation $t_r: C^\infty_{\RR} \to R$ of $C^\infty$-rings.\footnote{If $R = C^\infty_M$, then $r \in C^\infty_M(\RR) = C^\infty(\RR)$ corresponds to a map of manifolds $r: M \to \RR$ and this natural transformation is the pullback of functions $r^*: C^\infty(\RR) \to C^\infty(M)$.}
In particular, a smooth function $f: \RR \to \RR$ determines a map $t_f: C^\infty_{\RR} \to C^\infty_{\RR}$.
Thus, the natural transformation $t_f \circ t_r$ determines an element $f \circ r \in R(\RR)$.~\hfill$\Diamond$
\end{rmk}

It is now quick to provide a derived generalization.

\begin{dfn}
A {\em derived $C^\infty$-ring} is a functor of $\infty$-categories
\[
\cR: \CartSp \to \Spaces
\]
that preserves finite products.

Let $\sRing_{C^\infty}$ denote the $\infty$-category~$\Fun^{\rm prod}(\CartSp, \Spaces)$.
\end{dfn}

A concrete source of examples is to give a product-preserving functor into a category that naturally maps to $\Spaces$ (and preserves products).
For instance, 
consider such a functor $\CartSp \to \Ch_{\leq 0}$ and then postcompose with the Dold-Kan construction and localize:
\[
\CartSp \to \Ch_{\leq 0} \to \sAb \to \Spaces.
\]
The first functor is, in essence, a (connective) differential graded $C^\infty$-ring and the composite is its associated derived $C^\infty$-ring.
We will freely use such dg $C^\infty$-rings as avatars of derived $C^\infty$-rings.\footnote{Indeed, such dg $C^\infty$-rings present all derived $C^\infty$-rings. 
For an explanation, see Remark~2.2.11 of \cite{Nuiten}, which is the culmination of Chapter 2's systematic development of this framework.}

\begin{eg}
The {\em shifted} dual numbers $\RR[y]/(y^2)$, with generator $y$ having degree~$|y| < 0$,
determine an example of a derived $C^\infty$-ring.
Note that this algebra, modulo $y$, is simply the real numbers~$\RR$,
and we have seen that $\RR$ is encoded by $C^\infty_*$ as a $C^\infty$-ring.
Hence we simply define the shifted dual numbers as the functor $C^\infty_*[y]/(y^2)$ as a functor into~$\Ch$.~\hfill$\Diamond$
\end{eg}

\subsection{Derived manifolds}

Every smooth manifold determines a derived $C^\infty$-ring by its ring of functions $C^\infty(M)$,
and so this construction determines a functor
\[
\cO: \Mfld \to \sRing_{C^\infty}^\opp.
\]
We would like to identify a nice subcategory of $\sRing_{C^\infty}^\opp$ that contains $\Mfld$ but is better behaved.\footnote{The category $\Mfld$ does not contain many limits. For example, two submanifolds $S_1, S_2$ of a manifold $M$ might not have a fiber product (or pullback) in $\Mfld$: this fiber product does not exist if their intersection $S_1 \cap S_2$ is not a manifold.}

It is helpful to revisit how ordinary smooth manifolds are constructed as a road map.
Atlases let us assemble simple local pieces together;
this patching process is a type of colimit.
These local pieces are often described as subsets of a Cartesian space,
as the solutions to some system of equations;
such a subset is the zero set of a smooth map $\RR^N \to \RR^k$,
so a fiber product, a type of limit.
But we need  the solution set to be smooth everywhere,
and this condition of smoothness is part of what makes $\Mfld$ a finicky category to work in:
the zero set of a smooth function $\RR^N \to \RR^k$ is not always a smooth  manifold,
so some limit diagrams do not admit limits.
Inspired by this observation, we define derived manifolds by patching together the fiber products of {\em all} smooth functions on Cartesian spaces.

\begin{dfn}
Let $\DerMfld$ denote the full subcategory of $\sRing_{C^\infty}^{\opp}$ containing all pullbacks of diagrams\footnote{These diagrams are dual to pushout diagrams $C^\infty(\RR^N) \leftarrow C^\infty(\RR^k) \to C^\infty(*) = \RR$ in $\Ring_{C^\infty}$.}
\[
* \to \RR^k \leftarrow \RR^N
\] 
and closed under finite limits and retracts.
\end{dfn}

In other words, the derived manifolds include intersections of nontransverse submanifolds,
as well as the usual well-liked intersections.

\begin{eg}
As a simple but illuminating example,
consider the pullback of two copies of the origin in the real line: 
\begin{equation}
\label{eqn: double point}
* \xto{\rm zero} \RR \xleftarrow{\rm zero} *
\end{equation}
i.e., the self-intersection of the origin.
The inclusion of the origin $* \to \RR$ corresponds to the ring map $C^\infty(\RR) \to \RR$ sending a function $f$ to $f(0)$, 
its value at the origin.
In other words, functions at the origin are the quotient $C^\infty$-ring $C^\infty(\RR)/(x)$ where $x$ denotes the coordinate function on the real line.
It can also be presented by the dg $C^\infty$-ring $C^\infty(\RR)[\xi]$ with differential $x \partial_\xi$ and where $|\xi| = -1$;
this dg ring is its Koszul resolution in the $C^\infty$ setting.
Thus the functions on pullback~\eqref{eqn: double point} are encoded by the dg $C^\infty$-ring $C^\infty(\RR)[\xi_1, \xi_2]$ with differential $x \partial_{\xi_1} + x \partial_{\xi_2}$.
This dg $C^\infty$-ring is quasi-isomorphic to the shifted dual numbers $\RR[y]/(y^2)$ where $|y| = -1$.~\hfill$\Diamond$
\end{eg}

It is useful to know that there are many ways to model derived manifolds in concrete terms,
such as with supermanifolds equipped with cohomological vector fields or dg manifolds.
There are powerful results that let one move between different concrete models and that show how to leverage prior work.
In practice these methods produce objects in ordinary categories that admit a functor to the $\infty$-category~$\DerMfld$,
much as a cochain complex presents an object of the derived category.
In particular, there is a nice category of {\em dg manifolds} $\dgMfld$ whose $\infty$-categorical localization 
\[
\dgMfld \to \DerMfld
\]
presents derived manifolds, much as $\Chz$ presents $\cDz$~\cite{CarDG, Tar1}.

This $\infty$-category $\DerMfld$ satisfies a beautiful universal property, as shown by~\cite{CarSte},
showing that it is a well-behaved extension of~$\Mfld$.

For an $\infty$-category $\cC$ with finite limits,
let $\Fun^{\rm lex}(\DerMfld,\cC)$ denote the $\infty$-category of left-exact functors (i.e., those preserving finite limits).
Let $\Fun^{\pitchfork}(\Mfld,\cC)$ denote the $\infty$-category of functors (of $\infty$-categories) preserving transverse pullbacks in the category of smooth manifolds~$\Mfld$. 

\begin{thm}[\cite{CarSte}]
Let $\cC$ be an $\infty$-category that contains all finite limits and that is idempotent-complete.\footnote{This second condition says, roughly, that every idempotent endomorphism is a retract. That may seem a little obscure, so it might be helpful to observe that every manifold is the retract of an open subset of a Euclidean space. 
(Concretely, embed the manifold in a Euclidean space and then identify its normal bundle with the tubular neighborhood.  
But every  manifold is the retract of a vector bundle over it,
via inclusion as the zero section and projection to the base.)
Hence $\Mfld$ arises as the idempotent completion of its own subcategory of open sets in Euclidean spaces.
} 
There is an equivalence of $\infty$-categories 
\[
\Fun^{\rm lex}(\DerMfld,\cC) \xto{\simeq} \Fun^{\pitchfork}(\Mfld,\cC)
\]
by precomposition with the inclusion functor $\iota: \Mfld \to \DerMfld$.
\end{thm}

In this sense, derived manifolds arise by adding all finite limits that do not already exist among manifolds.
Following Spivak, one might say they are the spaces that are ``good for doing intersection theory.''

\subsection{Derived stacks}

The derived manifolds will provide the test spaces for our derived stacks, viewed as functors.
In other words, we will work with families parametrized by derived manifolds.

That means we need to equip $\DerMfld$ with the structure of a site,
and we do this in a simple-minded way,
borrowing from usual topology as follows.

First, every derived manifold $\cM$ has an underlying topological space~$|\cM|$.
In concrete terms, if $\cM$ is constructed as a finite limit in~$\sRing_{C^\infty}^\opp$,
we let $|\cM|$ denote the corresponding finite limit in~$\Top$.
Thus, in practice, $|\cM|$ is some closed subset of a Cartesian space~$\RR^N$,
and $\cM$ is that topological space~$|\cM|$ equipped with an interesting dg commutative algebra.

\begin{dfn}
\label{dfn: der mfld covers}
Given a derived manifold $\cM$, a collection of maps of derived manifolds 
\[
\{\cU_\alpha \to \cM\}_{\alpha \in A}
\]
is a {\em cover} if the associated continuous maps
\[
\{|\cU_\alpha| \to |\cM|\}_{\alpha \in A}
\]
are open embeddings, the union of their images is all of~$|\cM|$ (i.e., they are jointly surjective), and each map of $C^\infty$-rings $\cO(\cM) \to \cO(\cU_\alpha)$ is the associated $C^\infty$-localization.
\end{dfn}

This notion generates a Grothendieck topology on~$\DerMfld$.

\begin{dfn}
A {\em derived stack} is a functor $\XX: \DerMfld^\opp \to \Spaces$ that is a sheaf with respect to the Grothendieck topology from Definition~\ref{dfn: der mfld covers}.
\end{dfn}

In practice one can obtain a derived stack by describing a functor (i.e., prestack) and then sheafifying.

\begin{eg}
\label{eg: maps into manifold}
For any derived manifold $\cX$ --- and hence any smooth manifold~$X$ --- we have
\[
\cMap(-,\cX): \DerMfld^{\opp} \to \Spaces
\]
where 
\begin{align*}
\cMap(\cM,\cX) 
&= \Map_{\DerMfld}(\cM,\cX)\\
&= \Map_{\sRing_{C^\infty}}(\cO(\cX), \cO(\cM))
\end{align*}
with the right hand sides denoting the space of maps in the $\infty$-categories.~\hfill$\Diamond$
\end{eg}

Here we revisit another key example.

\begin{eg}
\label{eg: k-forms as derived stack}
Given a smooth manifold~$M$, 
consider the functor
\[
\BOmega^k_M: \DerMfld^\opp \to \Spaces
\]
where
\[
\BOmega^k_M(\cM) = 
\{ \text{sections of the bundle } \pi^*( \Lambda^k T^*M) \to M \times \cM \}
\]
where $\pi: M \times \cM \to M$ is the projection map.
This functor encodes $k$-forms on $M$ in families over the derived manifold~$\cM$.
It determines a derived stack, which we identify by the same symbol.~\hfill$\Diamond$
\end{eg}

\begin{eg}
\label{eg: closed forms as derived stack}
Building on the preceding example, there is a natural operator 
\[
\d: \BOmega^k_{M} \to \BOmega^{k+1}_{M}
\]
arising from the usual exterior derivative in the $M$-direction.
Hence one can formulate derived versions of the key constructions from Section~\ref{sec: deligne stuff}. 
For instance, taking the closed $p+1$-forms of Definition~\ref{dfn: closed forms} yields $\BOmega^{p+1}_{\cl,M}$.~\hfill$\Diamond$
\end{eg}

\begin{eg}
Let $\sBun_{p,M}^\nabla$ denote the derived stack of $U(1)$ $p$-bundles with connection on the smooth manifold~$M$.
Its value on a derived manifold $\cM$ consists of $\cM$-families of such bundles on~$M$.~\hfill$\Diamond$
\end{eg}

\begin{rmk}
For a systematic and extensive development of moduli of bundles (and bundles with connection) as derived stacks,
see~\cite{BMNS}.
Their work should make it possible to develop nonabelian gauge theories in a manner similar to that pursued here for abelian gauge theories.~\hfill$\Diamond$
\end{rmk}

One more ingredient from geometry will be helpful to have at hand.

A derived stack $\XX$ has a ring of functions or, more precisely, a structure sheaf~$\cO_\XX$.
Since we approach geometry in the functor of points style,
the phrasing runs as follows.
Recall that we have a ``forgetful'' functor
\[
\cO: \DerMfld \to \sRing_{C^\infty}^\opp.
\]
so that each derived manifold $\cM$ has a natural $C^\infty$-ring of functions:
just view $\cM$ as an object $\cO(\cM)$ in the bigger, ambient $\infty$-category~$\sRing_{C^\infty}$.
We want to leverage that structure to get $C^\infty$-rings out of~$\XX$.
Observe that we have the Yoneda embedding
\[
Y: \DerMfld \hookrightarrow \Fun(\DerMfld^\opp, \Spaces)
\]
and so we have an $\infty$-category $\DerMfld_{/\XX}$ as a subcategory of the overcategory $\Fun(\DerMfld^\opp, \Spaces)_{/\XX}$.
It encodes the information of~$\XX$ as a fibration
\[
\DerMfld_{/\XX} \to \DerMfld
\]
where the fiber over $\cM$ encodes the space of $\cM$-points of~$\XX$.
Thus we have a composite 
\[
\DerMfld_{/\XX} \to \DerMfld \xto{\cO} \sRing_{C^\infty}^\opp
\]
that encodes the $C^\infty$-ring of functions on every ``point'' of~$\XX$.
We call this functor the {\em structure sheaf} $\cO_\XX$ of the derived stack.

\begin{rmk}[Example~\ref{eg: k-forms as derived stack} revisited]
Roughly speaking, $\Omega^k_M(\cM)$ is $\Omega^k(M) \otimes \cO(\cM)$,
since it is a smooth function on $\cM$ valued in $k$-forms on~$M$.
This tensor product is inadequate even when $\cM$ is an ordinary smooth manifold;
one must complete the algebraic tensor product (e.g., with the completed projective tensor product).\footnote{We are blithely ignoring how to define this tensor product, just assuming something similar to the underived construction is possible. The real point is properly articulating the moduli problem.}
Instead we directly formulate the moduli problem as sections of the pullback bundle, 
which encodes the answer we really want.~\hfill$\Diamond$
\end{rmk}

\begin{rmk}[Example~\ref{eg: closed forms as derived stack} revisited]
\label{rmk: dgmfld version}
In the setting of dg manifolds, it is straightforward to define functors like
\[
\Omega^{p+1}_{\cl,M}: \dgMfld^\opp \to \Chz
\]
that presents the derived stack~$\BOmega^{p+1}_{\cl,M}$.
Roughly,
\[
\Omega^{p+1}_{\cl,M}(\cM) \simeq \Omega^{p+1}(M) \otimes_\RR \cO(\cM)
\]
but with the same tensor product issues discussed in the preceding remark.~\hfill$\Diamond$\end{rmk}

\subsection{Our field theories as derived stacks}

Recall that in our situation, 
we want the fields of our field theory to define a functor
\[
\cF: \Open(M)^\opp \to \dStk
\]
where $\Open(M)$ is the site of open subsets of a smooth manifold~$M$ (the ``spacetime'') or
a functor 
\[
\cF: \Riem_n^\opp \to \dStk
\]
where $\Riem_n$ is the site of $n$-dimensional Riemannian manifolds and isometries when working with a ``Euclidean'' field theory. 
(There is a site $\Lrtz_n$ for Lorentzian theories.)
We will use $\mathsf{ST}$ to denote a site of ``spacetimes'' for the rest of the section, 
which depends upon the kind of theory being studied;
readers are welcome to substitute $\Open(M)$ or $\Riem_n$ or whichever site they prefer.

In summary our goal is to describe a functor
\[
\cF: \mathsf{ST}^\opp \times \DerMfld^\opp \to \Spaces
\]
that encodes our fields.

In our examples, which are built from (truncated) Deligne complexes, there are two kinds of ingredients.

First, there are differential $k$-forms on a spacetime manifold:
\[
\BOmega^k_{\mathsf{ST}}: \mathsf{ST}^\opp \times \DerMfld^\opp \to \Spaces
\]
where
\[
\BOmega^k_{\mathsf{ST}} (M, \cM) = \BOmega^k_M(\cM)
\]
where the space $\BOmega^k_M(\cM)$ from Example~\ref{eg: k-forms as derived stack} encodes smooth families over the derived manifold $\cM$ of $k$-forms on the smooth manifold~$M$.
There is a natural operator 
\[
\d: \BOmega^k_{\mathsf{ST}} \to \BOmega^{k+1}_{\mathsf{ST}}
\]
arising from the usual exterior derivative in the spacetime (or $M$) direction.

Second, there is the mapping space into $S^1$ on a spacetime manifold:
\[
\cMap_{\mathsf{ST}}(-,S^1): \mathsf{ST}^\opp \times \DerMfld^\opp \to \Spaces
\] 
where
\[
\cMap_{\mathsf{ST}}((M, \cM),S^1) = \cMap(M \times \cM,S^1)
\]
where the space $\cMap(M \times \cM,S^1)$ from Example~\ref{eg: maps into manifold} encodes smooth families over the derived manifold $\cM$ of maps from the smooth manifold~$M$ into the circle.\footnote{Roughly,
\[
\cMap(M \times \cM, S^1) \simeq \Map(M,S^1) \otimes_\RR \cO(\cM)
\]
since $S^1 \cong \RR/\ZZ$ is an $\RR$-module.}

Both ingredients can be equipped with abelian group structures,
as we can present these derived stacks using functors into~$\Ch$, as mentioned in Remark~\ref{rmk: dgmfld version}.

We can assemble these ingredients into the field theories discussed in this paper.
For instance, to describe $p$-form Maxwell theory, 
we start with the functor
\[
\mathsf{ST}^\opp \times \dgMfld^\opp \to \Ch
\]
assigning to $((M,g),\cM)$ the cochain complex
\[
\Map(M \times \cM, S^1) \xto{\d \log} \Omega^1_M(\cM) \xto{\d} \cdots \to \Omega^p_M(\cM) \xto{\d \star \d} \Omega^{n-p}_M(\cM) \xto{\d} \cdots \xto{\d} \Omega^n_M(\cM)
\]
where $\Omega^p_M(\cM)$ sits in degree zero.
This functor determines a functor into $\cDz$, by localizing at quasi-isomorphisms, 
and hence it presents a functor
\[
\mathsf{ST}^\opp \times \DerMfld^\opp \to \Ab(\Spaces)
\]
by truncation.

\section{Homological perturbation lemma}
\label{sec: HPL}

We follow~\cite{crainic}, with mild modifications of notation.

\begin{dfn}
A {\em homotopy equivalence} is the data
\[
\begin{tikzcd}
(S,\d_S) \ar[r,shift left, "\iota"] & (B, \d_B) \ar[l, shift left, "\pi"] \ar[loop right, "\eta"] 
\end{tikzcd}
\]
where
\begin{itemize}
\item $(S,\d_S)$ and $(B,\d_B)$ are cochain complexes (``small'' and ``big,'' respectively),
\item $\iota$ and $\pi$ are quasi-isomorphisms (``inclusion'' and ``projection,'' respectively), and
\item $\eta$ is a degree $-1$ linear endomorphism of $B$ such that
\[
\iota \circ \pi = 1_B + \d_B \circ \eta + \eta \circ \d_B, 
\]
i.e., $\eta$ is a cochain homotopy equivalence between the identity on $B$ and~$\iota \circ \pi$.
\end{itemize}
\end{dfn}

The Hodge theorem for elliptic complexes
is a key source of examples for us.

\begin{dfn}
A {\em perturbation} of $(B,\d_B)$ is a degree one endomorphism $\delta$ of the graded vector space $B$ such that $\d_B + \delta$ is square-zero.
\end{dfn}

We view $(B, \d_B + \delta)$ as a deformation of $B$ and we would like to know if we can deform $S$ compatibly.

\begin{dfn}
Given a homotopy equivalence,
a degree~0 endomorphism $\delta$ of $B$ is {\em small} if the endomorphism $1_B - \delta \circ \eta$ is invertible.
\end{dfn}

Small perturbations often arise when $\delta$ is nilpotent or weighted by a formal parameter like~$\hbar$.

For concision, we will cease to write $\circ$ to indicate composition.

\begin{prp}[Homological Perturbation Lemma, or HPL]
If $\delta$ is a small perturbation of $\d_B$,
then there is a homotopy equivalence
\[
\begin{tikzcd}
(S,\d'_S) \ar[r,shift left, "\iota'"] & (B, \d_B + \delta) \ar[l, shift left, "\pi'"] \ar[loop right, "\eta'"] 
\end{tikzcd}
\]
where
\[
A = (1_B - \delta \eta)^{-1} \delta
\]
and
\[
\d'_S = \d_S+ \pi A \iota,\quad \iota' = \iota + \eta A \iota,\quad \pi' = \pi + \pi A \eta, \quad \eta' = \eta + \eta A \eta.
\]
\end{prp}

In words, a small perturbation of $B$ determines a perturbation of the whole homotopy equivalence.

As a special case, we revisit a classic situation (the staircase lemma and the spectral sequence for a nice double complex) using the HPL.

Consider a double complex concentrated in the upper right quadrant
\[
\begin{tikzcd}
\vdots &\vdots&\vdots&\\
C^{0,2} \ar[r, "h"] \ar[u, "v"]& C^{1,2} \ar[r, "h"] \ar[u, "v"]& C^{2,2} \ar[r, "h"] \ar[u, "v"]& \cdots\\
C^{0,1} \ar[r, "h"] \ar[u, "v"]& C^{1,1} \ar[r, "h"] \ar[u, "v"]& C^{2,1} \ar[r, "h"] \ar[u, "v"]& \cdots\\
C^{0,0} \ar[r, "h"] \ar[u, "v"] & C^{1,0} \ar[r, "h"] \ar[u, "v"]& C^{2,0} \ar[r,"h"] \ar[u, "v"]& \cdots
\end{tikzcd}
\]
In this situation, for each column $(C^{k, \bullet}, v)$, 
we can produce a deformation complex onto its vertical cohomology.
Putting this data together among all the rows, 
we have a homotopy equivalence
\begin{equation}
\label{horizontal def retract}
\begin{tikzcd}
(H^\bullet(C,v), 0) \ar[r,shift left, "\iota"] & (C, v) \ar[l, shift left, "\pi"] \ar[loop right, "\eta"] 
\end{tikzcd}
\end{equation}
where we can view this retract as holding on the totalization of the double complexes.

We view  the horizontal differential $h$ as the perturbation (i.e., as $\delta$).
Observe that $h \eta$ sends $C^{j,k}$ to $C^{j+1,k-1}$
--- so to speak, ``stepping down the staircase'' ---
and hence is nilpotent because eventually it passes beyond the bottom row.
(The order of nilpotency depends on the value of $j+k$.)
Thus $1-h \eta$ is invertible, 
and its inverse is given by the geometric series
\[
(1-h \eta)^{-1} = 1 + h \eta + (h\eta)^2 + \cdots + (h\eta)^n + \cdots.
\]
This operator takes an element and ``descends the staircase,''
leaving a trace of itself on each step.
Let 
\begin{equation}
\label{double complex def}
\alpha = (1-h \eta)^{-1} h
\end{equation}
play the role of $A$ from the HPL.

\begin{crl}
\label{crl: hpl for double}
Let $(C^{\bullet \bullet}, h + v)$ be a first quadrant double complex equipped with a deformation retraction 
along the rows as in~\eqref{horizontal def retract}.
As $1- h \eta$ is invertible (i.e., $h$ is small), 
set $\alpha = (1-h \eta)^{-1} h$.
Then there is a homotopy equivalence
\[
\begin{tikzcd}
(H^\bullet(C,v), \pi \alpha \iota) \ar[r,shift left, "\iota + \eta \alpha \iota"] & (C, h+v) \ar[l, shift left, "\pi + \pi \alpha \eta"] \ar[loop right, "\eta + \eta \alpha \eta"]
\end{tikzcd}
\]
between the full double complex and the cohomology for $v$ equipped with a ``transferred'' differential.
\end{crl}

Note that there is a similar statement if we use the horizontal cohomology first.

\printbibliography

\end{document}